\documentclass[10pt,conference]{llncs2e/llncs}
\usepackage{etex}

\usepackage[textfont={it,scriptsize}]{subfig}

\usepackage{cprotect}
\usepackage{booktabs}

\usepackage{amsthm}
\usepackage{amsmath,amssymb,mathtools,stmaryrd}
\usepackage{graphicx}
\usepackage{complexity}
\usepackage{tikz}
\usepackage{xspace}
\usepackage{multicol}
\usepackage{wrapfig}
\usepackage{enumitem}
\usepackage[bottom]{footmisc}
\usepackage{footnote}
\makesavenoteenv{table}
\makesavenoteenv{tabular}
\usepackage[font={scriptsize,it}]{caption}
\usepackage{adjustbox}

\usepackage{cancel}
\usepackage[normalem]{ulem}

\renewcommand\paragraph[1]{\vspace{.0ex}\par\noindent\textbf{#1}.\;\;}

\let\overrightarroworig\overrightarrow
\renewcommand{\overrightarrow}[1]{%
  \overrightarroworig{%
    \rule{0pt}{\heightof{$#1$}+1.5pt}%
    \hspace{-1mm}#1\vspace{-2.5pt}%
    }}

\let\overleftarroworig\overleftarrow
\renewcommand{\overleftarrow}[1]{%
  \overleftarroworig{%
    \rule{0pt}{\heightof{$#1$}+1.5pt}%
    \hspace{-1mm}#1\vspace{-2.5pt}%
    }}

\usepackage{color}
\definecolor{myblue}{rgb}{.8, .8, 1}
\usepackage{empheq}

\newlength\mytemplen
\newsavebox\mytempbox

\makeatletter
\newcommand\mybluebox{%
    \@ifnextchar[
       {\@mybluebox}%
       {\@mybluebox[0pt]}}

\def\@mybluebox[#1]{%
    \@ifnextchar[
       {\@@mybluebox[#1]}%
       {\@@mybluebox[#1][0pt]}}

\def\@@mybluebox[#1][#2]#3{
    \sbox\mytempbox{#3}%
    \mytemplen\ht\mytempbox
    \advance\mytemplen #1\relax
    \ht\mytempbox\mytemplen
    \mytemplen\dp\mytempbox
    \advance\mytemplen #2\relax
    \dp\mytempbox\mytemplen
    \colorbox{myblue}{\hspace{1em}\usebox{\mytempbox}\hspace{1em}}}
\makeatother

\definecolor{shadecolor}{cmyk}{.05,.05,0.05,0.05}
\definecolor{light-blue}{cmyk}{0.15,.15,.15,.15}
\newsavebox{\mysaveboxM} 
\newsavebox{\mysaveboxT} 

\newcommand*\Garybox[2][Example]{%
  \sbox{\mysaveboxM}{#2}%
  \sbox{\mysaveboxT}{\fcolorbox{black}{light-blue}{#1}}%
  \sbox{\mysaveboxM}{%
    \parbox[b][\ht\mysaveboxM+.5\ht\mysaveboxT+.5\dp\mysaveboxT][b]{\wd\mysaveboxM}{#2}%
  }%
  \sbox{\mysaveboxM}{%
    \fcolorbox{black}{shadecolor}{%
      \makebox[10em]{\usebox{\mysaveboxM}}%
    }%
  }%
  \usebox{\mysaveboxM}%
  \makebox[0pt][r]{%
    \makebox[\wd\mysaveboxM][c]{%
      \raisebox{\ht\mysaveboxM-0.5\ht\mysaveboxT+0.5\dp\mysaveboxT-0.5\fboxrule}{\usebox{\mysaveboxT}}%
    }%
  }%
}

\usetikzlibrary{automata}
\usetikzlibrary{positioning}
\usetikzlibrary{arrows,patterns}
\usetikzlibrary{calc}
\usetikzlibrary{decorations.pathreplacing,decorations.markings}

\usepackage{ifpdf}
    \ifpdf
\usepackage[colorlinks,citecolor=blue!60!black!70,linkcolor=blue!60!black!70]{hyperref}
    \else
    \fi

\let\phi=\varphi 
\let\epsilon=\varepsilon

\newcommand\concept[1]{\textit{#1}}

\newcommand\defmath[2]{\newcommand#1{\ensuremath{#2}\xspace}}

\renewcommand\implies{\Rightarrow}

\defmath\comm{\bowtie}
\defmath\lcomm{\stackrel{\leftarrow}{\comm}}
\defmath\rcomm{\stackrel{\rightarrow}{\comm}}

\defmath\alias{\mathsf{Alias}}
\defmath\strip{\mathsf{strip}}
\defmath\wait{\mathsf{wait}}
\defmath\register{\mathsf{register}}

\defmath\locks{\mathsf{Mutex}}
\defmath\lock{\mathsf{lock}}
\defmath\unlock{\mathsf{unlock}}
\defmath\locked{\mathsf{locked}}

\newcommand\Ln{L'} 
\newcommand\Rn{R'} 

\defmath\LFS{\mathit{LFS}}

\defmath\TS{{TS}}
\defmath\TAS{{TS}}
\defmath\PAS{{PT}}
\defmath\bisim{\sim}
\defmath\tatrans{\rightarrow}

\defmath\trtrans{\hookrightarrow}
\defmath\brtrans{\leadsto}
\defmath\wrtrans{\rightsquigarrow}

\defmath\RR{\mathcal{R}}
\defmath\LL{\mathcal{L}}
\defmath\NN{\mathcal{N}}
\defmath\WW{\mathcal{E}}
\defmath\EX{\mathcal{F}}
\defmath\CC{\mathcal{C}}
\defmath\II{\mathcal{I}}
\defmath\nRR{\overline{\RR}}
\defmath\nLL{\overline{\LL}}
\defmath\nNN{\overline{\NN}}
\defmath\nWW{\overline{\WW}}
\defmath\nEE{\overline{\EX}}

\defmath{\States}{S}
\defmath\Ta{\rightarrow}
\defmath\Ti{\rightarrow_{t}}
\defmath\Tj{\rightarrow_{u}}

\defmath\sink{l_\mathit{sink}}

\defmath\tx{{\overline{t}}}

\defmath{\cfg}{{\textsc{cfa}}}
\defmath\DYN{\mathsf{dyn}}
\defmath\cfgdyn{\cfg^{\DYN}}

\defmath{\threads}{\mathsf{Threads}}
\defmath{\pc}{\mathsf{pc}}
\defmath{\stmt}{\mathsf{stmt}}
\defmath{\stmts}{\mathsf{Stmts}}
\newcommand{\gstmt}[2]{\ensuremath{{#1}\blacktriangleright{#2}}}
\defmath{\vars}{\mathsf{Vars}}
\defmath{\data}{\mathsf{Data}}
\defmath{\vals}{\mathsf{Vals}}
\defmath{\inputs}{\mathsf{Inputs}}
\defmath{\locs}{\mathsf{Locs}}
\defmath{\bool}{\mathbb{B}}

\newcommand{\old}[1]{}

\defmath{\Left}{\mathsf{Left}}
\defmath{\Right}{\mathsf{Right}}

\newcommand\ccode[1]{\texttt{#1}}

\defmath{\pre}{\mathsf{Pre}}
\defmath{\post}{\mathsf{Post}}
\defmath{\ext}{\mathsf{Ext}}
\defmath{\err}{\mathsf{Err}}

\newcommand{\overarrowi}[1]{\xrightarrow{#1}}
\newcommand{\overarrow}[1]{
  \mathchoice{\raisebox{-3pt}{ $\overarrowi{#1}$ }}
             {\raisebox{-3pt}{ $\overarrowi{#1}$ }}
             {\raisebox{-3pt}{ $\overarrowi{#1}$ }}
             {\raisebox{-3pt}{ $\overarrowi{#1}$ }}}

\newcommand{\trans}[4]{\ensuremath{{#1\,}{\overarrow{#3}_{\hspace{-.9mm}#4}}{\,#2}}}
\newcommand{\tr}[2]{\trans{\hspace{-1mm}}{\hspace{.1mm}}{#1}{#2}}

\defmath\ta{\rightarrow}
\defmath\ti{\rightarrow_{t}}
\defmath\tj{\rightarrow_{u}}

\defmath\tpre{\overarrowi{\pre}}
\defmath\tpost{\overarrowi{\post}}
\defmath\tex{\overarrowi{\ext}}

\providecommand{\tuple}[1]{\ensuremath{\left( #1 \right)}}
\providecommand{\set}[1]{\ensuremath{\left\lbrace #1 \right\rbrace}}

\providecommand{\sizeof}[1]{\ensuremath{\left\vert{#1}\right\vert}}

\newcommand{\defn}{\,\triangleq\,}
\newcommand{\rrestr}[0]{\hspace{-.7mm}\mathrel{\reflectbox{\rotatebox[origin=tl]{-25}{$\|$}}}\hspace{-.7mm}}
\defmath{\lrestr}{\hspace{-.7mm}\mathrel{\reflectbox{\rotatebox[origin=tr]{25}{$\|$}}}\hspace{-.7mm}}

\defmath{\followedby}{\mathrel{.}}
\defmath{\becomes}{\coloneqq}
\defmath{\true}{{\mathrm{true}}\xspace}
\defmath{\false}{{\mathrm{false}}\xspace}
\defmath{\tbool}{{\mathsf{Bool}}}
\defmath{\sortbool}{{\mathbb{B}}}

\defmath{\Nat}{{\mathbb N}}
\defmath{\Bool}{{\mathbb B}}
\defmath{\sortnat}{{\mathbb{N}}}
\defmath{\Land}{{\bigwedge}}
\defmath{\Lor}{{\bigvee}}

\defmath\pr{^\pre}
\defmath\po{^\post}
\defmath\ex{^\ext}
\defmath\er{^\err}
\defmath\safe{^{\mathsf{safe}}}
\defmath\ner{\safe}
\defmath\lc{l_1}
\defmath\kc{l_2}
\defmath\ls{l_1\safe}
\defmath\ks{l_2\safe}
\defmath\lx{l_1\ex}
\defmath\kx{l_2\ex}
\defmath\lr{l_1\pr}
\defmath\kr{l_2\pr}
\defmath\lo{l_1\po}
\defmath\ko{l_2\po}

\defmath\CE{CE}
\defmath\CL{CL}
\defmath\CR{CR}

\title{Dynamic Reductions for Model Checking Concurrent Software}

\author{
  {Henning G\"unther\inst{1}}
  \and {Alfons Laarman\inst{1}}
  \and {Ana Sokolova\inst{2}} 
  \and {Georg~Weissenbacher\inst{1}}
}

\institute{
  TU Wien\thanks{\scriptsize\vspace{-.3ex} This work is
supported by the Austrian National Research Network S11403-N23 (RiSE)
of the Austrian Science Fund (FWF) \vspace{-.3ex}and by the Vienna Science and
Technology Fund (WWTF) through grant VRG11-005.}
\and
University of Salzburg
}

\begin{document}

\maketitle

\begin{abstract}
Symbolic model checking of parallel programs stands and falls
with effective methods of dealing with the explosion of interleavings.
We propose a dynamic reduction technique to avoid unnecessary interleavings.
By extending Lipton's original work with a notion of bisimilarity,
we accommodate dynamic transactions, and thereby
reduce dependence on the accuracy of static analysis, which
is a severe bottleneck in other reduction techniques.

The combination of symbolic model checking and dynamic reduction
techniques has proven to be challenging in the past.
Our generic reduction theorem nonetheless enables us to 
derive an efficient symbolic encoding, which we
implemented for IC3 and BMC.
The experiments demonstrate the power of dynamic reduction 
on several case studies and a large set of SVCOMP benchmarks.

\end{abstract}


\section{Introduction}
\label{sec:introduction}

The rise of multi-threaded software---a consequence of a
necessary technological shift from ever higher 
frequencies to multi-core architectures---exacerbates the
challenge of verifying programs automatically. While
automated software verification has made impressive advances 
recently thanks to
novel symbolic model checking techniques, such as
lazy abstraction \cite{hrmgs02,beyerCAV07},
interpolation \cite{impact}, and IC3 \cite{ic3}
for software \cite{CTIGAR,CimattiGMT14},  multi-threaded
programs still pose a formidable challenge.

The effectiveness of model checking in the presence of
concurrency is severely limited by the state explosion caused through thread 
interleavings.
Consequently, techniques that
avoid thread interleavings,
such as partial order reduction (POR) \cite{peled-93,parle89,godefroid}
or Lipton's reduction \cite{lipton}, are crucial to the
scalability of model checking,
while also benefitting other verification approaches~\cite{qadeer-atomicity,dimitrov2014commutativity,ElmasQT09}.

These reduction techniques, however, rely heavily
on the identification of statements that are either independent
or commute with the statements of all other threads, i.e. those that are
\concept{globally independent}.
For instance, the single-action rule~\cite{lamport-lipton}---a primitive precursor of Lipton
reduction---states that a sequential block of statements can be considered
an atomic transaction if all but one of the statements are
globally independent.
Inside an atomic block, all interleavings of other threads
can be discarded, thus yielding the reduction.

Identifying these globally independent statements requires
non-local \concept{static analyses}.
In the presence of pointers, arrays, and complicated branching
structures, however, the results of an up-front static analysis are typically 
extremely conservative, thus a severe bottleneck for good reduction.

\autoref{fig:illustration} shows an example with two threads (T1 and T2).
Let's assume
static analysis can establish that pointers \ccode{p} and \ccode{q}
never point to the same memory throughout the program's (parallel) execution.
This means that statements involving the pointers are globally independent,
hence they globally commute, e.g. an interleaving \ccode{*p++; *q = 1} always yields
the same result as \ccode{*q = 1; *p++;}.
Assuming that \ccode{*p++;} is also independent of the other
statements from T2 (\ccode{b = 2} and \ccode{c = 3}),
we can reorder any trace of the parallel program to a trace where
\ccode{*p++} and \ccode{*q = 2} occur subsequently without affecting the 
resulting state. The figure shows one example.
Therefore,
a syntactic transformation from \ccode{*p++; *q = 2} to
\ccode{atomic\{*p++; *q = 2\}} is a valid static reduction.

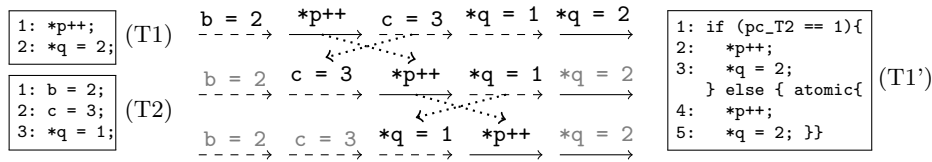
\begin{figure}[b]
\begin{minipage}{0.18\linewidth}
\cprotect\fbox{\scriptsize
\begin{minipage}{0.5\linewidth}
\begin{verbatim}
1: *p++;
2: *q = 2;
\end{verbatim}
\end{minipage}
}~(T1)

\vspace{1mm}

\cprotect\fbox{\scriptsize
\begin{minipage}{0.5\linewidth}
\begin{verbatim}
1: b = 2;
2: c = 3;
3: *q = 1;
\end{verbatim}
\end{minipage}
}~(T2)
\end{minipage}
\begin{minipage}{0.52\linewidth}
\begin{tikzpicture}

   \tikzstyle{e}=[minimum width=0cm]
   \tikzstyle{every node}=[font=\small, node distance=1.2cm]

	\node (s1) [e] {};
	\node (s2) [right of=s1,e] {};
	\node (s3) [right of=s2,e] {};
	\node (s4) [right of=s3,e] {};
	\node (s5) [right of=s4,e] {};
	\node (s6) [right of=s5,e] {};
	\path (s1.east) edge[->,dashed] node[above,pos=.45]{\ccode{b = 2}} (s2.west)
		  (s2.east) edge[->] 		node[above,pos=.45](a){\ccode{*p++}} (s3.west)
		  (s3.east) edge[->,dashed] node[above,pos=.45](e){\ccode{c = 3}} (s4.west)
		  (s4.east) edge[->,dashed] node[above,pos=.45]{\ccode{*q = 1}} (s5.west)
		  (s5.east) edge[->] 		node[above,pos=.45]{\ccode{*q = 2}} (s6.west);
	\node (s1a) [node distance=.8cm,below of=s1,e] {};
	\node (s2a) [right of=s1a,e] {};
	\node (s3a) [right of=s2a,e] {};
	\node (s4a) [right of=s3a,e] {};
	\node (s5a) [right of=s4a,e] {};
	\node (s6a) [right of=s5a,e] {};
	\path (s1a.east) edge[->,dashed] node[above,gray]{\ccode{b = 2}} (s2a.west)
		  (s2a.east) edge[->,dashed] node[above](ea){\ccode{c = 3}
		  										\vphantom{\ccode{*p++}}} (s3a.west)
		  (s3a.east) edge[->] 		node[above](aa){\ccode{*p++}} (s4a.west)
		  (s4a.east) edge[->,dashed] node[above](eea){\ccode{*q = 1}} (s5a.west)
		  (s5a.east) edge[->] 		node[above,gray]{\ccode{*q = 2}} (s6a.west);
	\node (s1b) [node distance=.8cm,below of=s1a,e] {};
	\node (s2b) [right of=s1b,e] {};
	\node (s3b) [right of=s2b,e] {};
	\node (s4b) [right of=s3b,e] {};
	\node (s5b) [right of=s4b,e] {};
	\node (s6b) [right of=s5b,e] {};
	\path (s1b.east) edge[->,dashed] node[above,gray]{\ccode{b = 2}} (s2b.west)
		  (s2b.east) edge[->,dashed] node[above,gray]{\ccode{c = 3}} (s3b.west)
		  (s3b.east) edge[->,dashed] node[above](eb){\ccode{*q = 1}} (s4b.west)
		  (s4b.east) edge[->] node(ab)[above]{\ccode{*p++}} (s5b.west)
		  (s5b.east) edge[->] node[above,gray]{\ccode{*q = 2}} (s6b.west);
	\path (a.south) edge[->,thick,dotted] node[pos=.45]{} (aa.north)
		  (aa.south) edge[->,thick,dotted] node[pos=.45]{} (ab.north)
		  (e.south) edge[->,thick,dotted] node[pos=.45]{} (ea.north)
		  (eea.south) edge[->,thick,dotted] node[pos=.45]{} (eb.north);

\end{tikzpicture}
\end{minipage}
\begin{minipage}{0.28\linewidth}
\cprotect\fbox{\scriptsize
\begin{minipage}{0.7\linewidth}
\begin{verbatim}
1: if (pc_T2 == 1){
2:   *p++;
3:   *q = 2;
   } else { atomic{
4:   *p++;
5:   *q = 2; }}
\end{verbatim}
\end{minipage}
}~(T1')
\end{minipage}
\caption{
(Left) C code for threads T1 and T2.
(Middle) Reordering (dotted lines) a multi-threaded execution trace (T1's actions are represented with straight arrows and T2's with `dashed' arrows).
(Right) The instrumented code for T1.}
\label{fig:illustration}
\end{figure}

Still, it is often hard to prove that pointers do not overlap throughout a program's execution.
Moreover, in many cases, pointers might temporarily
overlap at some point in the execution.
For instance, assume that initially \ccode{p} points to
the variable \ccode{b}.
This means that statements \ccode{b = 2} and \ccode{*p++}
no longer commute, because \ccode{b = 2; b++} 
yields a different result than \ccode{b++; b = 2}.
Nevertheless, if \ccode{b = 2} already happened, then we can still swap instructions and achieve the reduction as shown in~\autoref{fig:illustration}.
Traditional, static reduction methods cannot distinguish whether \ccode{b = 2}
already happened and yield no reduction.
\autoref{sec:motivation} provides various other 
real-world examples of hard cases for static analysis.

In \autoref{sec:instrument},
we propose a dynamic reduction method that is still based on a similar
syntactic transformation. Instead of merely making sequences of statements 
atomic, it introduces branches as shown in \autoref{fig:illustration} (T1').
A dynamic commutativity condition determines whether the branch with or
without reduction is taken. In our example, the condition checks whether
the program counter of T2 (\ccode{pc\_T2}) still points to the statement
\ccode{b = 2} (\ccode{pc\_T2} == 1). In that case, no reduction is performed,
otherwise the branch with reduction is taken.
In addition to conditions on the program counters, we provide other heuristics
comparing pointer and array values dynamically.

The instrumented code (T1') however poses one problem: the branching condition
no longer commutes with the statement that enables it.
In this case, the execution of \ccode{b = 2} disables the condition, thus
before executing \ccode{b = 2}, T1' ends up at Line 2, whereas after
\ccode{b = 2} it ends up at Line 4 (see \autoref{fig:noncommutativity}).
\begin{figure}[t]
\centering
\begin{minipage}[b]{.56\textwidth}
\begin{tikzpicture}

   \tikzstyle{e}=[minimum width=.6cm]
   \tikzstyle{every node}=[font=\scriptsize, node distance=1.2cm]

  \node (s1) {$\tuple{1,1}$};
  \node (s2) [below of=s1] {$\tuple{2,1}$};
  \node (s3) [node distance=4.5cm,right of=s2] {$\tuple{2,2}$};
  \path (s2) edge[->] node[midway,above]{\ccode{b = 2}} (s3);
  \path (s1) edge[->] node(m)[midway,right]{\ccode{pc\_T2}==1} (s2);

  \node (s0) [node distance=.35cm,above of=s1]  {$\tuple{T1,T2}$};

  \node (s4) [node distance=2.5cm,right of=s1] {$\tuple{1,2}$};
  \path (s1) edge[->] node [midway,above,sloped]{\ccode{b = 2}} (s4);

  \node (s3p) [node distance=1cm,above of=s3] {$\tuple{4,2}$};
  \path (s4) edge[->] node[midway,above,sloped]{\ccode{pc\_T2!=1}}
  					 node[yshift=.3cm,midway,above,sloped]{(\ccode{\bf else})} (s3p);

  \path (s3p) -- node[sloped,pos=.48]{$\neq$} (s3);
\end{tikzpicture}
\caption{Loss of commutativity}
\label{fig:noncommutativity}

\end{minipage}
\begin{minipage}[b]{.26\textwidth}

\begin{tikzpicture}

   \tikzstyle{e}=[minimum width=1cm]
   \tikzstyle{every node}=[font=\scriptsize, node distance=1cm]

  \node (s1) {$4$};
  \node (s2) [below of=s1] {$2$};
  \node (s3) [node distance=1.8cm,right of=s2] {$3$};
  \path (s2) edge[->] node[midway,above]{\ccode{*p++}} (s3);
  \path (s1) -- node(m)[midway,sloped]{$\cong$} (s2);

  \node (s4) [node distance=1.8cm,right of=s1] {$5$};
  \path (s1) edge[->] node [midway,above,sloped]{\ccode{*p++}} (s4);

  \path (s4) -- node[midway,sloped]{$\cong$} (s3);
\end{tikzpicture}
\caption{Bisimulation (T1)}
\label{fig:bisim-T1}

\end{minipage}
\vspace{-1em}
%
%
%
%
%
%
%
%
%
%
%
\end{figure}
To remedy this, we require in \autoref{sec:reduction}
that the instrumentation guarantees that the target states
in both branches are bisimilar.
\autoref{fig:bisim-T1} shows that locations 2 and 4 of T1' are bimilar,
written $2 \cong 4$,
which implies that any statement executable from the one is also executable
from the other ending again in a bisimilar location ($3\cong 5$).
As bisimularity is preserved under
parallel composition $(4,2) \cong (2,2)$,
we can prove the correctness of our dynamic reduction method in \autoref{app:proofs}.

The benefit of our syntactic approach is that the technique can be combined with
symbolic model checking checking techniques
(\autoref{sec:encoding} provides an encoding for our lean instrumentation).
Thus far, symbolic model checkers only supported more 
limited and static versions of reduction techniques as discussed
in \autoref{sec:related}. 

We implemented the dynamic reduction and encoding for LLVM bitcode,
mainly to enable support for C/C++ programs without dealing with
their intricate semantics
(the increased instruction count of LLVM bitcode is mitigated by the reduction).
The encoded transition relation is then passed to the
Vienna Verification Tool~(VVT)~\cite{vvt}, which implements both BMC
and IC3 algorithms extended with abstractions~\cite{CTIGAR}.
Experimental evaluation shows that (\autoref{sec:experiments})
dynamic reduction can yield several orders of magnitude gains in verification 
times.

%


\section{Motivating Examples}
\label{sec:motivation}

~\\
\vspace{-2.5em}

\begin{wrapfigure}[12]{R}{18.5em}
\centering\vspace{-2.5em}
  \fbox{
    \begin{minipage}{.9\columnwidth}
      \small
\scriptsize
      \begin{tabbing}
        \quad\=\quad\=\quad\=\quad\=\kill
        \>{\tt int *data = NULL;}\\
        \>{\tt void worker\_thread(int tid) \{}\\
        {\tt c:}\>\>{\tt if (data == NULL) \{}\\
        {\tt d:}\>\>\>{\tt int *tmp = read\_from\_disk(1024);}\\
        {\tt W:}\>\>\>{\tt if (!CAS(\&data, NULL, tmp)) free(tmp); }\\
        \>\>{\tt\}}\\
        \>\>{\tt for (int i = 0; i < 512; i++) }\\
        {\tt R:}\>\>\>{\tt process(data[i + tid * 512]);}\\
        \>{\tt\}}\\
        \>{\tt int main () \{}\\
        {\tt a:}\>\>{\tt pthread\_create(worker\_thread, 0);} // T1\\
        {\tt b:}\>\>{\tt pthread\_create(worker\_thread, 1);} // T2\\
        \>{\tt\}}
      \end{tabbing}
    \end{minipage}}
  \caption{Lazy initialization\label{fig:lazy}}
\end{wrapfigure}
\paragraph{Lazy initialization}
We illustrate our method with the code in \autoref{fig:lazy}.
The main function starts two threads
executing the \ccode{worker\_thread} function, which processes the contents of
\ccode{data} in the for loop at the end of the function.
Using a common pattern, a worker thread lazily delays the initialization of the 
global \ccode{data} pointer until it is needed.
It does this by reading some content from
disc and setting the pointer atomically
via a
compare-and-swap operation (\ccode{CAS})
at label {\tt W}
(whose semantics here is an atomic C-statement:
\ccode{if (data==NULL) \{ data = tmp; return 1; \} else return 0;}).
If it fails (returns 0), the locally allocated data is freed as the other thread
has won~the~race.

The subsequent read access at label {\tt R} is
only reachable once \ccode{data} 
has been initialized.
Consequently, the write access at {\tt W} cannot possibly
interfere with the read accesses at {\tt R},
and the many interleavings caused by both threads executing the for loop
can safely be ignored by the model checker.
This typical pattern is however too complex for static analysis to
efficiently identify, causing the model checker to conservatively
assume conflicting accesses, preventing any reduction.


\begin{figure}[b]
\begin{minipage}[b]{0.5\linewidth}
\scriptsize
\cprotect\fbox{
\begin{minipage}{.9\linewidth}
\begin{verbatim}
  int T[10] = {E,E,22,35,46,25,E,E,91,E};

  int find-or-put(int v) {
     int hash = v / 10;
     for (int i = 0; i < 10; i++) {
        int index = (i + hash) % 10;
        if (CAS(&T[index], E, v)) {
           return INSERTED;        
        } else if (T[index] == v) 
           return FOUND;
     }
     return TABLE_FULL;
  }
  int main() {
     pthread_create(find-or-put, 25);
     pthread_create(find-or-put, 42);
     pthread_create(find-or-put, 78);
  }
\end{verbatim}
\end{minipage}
}
\caption{Lockless hash table.}
\label{ex:ht}
\end{minipage}
\begin{minipage}[b]{0.5\linewidth}
\scriptsize
\cprotect\fbox{
\begin{minipage}{.9\linewidth}
\begin{verbatim}
  int x = 0, y = 0;
  int *p1, *p2;

  void worker(int *p) {
     while (*p < 1024)
        *p++;
  }
  int main(){
a:   if (*)
b:      { p1 = &x; p2 = &y; }
     else
c:      { p1 = &y; p2 = &x; }
     pthread_create(worker, p1); // T1
     pthread_create(worker, p2); // T2
     pthread_join(t1);
     pthread_join(t2);
     return x+y;
  }
\end{verbatim}
\end{minipage}
}
\caption{Load balancing.}
\label{ex:pointers}
\end{minipage}
\end{figure}

\paragraph{Hash table}~The code in \autoref{ex:ht} 
implements a lockless hash table~(from \cite{boosting}) inserting a value \ccode{v} by
scanning the bucket array \ccode{T} starting from \ccode{hash}, the
hash value calculated from \ccode v.
If an empty bucket is found (\ccode{T[index]==E}), then
\ccode{v} is atomically inserted using the
\ccode{CAS} operation.
If the bucket is not empty, the operation checks whether the value was already inserted in the bucket (\ccode{T[index] == v}).
If that is not the case, then it probes the next bucket of \ccode{T}
until either \ccode{v} is found to be already in the table, or it is inserted in an empty slot, or the table is full.
This basic bucket search order is called a linear \concept{probe sequence}.

A thread performing \ccode{find-or-put(25)}, for instance,
merely reads buckets \ccode{T[2]} to \ccode{T[5]}.
However, other threads might write an empty bucket, thus causing interference.
To show that these reads are independent, the static analysis would 
have to demonstrate that the writes happen to different buckets.
Normally this is done via alias analysis that tries to identify
the buckets that are written to (by the CAS operation).
However, because of the hashing and the probe sequence, such an analysis can
only conclude that all buckets may be written.
So all operations involving \ccode{T} will be classified as
non-movers, also the reads.
However if we look at the state of individual buckets, it turns out that
a common pattern is followed using the CAS operation: a bucket is only written
when it is empty, thereafter it doesn't change.
In other words, when a bucket \ccode{T[i]} does not contain \ccode{E}, then
any operation on it is a read and consequently is independent.

\paragraph{Load balancing}~\autoref{ex:pointers} shows a simplified example of a common pattern in multi-threaded software; load balancing:
The work to be done (counting to 2048) is split up between two threads (each of which counts to 1024).
The work assignment is represented by pointers \ccode{p1} and \ccode{p2}, and a dynamic hand-off to one of the two threads is simulated using non-determinism (the first if branch).
Static analysis cannot establish the fact that the partitions are independent, because they are assigned dynamically.
But because the pointer is unmodified after assignment,
its dereference commutes with  that in other worker threads.

In \autoref{sec:dynamic}, we show how the discussed operations in all three
examples can become dynamic movers, allowing for more reduction.


\section{Preliminaries}
\label{sec:prelim}

A concurrent program consists of a finite number of sequential procedures, one for each thread $i$. We model the syntax of each thread $i$ by a control flow graph (CFG) $G_i = (V_i, \delta_i)$ with $\delta_i \subseteq V_i \times A \times V_i$ and $A$ being the set of actions, i.e., statements.
$V_i$ is a finite set of locations, and
 $(l,\alpha,l') \in \delta_i$ are (CFG) edges.
 We abbreviate the actions for a thread $i$ with
$\Delta_i = \set{\alpha \mid \exists l,l' \colon (l,\alpha,l')\in\delta_i )}$.

\begin{wrapfigure}{r}{10.5em}
\begin{minipage}{8.5em}\vspace{-3.5em}
\begin{empheq}[box={\Garybox[Domains]}]{align*}
i,j,k 					& \colon \threads\\
a,b,x,y,p,p'			& \colon \vars\\
c,c',\dots				& \colon \vals\\
l,l',l_1,\ldots 		& \colon V_i\\
d,d'					& \colon \data\\
\pc,\pc',\ldots			& \colon \locs\\
\sigma,\sigma',\ldots	& \colon S\\
\alpha_i 				& \colon \mathcal{P}(\data^2)
\end{empheq}
\end{minipage}\vspace{-2em}
\end{wrapfigure}

A state of the concurrent system is composed of
(1) a location for each thread, i.e., a a tuple of thread locations (the set 
\locs contains all such tuples), and
(2) a data valuation, i.e., a mapping from variables ($\vars$) to data values ($\vals$). 
We take $\data$~to be the set of all data valuations. 
Hence, a state is a pair, $\sigma = (\pc,d)$ where
$\pc \in {\prod}_i V_i$ and $d \in \data$. 
The locations in each CFG 
actually correspond to the values  of the thread-local program counters for each thread. In particular, the global locations correspond to the global program counter $\pc$ being a tuple with $\pc_i \in V_i$ the thread-local program counter for thread $i$. We use $\pc[i := l]$ to denote $\pc[i := l]_i=l$ and
$\pc[i := l]_j = \pc_j$ for all $j\neq i$.

Each possible action $\alpha$ semantically corresponds to a binary
relation $\alpha \subseteq \data \times \data$ representing the
evolution of the data part of a state under the transition labelled by
$\alpha$. We call $\alpha$ the transition relation of the statement
$\alpha$, referring to both simply as $\alpha$.
We also use several simple statements from programming
languages, such as \texttt{C}, as actions.

The semantics of a concurrent program consisting of a finite number of threads, each with CFG $G_i = (V_i, \delta_i)$, is a 
transition system with data (TS) $C = (S, \to)$ with $S = \locs\times\data$, $\locs = \prod_i V_i$ and $\to = \bigcup_i \to_i$ where $\to_i$ is given by
$(\pc,d) \to_i (\pc',d')$ for $\exists \alpha\colon
\pc_i = l
 \land  (l,\alpha,l')\in\delta_i \land (d,d') \in \alpha \land \pc'=\pc[i := l']$. 
We also write 
$(\pc,d) \stackrel{\alpha}{\to}_i (\pc',d')$ for
$\pc_i = l \land (l,\alpha,l')\in\delta_i \land (d,d') \in \alpha \land \pc'=\pc[i := l']$. 
Hence, the concurrent program is an asynchronous execution of the parallel composition of all its threads. Each step (transition) is a local step of one of the threads.  
Each thread $i$ has a unique initial location $\pc_{0,i}$, and hence the TS has
one initial location $\pc_0$. Moreover, there is an initial data valuation $d_0$ as well.
Hence, the initial state of a TS is 
$\sigma_0\defn (\pc_0, d_0)$.

Since we focus on preserving simple safety properties (e.g. assertions) in our
reduction, w.l.o.g., we require one sink location per thread
\sink to represent errors (it has no outgoing edges, no selfloop).
Correspondingly, error states of a TS are those in which at least one thread is in the error location.

In the following, we introduce additional notation for states and relations.
Let $R \subseteq S \times S$ and $X \subseteq S$. Then
left restriction of $R$ to $X$ is $X \lrestr R \defn R\,\cap\,(X \times S)$ and right restriction
is $R\rrestr X \defn R\,\cap\,(S \times X)$.
Finally, the complement of $X$ is denoted
$\overline X \defn S  \setminus X$ (the universe of all states remains implicit in this notation).

\paragraph{Commutativity}
We let $R\circ Q$ denote the \concept{sequential composition}
of two binary relations $R$ and $Q$, defined~as:
$
\{ (x,z)\,\vert\,\exists y\colon (x,y)\in R
\wedge (y,z)\in Q \}\,
$.
Moreover,~let:
\begin{equation*}
  \begin{aligned}
R \comm Q &\defn  R \circ Q = Q \circ R
	&\hfill\text{(both-commute)}\\
R \rcomm Q &\defn R \circ Q \subseteq Q \circ R
	&\hfill\text{($R$ right commutes with $Q$)}\\
R \lcomm Q &\defn R \circ Q \supseteq Q \circ R
	&\hfill\text{($R$ left commutes with $Q$)}\\
  \end{aligned}
\end{equation*}
Illustrated graphically for transition relations, $\rightarrow_i$
right commutes with $\rightarrow_j$ iff
\begin{equation}\label{eq:commute}
\begin{aligned}
\text{

\begin{tikzpicture}\centering

  \node (s0) {\small $\sigma$};

  
  \node (s1) [below of=s0] {\small $\sigma'$};
  \node (s2) [node distance=.5cm,right of=s1, xshift=.9cm] {\small $\sigma''$};
  \path (s0) -- node(m)[midway,sloped]{$\rightarrow_i$} (s1);
  \path (s1) -- node[midway]{$\rightarrow_j$} (s2);

  \node [right of=s2,yshift=.5cm,xshift=-.5cm]{$\Rightarrow$};

  \node [left of=m,xshift=-.2cm]
   {\small $\forall \sigma,\sigma',\sigma'':$};

  \node (s3) [node distance=.5cm,right of=s0, xshift=3cm]{\small $\sigma$};
  \node (s4) [node distance=.5cm,right of=s3, xshift=.9cm] {\small $\sigma'''$};
  \node (s5) [below of=s4] {\small $\sigma''$};
  \path (s4) -- node[midway,sloped]{$\rightarrow_i$} (s5);
  \path (s3) -- node[pos=.45]{$\rightarrow_j$} (s4);

  \node (s6) [gray,below of=s3] {\small $\sigma'$};
  \path (s3) -- node(m)[gray,midway,sloped]{$\rightarrow_i$} (s6);
  \path (s6) -- node[gray,pos=.48]{$\rightarrow_j$} (s5);
        
  \node [left of=m,xshift=.3cm] {\small $\exists \sigma''':$};

\end{tikzpicture}}
\end{aligned}
\end{equation}
%
%
Conversely, $\rightarrow_j$ \emph{left commutes} with $\rightarrow_i$.
The typical example of (both) commuting operations
$\tr{\alpha}{i}$ and $\tr{\beta}{i}$ is when $\alpha$
and $\beta$ access a disjoint set of variables.
Two operations may commute even if both access the same variables,
e.g., if both only read or
both (atomically) increment/decrement the same variable.

\begin{figure}[b]
\centering
\scalebox{.8}{

\begin{tikzpicture}[node distance=1.7cm]
\newcommand\xstate[3]{(\tuple{#1,#2},#3)}
\renewcommand\xstate[3]{\tuple{#3}}

\tikzstyle{e}=[line width=2pt]

    \node (s0) 	 					{};
    \node (s1) [below left	of=s0] 	{};
    \node (s2) [below right of=s0] 	{};
    \node (s3) [below left 	of=s1] 	{};
    \node (s4) [below right of=s1]	{};
    \node (s5) [below right of=s2] 	{};
    \node (s6) [below right	of=s3] 	{};
    \node (s7) [below right	of=s4] 	{};
    \node (s8) [below left	of=s7] 	{};

    \path (s0) edge[->] node(t1')[midway,above,sloped] {\ccode{a=0;}}
    					node(t1')[midway,below,sloped] {$1$}  (s1);
    \path (s0) edge[->] node(t1')[midway,above,sloped] {\ccode{x=1;}}
    					node(t1')[midway,below,sloped] {$2$} (s2);
    \path (s1) edge[->] node(t1')[midway,above,sloped] {\ccode{b=2;}}
    					node(t1')[midway,below,sloped] {$1$} (s3);
    \path (s1) edge[->] node(t1')[midway,above,sloped] {\ccode{x=1;}}
					    node(t1')[midway,below,sloped] {$2$} (s4);
    \path (s2) edge[->] node(t1')[midway,above,sloped] {\ccode{a=0;}}
					    node(t1')[midway,below,sloped] {$1$} (s4);
    \path (s2) edge[->] node(t1')[midway,above,sloped] {\ccode{y=2;}}
					    node(t1')[midway,below,sloped] {$2$} (s5);
    \path (s4) edge[->] node(t1')[midway,above,sloped] {\ccode{b=2;}}
					    node(t1')[midway,below,sloped] {$1$} (s6);
    \path (s4) edge[->] node(t1')[midway,above,sloped] {\ccode{y=2;}}
    					node(t1')[midway,below,sloped] {$2$} (s7);
    \path (s3) edge[->] node(t1')[midway,above,sloped] {\ccode{x=1;}}
					    node(t1')[midway,below,sloped] {$2$} (s6);
    \path (s5) edge[->] node(t1')[midway,above,sloped] {\ccode{a=0;}}
    					node(t1')[midway,below,sloped] {$1$} (s7);
    \path (s7) edge[->] node(t1')[midway,above,sloped] {\ccode{b=2;}}
					    node(t1')[midway,below,sloped] {$1$} (s8);
    \path (s6) edge[->] node(t1')[midway,above,sloped] {\ccode{y=2;}}
    					node(t1')[midway,below,sloped] {$2$} (s8);

    \path (s0) edge[->,e,bend right=45] node(n1)[midway,above,sloped] {\ccode{a=0;b=2;}} (s3);
    \path (s0) edge[->,e,bend left=45] node(n2)[midway,above,sloped] {\ccode{x=1;y=2;}} (s5);
    \path (s5) edge[->,e,bend left=45] node(t1')[midway,below,sloped] {\ccode{a=0;b=2;}} (s8);
    \path (s3) edge[->,e,bend right=45] node(t1')[midway,below,sloped] {\ccode{x=1;y=2;}} (s8);

    \node (s0p) 	[node distance=5.5cm,right of=s0] {};
    \node (s1) [below left	of=s0p] 	{};
    \node (s2) [below right of=s0p] 	{};
    \node (s3) [below left 	of=s1] 	{};
    \node (s4) [below right of=s1]	{};
    \node (s5) [below right of=s2] 	{};
    \node (s6) [below right	of=s3] 	{};
    \node (s7) [below right	of=s4] 	{};
    \node (s8) [below left	of=s7] 	{};

    \path (s0p) edge[e,->] node(t1')[midway,above,sloped] {\ccode{a=0;}}
    					node(t1')[midway,below,sloped] {$1$}  (s1);
    \path (s0p) edge[->] node(t1')[midway,above,sloped] {\ccode{x=1;}}
    					node(t1')[midway,below,sloped] {$2$} (s2);
    \path (s1) edge[->] node(t1')[midway,above,sloped] {\ccode{b=2;}}
    					node(t1')[midway,below,sloped] {$1$} (s3);
    \path (s1) edge[e,->] node(t1')[midway,above,sloped] {\ccode{x=1;}}
					    node(t1')[midway,below,sloped] {$2$} (s4);
    \path (s2) edge[->] node(t1')[midway,above,sloped] {\ccode{a=0;}}
					    node(t1')[midway,below,sloped] {$1$} (s4);
    \path (s2) edge[->] node(t1')[midway,above,sloped] {\ccode{y=2;}}
					    node(t1')[midway,below,sloped] {$2$} (s5);
    \path (s4) edge[e,->] node(t1')[midway,above,sloped] {\ccode{b=2;}}
					    node(t1')[midway,below,sloped] {$1$} (s6);
    \path (s4) edge[->] node(t1')[midway,above,sloped] {\ccode{y=2;}}
    					node(t1')[midway,below,sloped] {$2$} (s7);
    \path (s3) edge[->] node(t1')[midway,above,sloped] {\ccode{x=1;}}
					    node(t1')[midway,below,sloped] {$2$} (s6);
    \path (s5) edge[->] node(t1')[midway,above,sloped] {\ccode{a=0;}}
    					node(t1')[midway,below,sloped] {$1$} (s7);
    \path (s7) edge[->] node(t1')[midway,above,sloped] {\ccode{b=2;}}
					    node(t1')[midway,below,sloped] {$1$} (s8);
    \path (s6) edge[e,->] node(t1')[midway,above,sloped] {\ccode{y=2;}}
    					node(t1')[midway,below,sloped] {$2$} (s8);

\end{tikzpicture}
}
\caption{Example transition system composed of two independent threads (twice).
Thick lines show a Lipton reduced system (left) and a partial-order~reduction
(right).
}
\label{f:lipton}
\end{figure}
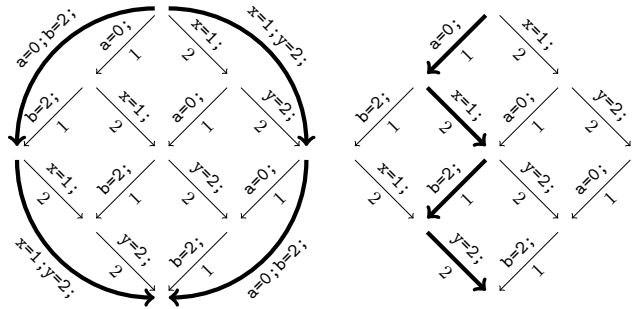

\paragraph{Lipton Reduction}
Lipton~\cite{lipton} devised a method
that merges multiple sequential
statements into one atomic operation,
and thereby radically reducing the
number of states reachable from the initial state as \autoref{f:lipton} shows 
for a transition system composed of two (independent, thus commuting) threads.

Lipton called a transition $\stackrel{\alpha}{\to}_i$ a right/left mover
if and only if it satisfies:

\begin{equation*}
  \label{eq:leftrightmovers}
  \begin{aligned}
	\stackrel{\alpha}{\to}_i \rcomm \bigcup_{j\neq i} \to_j &~\text{(right mover)}
	& \stackrel{\alpha}{\to}_i \lcomm \bigcup_{j\neq i} \to_j  &~\text{(left mover)} \\
  \end{aligned}
\end{equation*}

\noindent
\concept{Both-movers} are transitions that are both left and right movers,
whereas \concept{non-movers} are neither.
The sequential composition of two movers is also a corresponding mover,
and vice versa.
Moreover, one may always safely classify an action as a non-mover, although
having more movers yields better reductions.

Lipton reduction only preserves halting.
We present Lamport's~\cite{lamport-lipton} version, which preserves
safety properties such as $\Box \varphi$:
Any sequence $\tr{\alpha_1}{i}\circ \tr{\alpha_2}{i}\circ
\dots \circ \tr{\alpha_{n-1}}{i} \circ \tr{\alpha_n}{i}$ can be
\concept{reduced} to a single \concept{transaction} 
\tr{\alpha}{i} where $\alpha = \alpha_1 ;\dots ;\alpha_n$
(i.e. a compound statement with the same local behavior),
if for some $1 \le k < n$:
\begin{enumerate}[parsep=0pt]
\item[L1.] statements before $\alpha_k$ are right movers, i.e.:
	$\tr{\alpha_1}{i}\circ \dots\circ \tr{\alpha_{k-1}}{i} \rcomm \bigcup_{j\neq i} \to_j$,
\item[L2.] statements after $\alpha_{k}$ are left movers, i.e.:
	$\tr{\alpha_{k+1}}{i}\circ \dots\circ \tr{\alpha_{n}}{i}  \lcomm \bigcup_{j\neq i} \to_j$,
\item[L3.] statements after $\alpha_1$ do not block, i.e.:
	$\forall 1< x \le n, d\colon \exists d' \colon (d,d') \in \alpha_x$, and
	\vphantom{$\lcomm\bigcup_{j\neq i}$}
\item[L4.] 
	$\varphi$ is not disabled by $\tr{\alpha_1}{i}\circ \dots\circ \tr{\alpha_{k-1}}{i}$, nor enabled by
	 $\tr{\alpha_{k+1}}{i}\circ \dots\circ \tr{\alpha_{n}}{i}$.\parbox{0pt}{\vphantom{$\lcomm\bigcup_{j\neq i}$}}
\end{enumerate}
The action $\alpha_k$ might interact with other threads and therefore
is called the \emph{commit} in the
database terminology~\cite{papadimitriou}.
Actions preceding it are called \concept{pre-commit} actions and gather resources, such as locks.
The remaining  actions are \concept{post-commit} actions that (should) release these resources.
We refer to pre(/post)-commit transitions including source and target states
as the \concept{pre(/post) phase}.

\section{Dynamic Reduction}
\label{sec:dynamic}

The reduction outlined above depends on the identification of 
movers. And to determine whether a statement is a mover,
the analysis has to consider \emph{\underline{all} other statements in \underline{all} other threads}.
Why is the definition of movers so strong?
The answer is that `movability' has to be preserved in all future
computations for the reduction not to miss any relevant behavior.

\begin{wrapfigure}[14]{r}{.34\textwidth}
\vspace{-3.2em}
\scalebox{.7}{

\begin{tikzpicture}[node distance=1.7cm]
\newcommand\xstate[3]{(\tuple{#1,#2},#3)}
\renewcommand\xstate[3]{\tuple{#3}}

\tikzstyle{e}=[line width=2pt]

    \node (s0) 	 					{\parbox{.5cm}{$\xstate{1}{a}{x,y}$\\$\xstate{1}{a}{0,0}$}};
    \node (s1) [below left	of=s0] 	{$\xstate{2}{a}{0,0}$};
    \node (s2) [below right of=s0] 	{$\xstate{1}{b}{0,1}$};
    \node (s3) [below left 	of=s1] 	{$\xstate{3}{a}{0,2}$};
    \node (s4) [below right of=s1]	{$\xstate{2}{b}{0,1}$};
    \node (s5) [below right of=s2] 	{$\xstate{1}{c}{1,1}$};
    \node (s6) [below right	of=s3] 	{$\xstate{3}{b}{0,2}$};
    \node (s7) [below right	of=s4] 	{$\xstate{2}{c}{1,1}$};
    \node (s7p)[below	of=s5] 		{$\xstate{2}{c}{0,1}$};
    \node (s7pp) [below	of=s7] 	{$\xstate{2}{c}{1,2}$};
    \node (s8pp) [below of=s7p] 	{$\xstate{3}{b}{0,2}$};
    \node (s6pp) [below	of=s6] 	{$\xstate{3}{b}{2,2}$};
    \node (s6p) [below of=s3] 	{$\xstate{3}{b}{0,1}$};
    \node (s8p) [below of=s6p] 	{$\xstate{3}{b}{1,1}$};

    \path (s0) edge[->]    node[below,sloped]{1} node(t1')[midway,above,sloped] {\ccode{x=0;}} (s1);
    \path (s0) edge[->]    node[below,sloped]{2} node(t1')[midway,above,sloped] {\ccode{y=1;}} (s2);
    \path (s1) edge[->]    node[below,sloped]{1} node(t1')[midway,above,sloped] {\ccode{y=2;}} (s3);
    \path (s1) edge[->]    node[below,sloped]{2} node(t1')[midway,above,sloped] {\ccode{y=1;}} (s4);
    \path (s2) edge[->]    node[below,sloped]{1} node(t1')[midway,above,sloped] {\ccode{x=0;}} (s4);
    \path (s2) edge[->]    node[below,sloped]{2} node(t1')[midway,above,sloped] {\ccode{x=y;}} (s5);
    \path (s4) edge[->]    node[below,sloped]{1} node(t1')[midway,above,sloped] {\ccode{y=2;}} (s6);
    \path (s4) edge[->]    node[below,sloped]{2} node(t1')[midway,above,sloped] {\ccode{x=y;}} (s7);
    \path (s5) edge[->,e]  node[below,sloped]{1} node(t1')[midway,above,sloped] {\ccode{x=0;}} (s7p);
    \path (s3) edge[->,e]  node[below,sloped]{2} node(t1')[midway,above,sloped] {\ccode{y=1;}} (s6p);
    \path (s6p) edge[->,e] node[below,sloped]{2} node(t1')[midway,above,sloped] {\ccode{x=y;}} (s8p);
    \path (s7p) edge[->,e] node[below,sloped]{1} node(t1')[midway,above,sloped] {\ccode{y=2;}} (s8pp);
    \path (s6) edge[->]    node[below,sloped]{2} node(t1')[midway,above,sloped] {\ccode{x=y;}} (s6pp);
    \path (s7) edge[->]    node[below,sloped]{1} node(t1')[midway,above,sloped] {\ccode{y=2;}} (s7pp);

    \path (s0) edge[->,e,bend right=45] node(t1')[midway,above,sloped] {\ccode{x=0;y=2;}} (s3);
    \path (s0) edge[->,e,bend left=45] node(t1')[midway,above,sloped] {\ccode{y=1;x=y;}} (s5);

\end{tikzpicture}
}
\caption{Transition system of $\protect\tr{x:=0}{1}\protect\tr{y:=2}{1}\|\protect\tr{y:=1}{2}\protect\tr{x:=y}{2}$. Thick lines show an incorrect~reduction,
missing $(2,2)$~and~$(1,2)$.}
\label{f:composition}
\end{wrapfigure}
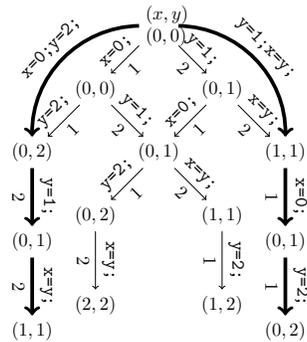
For instance,
consider the system composed of \ccode{x=0; y=2} and
\ccode{y=1; x=y} with
initial state $\sigma_0=(\pc_0,d_0)$, $d_0=\tuple{x=0,y=0}$ and $\pc_0=\tuple{1,1}$
using line numbers as program counters.
\autoref{f:composition} shows the TS of this system, from which we
can derive that \ccode{x:=0} and \ccode{y:=1} do not commute except
in the initial state
(see the diamond structure of the top 3 and the middle state).
Now assume, we have a dynamic version of Lipton reduction that
allows us to apply the reduction 
\ccode{atomic\{x=0; y=2;\}} and
\ccode{atomic\{y=1; x=y;\}}, but only in the initial state where
both \ccode{x=0} and \ccode{y=1} commute.
The resulting reduced system, as shown with bold arrows, now 
discards various states.
Clearly, a safety property such as $\Box \neg(x=1 \land y = 2)$
is not preserved anymore by such a reduction, even though 
\ccode{x=0} and \ccode{y=1} never disable the property (L4 in \autoref{sec:prelim} holds).

The mover definition comparing all behaviors of all other threads is thus merely 
a way to (over)estimate the computational future. 
But we can do better, without precisely calculating the future computations
(which would indulge in a task that the reduction is trying to avoid in the
first place).
For example, unreachable code should not negatively
impact movability of statements in the entire program.
By the same line of reasoning, we can conclude that lazy initialization
procedures (e.g. \autoref{fig:lazy}) should not eliminate movers in the
remainder of the program.
Intuitively, one can run the program until after initialization, then remove the
initialization procedure and restart the verification using that state as the 
new initial state. 
Similarly, reading unchanging buckets in the hash table of
\autoref{ex:ht} should not cause interference and
dynamically assigned, yet disjoint, pointers still do not overlap,
so these bucket reads and pointer dereferences could also become movers
after initialization.
The current section provides dynamic notion of movability and a generalized
reduction theorem that can use this new notion.
Proofs of all lemmas and theorems can be found in
\autoref{app:proofs}.

\subsection{Dynamic Movers}
\label{s:dyn-movers}


Recall from the example of \autoref{fig:illustration} that we introduce
branches in order to guide the dynamic reductions.
This section formalizes the concept of a dynamic both-moving condition,
guarding these branches.
We only consider both movers for ease of explanation. Nonetheless,
\autoref{app:proofs} considers left and right movers.

\begin{definition}[Dynamic both-moving conditions]\label{def:dynamicboth}~\\
A state predicate (a subset of states)
$c_\alpha$ is a dynamic both-moving
condition for a CFG edge $\tuple{l,\alpha,l'}\in\delta_i$,
if for all $j\neq i$: 
$(c_\alpha \lrestr \tr{\alpha}i) \bowtie (c_\alpha \lrestr \tr{}j)$
and both $\tr\alpha{i}$,$\tr{}j$ preserve $c_\alpha=\true$, i.e.
$c_\alpha\lrestr \tr{}j \rrestr \overline{c_\alpha}= 
c_\alpha\lrestr \tr{\alpha}i \rrestr \overline{c_\alpha}= \emptyset$.
\end{definition}

One key property of a dynamic both-moving condition for $\alpha\in\Delta_i$ is its monotonicity:
In the transition system,
the condition $c_\alpha$ can be enabled by remote threads ($j\neq i$), but never disabled.
While the definition allows us to define many practical heuristics,
we have identified the following both-moving conditions as useful.
When these fail, $c_\alpha := \false$ can be taken to designate
$\alpha$ as a non-mover statically.
Although our heuristics still rely on static analysis,
the required information is easier to establish
(e.g. with basic control-flow analysis and
the identification of CAS instructions)
than for the global mover condition.

\begin{description}
\item[Reachability] 
As in \autoref{fig:lazy}, interfering actions, such as the write at label
\ccode{W}, may become unreachable
once a certain program location has been reached.
The dynamic condition for the read 
$\alpha\defn $\,\ccode{process(data[i + tid * 512])}$_i$ therefore becomes:
$c_\alpha := \bigwedge_{j\neq i} \bigwedge_{l\in L(j)} \pc_j \neq l$,
where $L(j)$ is the set of all locations in $V_j$ that can reach the location with label \ccode W
in $V_j$. For example, for thread T1 we obtain $c_\alpha :=$ \ccode{pc\_T2 != a,b,c,d,W}
(abbreviated).

Deriving this condition merely requires a simple reachability check on the~CFG.


\item[Static pointer dereference] If pointers are not modified in the future,
then their dereferences commute 
if they point to different memory locations.

For thread T1 in the pointer example in \autoref{ex:pointers},
we obtain $c_\alpha := \ccode{p1 != p2 \&\&}$ $\ccode{pc\_T2 != a,b,c}$
(here \ccode{*p++} is the pointer dereference with \ccode{p = p1}).

\item[Monotonic atomic] A CAS instruction \ccode{CAS(p, a, b)} is monotonic,
if its expected value \ccode{a} is never equal to the value \ccode{b}
that it tries to write.
Assuming that no other instructions write to the location where \ccode{p}
refers to, this means that once it is set to \ccode{b}, it never changes again.

In the hash table example in \autoref{ex:ht},
there is only a CAS instruction writing to the
array \ccode T. The dynamic moving condition is:
$c_\alpha := \ccode{T[index] != E}$.

\end{description}

\begin{lemma}\label{lem:dyncond}
The above conditions 
are dynamic both-moving conditions.
\end{lemma}


\subsection{Instrumentation}\label{sec:instrument}

\renewcommand\colorbox[2]{{\adjustbox{margin=1pt,bgcolor=#1,cframe=#1 0pt 0pt -1.5pt}{#2}}}

\autoref{fig:illustration} demonstrated how our 
instrumentation adds branches to dynamically implement
the basic single-action rule.
Lipton reduction however distinguishes between pre- and post-commit phases. 
Here, we provide an instrumentation
that satisfies the constraints on these phases
(see  L1--L4 in \autoref{sec:prelim}).
Roughly, 
we transform each CFG $G_i = (V_i, \delta_i)$ into an instrumented
$G'_i\defn (V_i', \delta'_i)$ as follows:
\begin{enumerate}[noitemsep, topsep=0pt, parsep=0pt, partopsep=0pt]
\item
Replicate all $l_a\in V_i$ to new locations in
$V_i' = \set{
\colorbox{green!30}{$l_a^N$},
\colorbox{orange!30}{$l_a^R$},
\colorbox{red!30}{$l_a^L$},
\colorbox{orange!30}{$l_a^{\Rn}$},
\colorbox{red!30}{$l_a^{\Ln}$} \mid l_a \in V_i
}$: Respectively, there are
\colorbox{green!30}{external},
\colorbox{orange!30}{pre-\vphantom{l}}, and
\colorbox{red!30}{post-} locations,
plus two auxiliary pre- and post- locations for along branches.

\item
Add edges/branches with dynamic moving conditions according to~\autoref{t:instrument}.
\end{enumerate}



\begin{table}[b]
\centering
\caption{The CFG instrumentation
}
\label{t:instrument}
\begin{tabular}{l|p{5cm}|p{6.5cm}}
\hline
&
$G_i\defn (V_i,\delta_i) $&
$V_i',\delta'$ in $G_i'$ (pictured)
\\\hline

R1& 
$\forall \tuple{l_a, \alpha, l_b} \in \delta_i
\colon$ &
\raisebox{-.5\height}{
\begin{tikzpicture}\centering
  \tikzstyle{e}=[->]
  \tikzstyle{every node}=[inner sep=0,outer sep=1.5pt,font=\small, node distance=1.5cm]
  \node (x) at (-2.4,0) [anchor=east] {\hphantom{\colorbox{red!30}{$l_a^L$}}}; 
  \node (r) at (0,0) [anchor=east] {\colorbox{green!30}{$l_a^{N}$}};
  \node (rr) at (2.5, .3) {\colorbox{orange!30}{$l_b^{R}$}};
  \node (rl) at (2.5,-.3) {\colorbox{red!30}{$l_b^{L}$}};
  \path (r) edge[e,bend right=5] node[above,sloped] {$c_\alpha\lrestr\alpha$} (rr);
  \path (r) edge[e,bend left =5] node[below,sloped] {$\lnot c_\alpha\lrestr\alpha$} (rl);
\end{tikzpicture}  \hspace{-1em}
}\\\hline

R2&
$\forall \tuple{l_a, \alpha, l_b} \in \delta_i
\colon$ &
\raisebox{-.5\height}{
\begin{tikzpicture}\centering
  \tikzstyle{e}=[->]
  \tikzstyle{every node}=[inner sep=0,outer sep=1.5pt,font=\small, node distance=1.5cm]
  \node (x) at (-2.4,0) [anchor=east] {\hphantom{\colorbox{red!30}{$l_a^L$}}}; 
  \node (r) at (0,0) [anchor=east] {\colorbox{orange!30}{$l_a^{\Rn}$}};
  \node (rr) at (2.5, .3) {\colorbox{orange!30}{$l_b^{R}$}};
  \node (rl) at (2.5,-.3) {\colorbox{red!30}{$l_b^{L}$}};
  \path (r) edge[e,bend right=5] node[above,sloped] {$c_\alpha\lrestr\alpha$} (rr);
  \path (r) edge[e,bend left =5] node[below,sloped] {$\lnot c_\alpha\lrestr\alpha$} (rl);
\end{tikzpicture}  \hspace{-1em}
}\\\hline

R3&
$\forall l_a \in V_i \colon$
& \hspace{-1em}
\raisebox{-.5\height}{
\begin{tikzpicture}\centering
  \tikzstyle{e}=[->]
  \tikzstyle{every node}=[font=\small, node distance=1.5cm]

  \node (l) at (0,0) [anchor=east] {\colorbox{orange!30}{$l_a^R$}};
  \node (n) at (2.2, 0) {\colorbox{orange!30}{$l_a^{\Rn}$}};
  \path (l) edge[e] node[above,sloped] {\true} (n);
\end{tikzpicture}
}
\\\hline

R4&
$\forall l_a \in V_i \setminus \LFS_i \colon$
& \hspace{-1em}
\raisebox{-.5\height}{
\begin{tikzpicture}\centering
  \tikzstyle{e}=[->]
  \tikzstyle{every node}=[font=\small, node distance=1.5cm]

  \node (l) at (0,0) [anchor=east] {\colorbox{red!30}{$l_a^L$}};
  \node (ll) at (2.2, .3) {\colorbox{red!30}{$l_a^{\Ln}$}};
  \node (ln) at (2.2,-.3) {\colorbox{green!30}{$l_a^{N}$}};
  \path (l) edge[e,bend left=-5] node[above left,sloped] {$c({l_a})$} (ll);
  \path (l) edge[e,bend right=-5] node[below left,sloped] {$\lnot c({l_a})$} (ln);

  \node (l) at (4.1, 0) {\parbox{2.5cm}{
	with~$c({l_a})\defn $\\
	$\bigwedge_{(l_a,\alpha,l_{b})\in \delta_i} c_{\alpha}$
}};
\end{tikzpicture} \hspace{-3em}
}
\\\hline

R5&
$\forall \tuple{l_a, \alpha, l_b} \in \delta_i,\allowbreak l_a\in V_i \setminus \LFS_i\colon$
&  \hspace{-1em}
\raisebox{-.5\height}{
\begin{tikzpicture}\centering
  \tikzstyle{e}=[->]
  \tikzstyle{every node}=[font=\small, node distance=1.5cm]

  \node (x) at (0,0) [anchor=east] {\hphantom{\colorbox{red!30}{$l_a^L$}}}; 
  \node (lp) at (2.2,0)  {\colorbox{red!30}{$l_a^{\Ln}$}};
  \node (l) at (4.95, 0) {\colorbox{red!30}{$l_b^{L}$}};
  \path (lp) edge[e] node[above,sloped] {$\alpha$} (l);
\end{tikzpicture}  \hspace{-1em}
}
\\\hline

R6&
$\forall l_a \in \LFS_i \colon$
& \hspace{-1em}
\raisebox{-.5\height}{
\begin{tikzpicture}\centering
  \tikzstyle{e}=[->]
  \tikzstyle{every node}=[font=\small, node distance=1.5cm]

  \node (l) at (0,0) [anchor=east] {\colorbox{red!30}{$l_a^L$}};
  \node (n) at (2.2, 0) {\colorbox{green!30}{$l_a^{N}$}};
  \path (l) edge[e] node[above,sloped] {\true} (n);
\end{tikzpicture}
}
\\\hline
\end{tabular}
\end{table}

The rules in \autoref{t:instrument} precisely describe the
instrumented edges in  $G'_i$:
for each graph part in the original $G_i$ (middle column),
the resulting parts of $G'_i$ are shown (right column).
As no non-movers are allowed in the post phase, R4 only checks the dynamic moving condition for all outgoing 
transitions of a post-location \colorbox{red!30}{$l_a^{L}$}.
If it fails, the branch goes to an external location \colorbox{green!30}{$l_a^{N}$}
from where the actual action can be executed (R1).
If it succeeds, then the action commutes and can safely
be executed while remaining in the post phase (R5).
We do this from an intermediary post location \colorbox{red!30}{$l_a^{\Ln}$}.
Since transitions $\alpha$ thus need to be split up into two steps in the post 
phase, dummy steps need to be introduced in the pre phase (R1 and R2)
to match this (R3),
otherwise we lose bisimilarity (see subsequent subsection).
As an intermediary pre location, we use \colorbox{orange!30}{$l_a^{\Rn}$}.

All new paths in the instrumented $G_i'$ adhere to the pattern:\\
$\colorbox{green!30}{$l^N_{1}$}\tr{\alpha_1}{}
 \colorbox{orange!30}{$l^R_{2}$}\dots
 \colorbox{orange!30}{$l^R_{k}$}\tr{\alpha_k}{}
 \colorbox{red!30}{$l^L_{k+1}$} \dots
 \colorbox{red!30}{$l^L_{n}$} \tr{\alpha_n}{}
 \colorbox{green!30}{$l^N_{n+1}$}$.
Moreover, using
the notion of \concept{location feedback sets} (LFS) defined in~\autoref{def:lfs},
R4 and R6 ensure that
all cycles in the post phase contain an external state.
This is because our reduction theorem (introduced later) allows
non-terminating transactions as long as they remain in the pre-commit phase
(it thus generalizes~L3).
\autoref{ex:instrument} shows a simple example CFG with its instrumentation.
The subsequent reduction will completely hide the internal states,
avoiding exponential blowup in the TS (see \autoref{sec:reduction}).

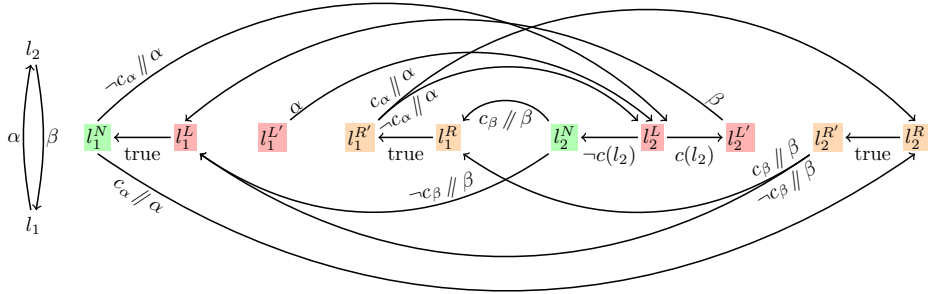
\begin{figure}[b]
\hspace{-1em}
\adjustbox{width=1.03\textwidth,keepaspectratio}{

%
%
\tikzstyle{every node}=[inner sep=0,outer sep=1,font=\normalsize, node distance=1.5cm]
\tikzstyle{proc}=[inner sep=2pt]
\tikzstyle{e}=[->,line width=.75pt]
 
\scriptsize
\begin{tikzpicture}

\newcommand\lfslocs{(0,0)/1/LFS}
\newcommand\nlfslocs{(8,0)/2/N}
\xdef\cfglocs{\nlfslocs,\lfslocs}

\newcommand\redge{1/2//\alpha/B}
\newcommand\ledge{2/1//\beta/L}

\xdef\cfgedge{\redge,\ledge}

    \foreach \pos/\name/\t in \cfglocs
        \node[proc] (\name) at \pos  {};

    \foreach \pos/\name/\t in \cfglocs {
       \ifthenelse{\equal{\name}{1}}{
        \node[proc] (X\name) at (0,-1.5)  {$l_\name$};
        }{
        \node[proc] (X\name) at (0, 1.5) {$l_\name$};
        }
       }

	\foreach \s/\d/\p/\a/\t in \cfgedge {
    	\draw[e] (X\s) to[bend left=10] node[midway,sloped,above,outer sep=4pt] (x\s) {} (X\d);
		\node at (x\s) {$\a$};
	}

    \foreach \pos/\name/\t in \cfglocs
        \node[proc] (N\name) at  ($ (\name)+(1.1,0) $) {\colorbox{green!30}{$l_\name^N$}};

    \foreach \pos/\name/\t in \cfglocs
        \node[proc] (L\name) at  ($ (\name)+(2.6,0) $) {\colorbox{red!30}{$l_\name^L$}};

    \foreach \pos/\name/\t in \cfglocs
        \node[proc] (Ln\name) at  ($ (\name)+(4.1,0) $) {\colorbox{red!30}{$l_\name^{\Ln}$}};

    \foreach \pos/\name/\t in \cfglocs
        \node[proc] (R\name) at  ($  (\name)+(7.1,0) $) {\colorbox{orange!30}{$l_\name^R$}};

    \foreach \pos/\name/\t in \cfglocs
        \node[proc] (Rn\name) at  ($ (\name)+(5.6,0) $) {\colorbox{orange!30}{$l_\name^{\Rn}$}};

	\foreach \s/\d/\p/\a/\t in \cfgedge {
       \ifthenelse{\equal{\s}{1}}{
       \draw[e] (N\s.south) to[\p,bend left=-35] node[pos=.06,sloped,below] (x\s) {$c_\a\lrestr\a$} (R\d.south);
       }{
    	\draw[e] (N\s.north west) to[\p,bend left=-55] node[pos=.5,sloped,below,yshift=-3] (x\s) {$c_\a\lrestr\a$} (R\d.north east);
	   }
	  \ifthenelse{\equal{\s}{1}}{
    	\draw[e] (N\s.north) to[\p,bend left=45] node[pos=.09,sloped,above] (x\s) {$\neg c_\a\lrestr\a$} (L\d.north east);
       }{
    	\draw[e] (N\s.south west) to[\p,bend left=35] node[pos=.3,sloped,above] (x\s) {$\neg c_\a\lrestr\a$} (L\d.south east);
	   }
	  \ifthenelse{\equal{\s}{1}}{
    	\draw[e] (Rn\s.north east) to[\p,bend left=45] node[pos=.05,sloped,above] (x\s) {$c_\a\lrestr\a$} (R\d.north);
       }{
    	\draw[e] (Rn\s.south west) to[\p,bend left=35] node[pos=.08,sloped,above] (x\s) {$c_\a\lrestr\a$} (R\d.south east);
	   }
	  \ifthenelse{\equal{\s}{1}}{
    	\draw[e] (Rn\s.north east) to[\p,bend left=45] node[pos=.1,sloped,below] (x\s) {$\neg c_\a\lrestr\a$} (L\d.north west);
		}{
    	\draw[e] (Rn\s.south west) to[\p,bend left=35] node[pos=.05,sloped,below] (x\s) {$\neg c_\a\lrestr\a$} (L\d.south east);
		}
	}

    \foreach \pos/\name/\t in \cfglocs
        \draw[e] (R\name) edge[bend left=0] node[midway,auto,yshift=-4] {\true} (Rn\name);

    \foreach \pos/\name/\t in \nlfslocs { 
   		\draw[e] (L\name) edge[bend right=0] node[pos=.5,sloped,below,yshift=-3] {$c(l_\name)$} (Ln\name);
		\draw[e] (L\name) edge[bend left=0] node[pos=.5,sloped,below,yshift=-3] {$\neg c(l_\name)$} (N\name);
	}

	\foreach \s/\d/\p/\a/\t in \cfgedge {
       \ifthenelse{\equal{\t}{LFS}}{}{ 
       \ifthenelse{\equal{\s}{1}}{
	       	\draw[e] (Ln\s.north east) to[\p,bend left=40] node[pos=.03,sloped,above] (x\s) {$\a$} (L\d.north);      
       	}{
    		\draw[e] (Ln\s.north west) to[\p,bend left=-40] node[pos=.03,sloped,above] (x\s) {$\a$} (L\d.north);
		}
		}
	}

    \foreach \pos/\name/\t in \lfslocs 
        \draw[e] (L\name) edge[bend left=0] node[midway,auto,yshift=-4] {\true} (N\name);

\end{tikzpicture}

}
\caption{Instrumentation (right) of a 2-location CFG (left) with $\LFS=\set{l_1}$.}
\label{ex:instrument}
\end{figure}

\begin{definition}[LFS]\label{def:lfs}
A location feedback set (LFS) for thread $i$ is a subset $\LFS_i\subseteq V_i$
such that for each cycle $C = l_1,..,l_n,l_1$ in $G_i$ it holds that
\mbox{$\LFS_i\cap C \neq \emptyset$}.
The corresponding (state) feedback set (FS) is:
$\CC_i\defn \{ \tuple{\pc, d} \mid \pc_i\in \LFS_i) \}$.
\end{definition}

\begin{corollary}[\cite{kurshan1998static}]\label{cor:sfs}
$\bigcup_i \CC_i$ is a feedback set in the TS.
\end{corollary}

The instrumentation yields the following 3/4-partition of states
for all threads~$i$:
\begin{align}
\label{def:state-error}
\WW_i
	&\defn\bigr\{ (\pc,d) \mid
				\pc_i\in \{l_{sink}^N,l_{sink}^R,l_{sink}^L \} \bigr\}
	&\text{(\colorbox{green!30}{Error})} \\
\RR_i
	&\defn \bigr\{ \tuple{\pc, d} \mid \pc_i \in
					\{l^R,l^{\Rn}\} \bigr\} \setminus \WW_i
	&\text{(\colorbox{orange!30}{Pre-commit})} \\
\LL_i
	&\defn \bigr\{	
		\tuple{\pc, d} \mid \pc_i \in \{l^{L}, l^{\Ln} \}\bigr\} \setminus \WW_i
	&\text{(\colorbox{red!30}{Post-commit})} \\
\EX_i
	&\defn \bigr\{ \tuple{\pc, d}  \mid \pc_i \in
  				\{l^{N}\}\bigr\}  \setminus \WW_i
	&\text{(\colorbox{green!30}{Ext./non-error})} \\
\label{def:state-external}
\NN_i
	&\defn \EX_i \uplus \WW_i 
	&\text{(\colorbox{green!30}{External})}
\end{align}
\noindent
The new initial state is $(\pc_0', d_0)$, with 
$\forall i \colon \pc_{0,i}' = l_{0,i}^N$.
Let $\locs' \defn \prod_i V_i'$ and $C'\defn (\locs'\times\data, \to')$
be the transition system of the instrumented CFG.
The instrumentation preserves the behavior of the original~system:
\begin{lemma}\label{lem:instrument}
~\hspace{-2.1ex}
An error state is $\to$-reachable in the original~system
iff an error state is $\to'$-reachable
in the instrumented~system.
\end{lemma}


Recall the situation illustrated in \autoref{fig:bisim-T1} within the example in \autoref{fig:illustration}. Rules R1, R2, and R4 of our instrumentation in \autoref{t:instrument} give rise to a similar problem as illustrated in the following.

\begin{equation*}
\begin{aligned}
\text{

\begin{tikzpicture}\centering

   \tikzstyle{e}=[minimum width=1cm]
   \tikzstyle{every node}=[font=\small, node distance=.5cm]

  \node (s1) {$\sigma_1$};
  \node (s2) [node distance=1.2cm,below of=s1] {$\sigma_2$};
  \node (s3) [right of=s2, xshift=.9cm] {$\sigma_3$};
  \path (s2.east) -- node[pos=.45]{$\tatrans_j$} (s3.west);
  \path (s1.south) -- node[midway,sloped]{$\tr{c_\alpha\lrestr \alpha}{}_i$} (s2.north);
  \node (s4) [xshift=.9cm,gray,right of=s1] {$\sigma_4$};
  \node (s3p) [gray,node distance=1.2cm,right of=s3] {$\sigma_3'$};
  \path (s1.east) -- node[gray,midway,sloped]{$\tatrans_j$}
  (s4.west);
  \path (s4.south east) -- node[gray,midway,sloped]{$\tr{\neg c_\alpha\lrestr \alpha}{}_{i}$} (s3p.north);

  \node (S1) [left of=s1,node distance=.5cm,e] {};
  \node (S2) [left of=s2,node distance=.5cm,e] {\scriptsize $l_a^{L}$};

  \node (S4) [right of=s3p,node distance=.5cm,e] {\scriptsize $l_a^{R}$};
  \node (S3) [above left of=s3,node distance=.5cm,e] {\scriptsize $l_a^{L}$};


  \path (s3) -- node[gray,sloped,pos=.48]{$\neq$} (s3p);
%
%
%

\end{tikzpicture}

\begin{tikzpicture}\centering

   \tikzstyle{e}=[minimum width=1cm]
   \tikzstyle{every node}=[font=\small, node distance=1cm]

  \node (s1) {$\sigma_1$};
  \node (s2) [right of=s1, xshift=.9cm] {$\sigma_2$};
  \path (s1) -- node[pos=.45]{$\tatrans_j$} (s2);

  \node (s4) [node distance=1.2cm,gray,below of=s1] {$\sigma_4$};
  \path (s1) -- node[gray,midway,sloped]{$\tr{\neg c(l_a)}{\hspace{-.5ex}i}$} (s4);
  \node (s3p) [gray,node distance=1.2cm,below of=s2] {$\sigma_3'$};
  \path (s4) -- node[gray,pos=.48]{$\tatrans_j$} (s3p);
  
  \node (s3) [node distance=1.5cm,right of=s3p] {$\sigma_3$};
  \path (s2) -- node[midway,sloped]{$\tr{c(l_a)}{\hspace{-.5ex}i}$} (s3.north);

  \node (S1) [left of=s1,node distance=.5cm,e] {\scriptsize $l_a^L$ };
  \node (S4) [left of=s4,node distance=.5cm,gray,e] {\scriptsize $l_a^{N}$};
  \node (S2) [right of=s2, node distance=.5cm,e] {\scriptsize $l_a^L$};
  \node (S3) [right of=s3,node distance=.5cm,e] {\scriptsize $l_a^{L'}$};
  \node (S4) [gray,above left of=s3p,node distance=.4cm,e] {\scriptsize $l_a^{N}$};
  \path (s3p) -- node[sloped,pos=.48]{$\neq$} (s3);

\end{tikzpicture}}
\end{aligned}
\end{equation*}

Hence, our instrumentation introduces non-movers.
Nevertheless, we can prove that the target states are bisimilar. 
This enables us to introduce a weaker notion of commutativity up to bisimilarity
which effectively will enable a reduction along one branch
(where reduction was not originally possible).
The details of the reduction are presented in the following section.
We emphasize that our implementation does not introduce any unnecessary non-movers.


\subsection{Reduction}\label{sec:reduction}

We now formally define the notion of thread bisimulation required for the reduction, as well as commutativity up to bisimilarity.


\begin{definition}[thread bisimulation]\label{def:bisim}
An equivalence relation $R$ on the states of a TS $(S, \to)$ is a thread bisimulation iff 

\centering
\begin{tikzpicture}

   \tikzstyle{e}=[minimum width=1cm]
   \tikzstyle{every node}=[font=\small, node distance=1.2cm]

  \node (s1) {$\sigma$};
  \node (s2) [below of=s1] {$\sigma'$};
  \node (s3) [right of=s1] {$\sigma_1$};

  \path (s1) edge[-] node[pos=.5,sloped,above] (nn) {$R$} (s2);
  \path (s1) -- node(m)[midway,sloped]{$\tatrans_i$} (s3);


  \node (s1p) [gray,right of=s3,xshift=1cm] {$\sigma$};
  \node (s2p) [below of=s1p] {$\sigma'$};
  \node (s3p) [right of=s1p] {$\sigma_1$};

  \path (s1p) edge[-] node[gray,pos=.5,sloped,above] (nnn) {$R$} (s2p);
  \path (s1p) -- node(m)[gray,midway,sloped]{$\tatrans_i$} (s3p);

  \node (s4p) [right of=s2p] {$\sigma'_1$};
  \path (s2p) -- node [midway,sloped]{$\tatrans_i$} (s4p);
  \path (s3p) edge[-] node[pos=.5,sloped,above]{$R$} (s4p);

  \node (n) [left of=nnn] {$\implies \exists \sigma_1'\colon$};

  \node (n) [left of=nn] {$\forall \sigma,\sigma',\sigma_1,i\colon$};
  
\end{tikzpicture}
\end{definition}

Standard bisimulation~\cite{Milner,Park} is an equivalence relation $R$ which satisfies the property 
from \autoref{def:bisim} when the indexes $i$ of the transitions are removed. 
Hence, in a thread bisimulation, in contrast to standard bisimulation, the transitions performed by 
thread $i$ will be matched by transitions performed by the same thread $i$.
As we only make use of thread bisimulations, we will often refer to them simply as
bisimulations.

\begin{definition}[commutativity up to bisimulation] \label{def:comm-bisim}
Let $R$ be a thread bisimulation on a TS $(S, \to)$.
The right and left commutativity up to $R$ of the transition relation $\to_i$
with $\to_j$, notation
$\to_i\, \rcomm_{R} \to_j$ /$\to_i\, \lcomm_{R} \to_j$ are defined as follows.

\noindent
$\to_i\, \rcomm_{R} \to_j\,\,\,\, \Longleftrightarrow$~~~~~~~~~~~~~~~~~~~~~~~~~~~~
$\to_i\, \lcomm_{R} \to_j\,\,\,\, \Longleftrightarrow$\\
\begin{tikzpicture}
   \tikzstyle{e}=[minimum width=1cm]
   \tikzstyle{every node}=[font=\small, node distance=.45cm]

  \node (s1) {$\sigma_1$};
  \node (s2) [node distance=1.2cm,below of=s1] {$\sigma_2$};
  \node (s3) [right of=s2, xshift=.6cm] {$\sigma_3$};
  \path (s2) -- node[pos=.45]{$\tatrans_j$} (s3);
  \path (s1) -- node(m)[midway,sloped]{$\tatrans_i$} (s2);

  \node (s1n) [xshift=2.5cm,right of=s1] {$\sigma_1$};
  \node (s2n) [gray,node distance=1.2cm,below of=s1n] {$\sigma_2$};
  \node (s3n) [gray,right of=s2n, xshift=.6cm] {$\sigma_3$};
  \path (s2n) -- node[gray,pos=.45]{$\tatrans_j$} (s3n);
  \path (s1n) -- node(m)[gray,midway,sloped]{$\tatrans_i$} (s2n);

  \node (s0n) [left of=m, xshift=-.4cm]  {$\implies \exists \sigma_3',\sigma_4\colon$};

  \node (s4n) [xshift=.9cm,right of=s1n] {$\sigma_4$};
  \path (s1n) -- node [midway,sloped]{$\tatrans_j$} (s4n);

  \node (s3pn) [node distance=1.1cm,right of=s3n] {$\sigma_3'$};
  \path (s4n) -- node[midway,sloped]{$\tatrans_i$} (s3pn);

  \path (s3pn) -- node(AA)[sloped,pos=.1]{} (s3n);

  \node (s3pn) [node distance=1em,below of=AA,xshift=-.3cm] {$\tuple{\sigma_3,\sigma_3'}\in R$};
  
\end{tikzpicture}
~~
\begin{tikzpicture}

   \tikzstyle{e}=[minimum width=1cm]
   \tikzstyle{every node}=[font=\small, node distance=.45cm]

  \node (s1) {$\sigma_1$};
  \node (s2) [right of=s1, xshift=.6cm] {$\sigma_2$};
  \node (s3) [node distance=1.2cm,below of=s2, xshift=.9cm] {$\sigma_3$};
  \path (s2) -- node(m)[midway,sloped]{$\tatrans_j$} (s3);
  \path (s1) -- node[midway,sloped]{$\tatrans_i$} (s2);

  \node (s1n) [xshift=2.2cm,right of=s2] {$\sigma_1$};
  \node (s2n) [gray,right of=s1n, xshift=.6cm] {$\sigma_2$};
  \node (s3n) [gray,node distance=1.2cm,below of=s2n, xshift=.9cm] {$\sigma_3$};
  \path (s2n) -- node[gray,midway,sloped]{$\tatrans_j$} (s3n);
  \path (s1n) -- node[gray,midway,sloped]{$\tatrans_i$} (s2n);

  \node (s4n) [node distance=1.2cm,below of=s1n] {$\sigma_4$};
  \path (s1n) -- node(m) [midway,sloped]{$\tatrans_j$} (s4n);

  \node (s3pn) [node distance=.9cm,right of=s4n] {$\sigma_3'$};
  \path (s4n) -- node[midway,sloped]{$\tatrans_i$} (s3pn);

  \node (s0n) [left of=m, xshift=-.4cm]  {$\implies \exists \sigma_3',\sigma_4\colon$};
  
  \path (s3pn) -- node(AA)[sloped,pos=.1]{} (s3n);

  \node (s3pn) [node distance=1em,below of=AA] {$\tuple{\sigma_3,\sigma_3'}\in R$};
  
\end{tikzpicture}
%
%
%
%
%
\end{definition}

Our reduction works on parallel transaction systems (\PAS), a specialized TS.
While its definition (\autoref{def:pas2}) looks complicated, most rules
are concerned with ensuring that all paths in the underlying TS
form transactions, i.e. that they conform to the pattern
$\colorbox{green!30}{$\sigma_{1}$}\tr{\alpha_1}{}\colorbox{orange!30}{$\sigma_{2}$}\dots\colorbox{orange!30}{$\sigma_{k}$}\tr{\alpha_k}{}\colorbox{red!30}{$\sigma_{k+1}$} \dots\colorbox{red!30}{$\sigma_{n}$} \tr{\alpha_n}{}\colorbox{green!30}{$\sigma_{n+1}$}$, where  $\alpha_k$ is the non-mover, etc. 
We have from the perspective of thread $i$ that:
$\sigma_1$ and $\sigma_{n+1}$ are \colorbox{green!30}{\concept{external}}, 
$\forall 1< x \le k \colon \sigma_{x}$ \colorbox{orange!30}{\concept{pre-commit}}, 
and $\forall k< x \le n \colon \sigma_{x}$ \colorbox{red!30}{\concept{post-commit}} states. The rest of the conditions ensure bisimilarity and constrain error locations.

The reduction theorem, \autoref{th:blockred}, then tells us that reachability of error states is preserved (and reflected)
if we consider only \PAS-paths between globally external states \NN.
The reduction thus removes all internal states $\II$ where $\II \defn \bigcup_i \II_i$
and
$\II_i \defn \LL_i\cup \RR_i$
(at least one internal phase). 

\begin{definition}[transaction system]
\label{def:pas2}
A parallel transaction system $\PAS$
is a transition system $\TS=\tuple{S, \tatrans}$
whose states are partitioned in three sets of phases and error states in one of
the phases, for each thread $i$. For each thread $i$, there exists a thread bisimulation 
relation $\cong_i$.
Additionally, the following properties hold
(for all $i$, all $j \neq i$):
\begin{enumerate}
\item\label{item2:part}
	$S = \RR_i \uplus \LL_i \uplus \NN_i$ and $\NN_i = \WW_i \uplus \EX_i$
	\hfill
	(the 3/4-partition)

\item \label{item2:postterm}
	$\forall \sigma \in \LL_i \colon \exists \sigma'\in\NN_i \colon
	\sigma \tatrans_i^+ \sigma'$
	 \hfill
	\hspace{-1em}(post phases terminate)


\item\label{item2:invar}
	$\tatrans_i\,\subseteq \LL_j^2\cup{\RR}_j^2\cup{\WW}_j^2
	\cup \EX_j^2 $
	\hfill
	($i$ preserves $j$'s phase)

\item\label{item2:errors}
	$\WW_i \lrestr \tatrans_i \rrestr \overline{\WW_i}=\emptyset$
	\hfill
	(local transitions preserve errors)

\item \label{item2:post}
	${\LL}_i\lrestr \tatrans_i\rrestr {\RR}_i=\emptyset$
\hfill
	((locally) post does not reach pre)

\item\label{item2:errors2}
	$\cong_{i}\,\subseteq \WW_i^2\cup \overline{\WW_i}^2$
	\hfill
	(bisimulation preserves (non)errors)

\item\label{item2:bisim-threads}
	$\cong_i\,\subseteq {\LL}_j^2\cup{\RR}_j^2\cup{\WW}_j^2\cup \EX_j^2 $
	\hfill
	($\cong_i$ entails $j$-phase-equality)
	

\item \label{item2:rmover}
	$(\tatrans_i \rrestr \RR_{i})  \rcomm_{\set j}\, \tatrans_j$
		 \hfill
	($i$ to pre right commutes up to $\cong_j$ with~$j$)

\item \label{item2:lmover}\hspace{.0mm}$
	  (\LL_i \lrestr \tatrans_i) \lcomm_{\set{i,j}}\, \tatrans_j$
		 \hfill
	($i$ from post left~commutes up to $\cong_{\{i,j\}}$~with~$j$)	
	
\end{enumerate}
In \autoref{item2:rmover} and \autoref{item2:lmover}, $\rcomm_{Z}$ and $\lcomm_{Z}$ (for a set of 
threads $Z$) are short notations for $\rcomm_{\cong_Z}$ and $\lcomm_{\cong_Z}$, respectively,
with $\cong_Z$ being the transitive closure of the union of all $\cong_i$ for $i \in Z$. 
\end{definition}

%

\begin{theorem}
\label{th:blockred}
The block-reduced transition relation $\brtrans$ of a parallel transaction system
$\PAS=\tuple{S, \tatrans}$ 
is defined in two steps:
\begin{align*}
\trtrans_i &\defn \NN_{j\neq i} \lrestr \tatrans_i
&\text{($i$ only transits when all $j$ are in external)}\\
\brtrans_i &\defn \NN_i\lrestr 
	(\trtrans_i\rrestr \overline{\NN_i})^* \trtrans_i \rrestr \NN_i
&\text{(block steps skip internal states $\overline{\NN_i}$)}
\end{align*}
Let $\brtrans \defn \bigcup_i \brtrans_i$, $\NN \defn \bigcap_{i} \NN_i$ and $\WW \defn \bigcup_{i} \WW_i$.
We have $p\tatrans^{*} q$ for $p\in \NN$ and $q\in \WW$ if and only if
$p \brtrans^{*} q'$ for $q'\in \WW$.
\end{theorem}

Our instrumentation  from \autoref{t:instrument} in~\autoref{sec:instrument}
indeed gives rise to a \PAS (\autoref{lem:instrument2}) with the state partitioning from (\autoref{def:state-error}--\ref{def:state-external}).
The following equivalence relation $\sim_i$ over locations
becomes the needed bisimulation $\cong_i$ when lifted to states.
(The locations in the rightmost column of \autoref{t:instrument} are intentionally positioned such that vertically aligned locations are bisimilar.)
\begin{align*}
\sim_i\defn&\set{ (l^X,l^Y)\mid l\in V_i\land X,Y\in\set{L,R}}\cup
 			  \set{(l^{X},l^{Y})\mid l\in V_i \land
					X,Y\in\set{N,\Rn,\Ln}}\\
\cong_i \defn&\set{ ((\pc,\hspace{-.3ex}d),\hspace{-.5ex}(\pc',\hspace{-.3ex}d'))\hspace{-.5ex} \mid\hspace{-.5ex}
    	 d\hspace{-.5ex}=\hspace{-.5ex}d'
    \land \pc_i\hspace{-.5ex}\sim_i\pc'_i
    \land \forall j\neq i\colon\hspace{-.5ex} \pc_j\hspace{-.5ex}=\hspace{-.5ex}\pc_j' }
\end{align*}
The dynamic both-moving condition in \autoref{def:dynamicboth} is sufficient to prove (\autoref{item2:rmover}--\ref{item2:lmover}). The LFS notion in
\autoref{def:lfs} is sufficient to prove 
post-phase termination (\autoref{item2:postterm}).

\begin{lemma}\label{lem:instrument2}
The instrumented TS $C' = (\locs'\times\data, \to')$ is a \PAS.
\end{lemma}

All of the apparent exponential blowup of the added phases
($5^{\sizeof{\threads}}$)
is hidden by the reduction
as \brtrans only reveals external states $\NN\defn \bigcap_{i} \NN_i$
(note that $S = \II \uplus \NN$)
and there is only one instrumented external location
(replicated sinks can be eliminated easily with a more verbose instrumentation).

\section{Block Encoding of Transition Relations}
\label{sec:encoding}

We implement the reduction by encoding a transition relation 
for symbolic model checking.
%
Transitions encoded as SMT formulas may not contain cycles.
Although our instrumentation conservatively eliminates cycles in the post-commit 
phase of transactions with external states,
cycles (without external locations) can still occur in the pre-phase.
To break these remaining cycles, we use a refined location feedback set
$\LFS_i'$ of the 
instrumented CFG without external locations $G_i' \setminus \set{l^N\in V_i'}$
(this also removes edges incident to external locations).

Now, we can construct a new block-reduced relation $\twoheadrightarrow$.
It resembles the definition of
\brtrans in \autoref{th:blockred},
except for the fact that the execution of thread $i$ can be interrupted in an 
internal state $\CC_i'$ ($\LFS_i'$ lifted to states) in order to break the remaining cycles.
\[
\twoheadrightarrow \defn\hspace{-.5ex} \bigcup_i\hspace{-.5ex}\twoheadrightarrow_i\text{, where }
	\twoheadrightarrow_i \defn\hspace{-.5ex} \mathcal{X}_i\lrestr 
	(\trtrans_i\rrestr \overline{\mathcal{X}_i})^* \trtrans_i \rrestr \mathcal{X}_i
	\text{~~with~~}
	\mathcal{X}_i\hspace{-.5ex}\defn\hspace{-.5ex} \NN_i\cup \CC_i'
\]
Here, the use of $\trtrans_i$ (from \autoref{th:blockred})
warrants that only thread $i$ can transit from the newly exposed
internal states $\CC_i' \subseteq \NN_{j\neq i}$.
Therefore, by carefully selecting the exposed locations of $\CC_i'$,
e.g. only \colorbox{orange!30}{$l_a^{R}$}, the overhead is limited to a
factor~two.


To encode $\twoheadrightarrow$,
we identify blocks of paths that start and end in external or LFS locations, but do not traverse external or LFS locations and encode them using large blocks~\cite{BeyerCGKS09}.
This automatically takes care of disallowing intermediate states, except
for the states $\CC_i'$ exposed by the breaking of cycles.
At the corresponding locations, we thus add constraints to the
scheduler encoding to only allow the current thread to execute.
To support \texttt{pthreads} constructs, such as locks and thread management,
we use similar scheduling mechanisms.

\section{Experiments}
\label{sec:experiments}
We implemented the encoding with dynamic reduction in the
Vienna Verification Tool (VVT)~\cite{vvt,vvtwebsite}.
VVT implements CTIGAR~\cite{CTIGAR}, an IC3~\cite{ic3}
algorithm with predicate-abstraction, and bounded model checking (BMC)~\cite{GuntherW14}.
VVT came fourth in the concurrency category of
SVComp 2016~\cite{Beyer2016} the first
year it participated, only surpassed by tools based on BMC or 
symbolic simulation.

We evaluated our dynamic reductions on the running examples and
compared the running time of the following configurations:
\begin{itemize}
\item BMC with all dynamic reductions (\emph{BMC-dyn} in the graphs);
\item BMC with only static reductions and phase variables from~\cite{qadeer-transactions} (\emph{BMC-phase});
\item IC3 with all dynamic reductions (\emph{IC3-dyn}); and
\item IC3 with only static reductions and phase variables from~\cite{qadeer-transactions} (\emph{IC3-phase}).
\end{itemize}

We used a time limit of one hour for each run and ran each
instance 4 times to even out results of non-determinism.
Variation over the 4 runs was insignificant, so we omit plotting it.
Missing values in the graphs indicate a timeout.
\paragraph{Hashtable}
The lockless hash table of \autoref{ex:ht} is used together 
with the monotonic atomic heuristic in three
experiments:
\begin{enumerate}
\item Each thread inserts one element into an empty hash table. The
  verification condition is that all inserted elements are present in
  the hash table after all threads have finished executing. We can see
  in \autoref{fig:hashtable-empty} that the dynamic reduction benefits
  neither BMC nor IC3. This is because every thread changes the hash
  table thus forcing an exploration of all interleavings.

  The overhead of using dynamic reductions, while significant in the
  BMC case, seems to be non-existent in IC3.
\item Every thread attempts to insert an already-present element into
  the hash table. The verification condition is that every
  \emph{find-or-put} operation reports that the element is already
  present.

  Since a successful lookup operation doesn't change the hash table,
  the dynamic reduction now takes full effect: While the static
  reduction can only verify two threads for BMC and four for IC3, the
  dynamic reduction can handle six threads for BMC and more than seven
  for IC3.
\item Since both of the previous cases can be considered corner-cases
  (the hash table being either empty or full), this
  configuration has half of the threads inserting values already present
  while the other half insert new values.

  While the difference between static and dynamic reductions is not as
  extreme as before, we can still see that we profit from dynamic
  reductions, being able to verify two more threads in the IC3 case.
\end{enumerate}

\begin{figure}[t]
  \centering\vspace{-1.7em}
  \subfloat[Table inserts\label{fig:hashtable-empty}]{\includegraphics[width=.5\linewidth]{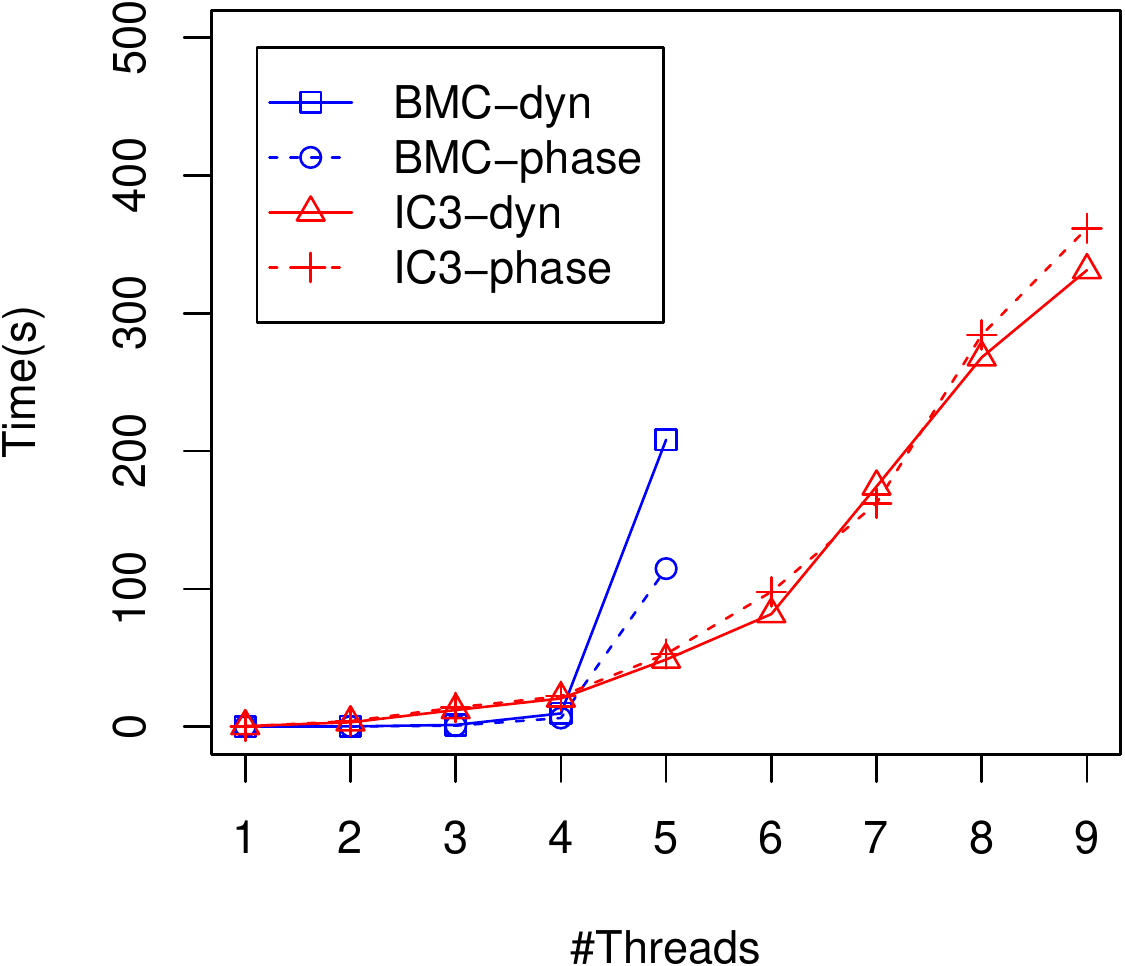}}~
  \subfloat[Table lookups\label{fig:hashtable-full}]{\includegraphics[width=.5\linewidth]{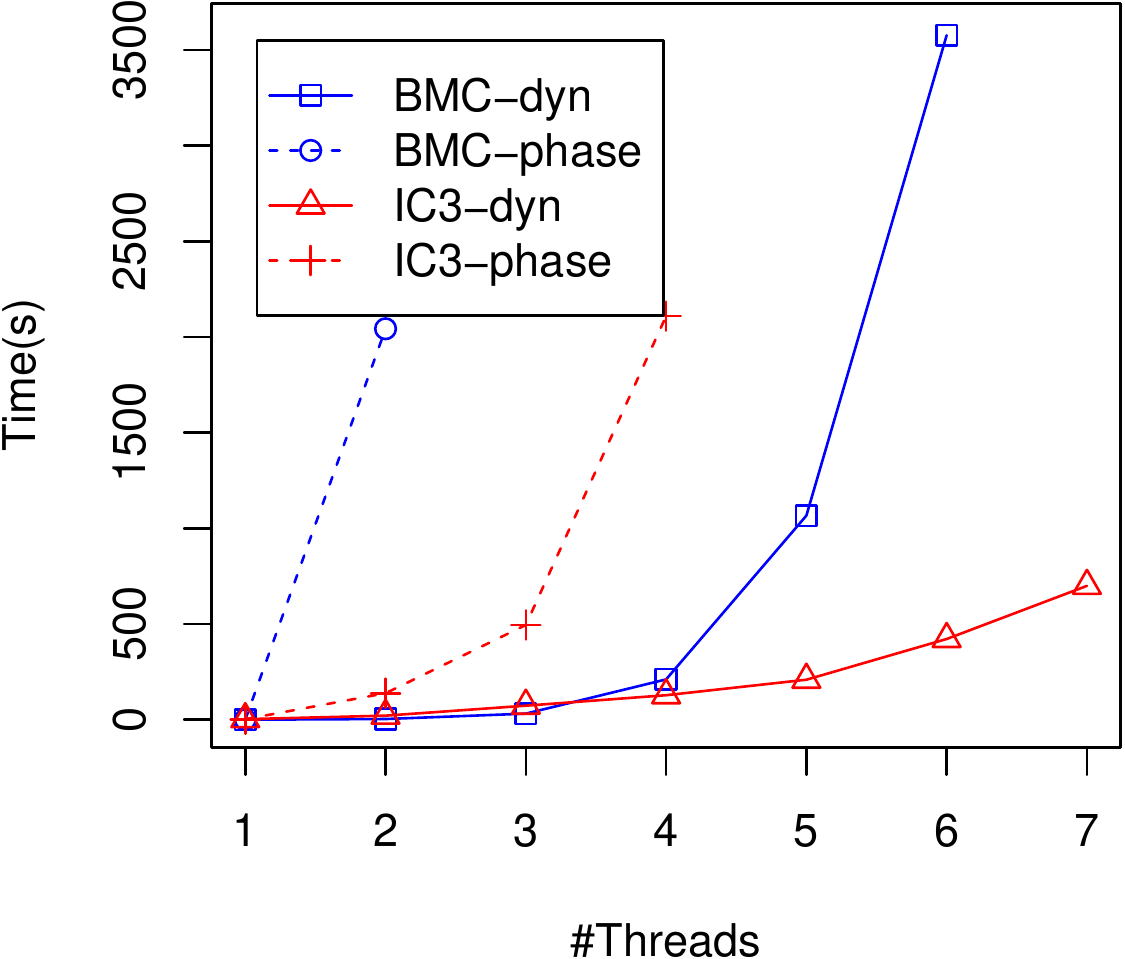}} 

  \subfloat[Table inserts / lookups\label{fig:hashtable-mixed}]{\includegraphics[width=.5\linewidth]{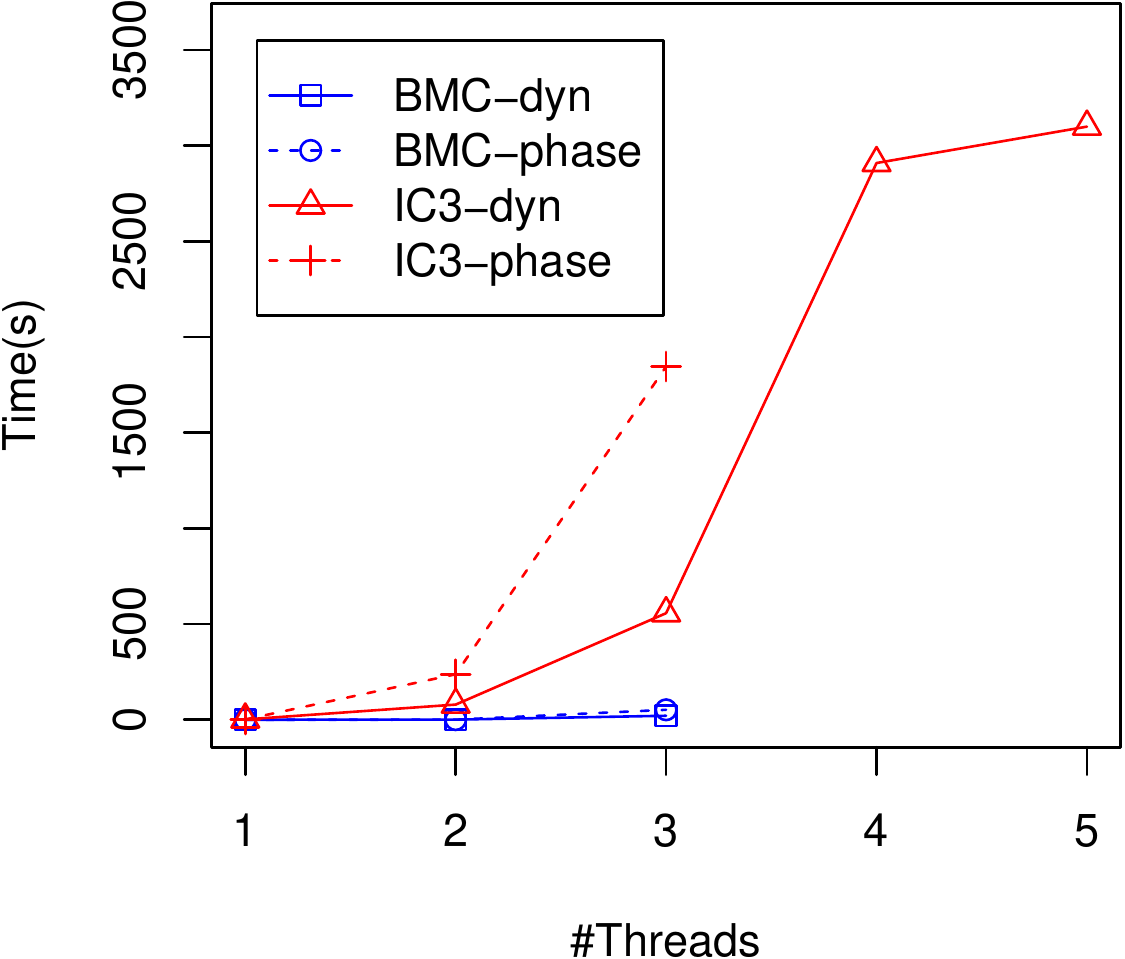}}~
  \subfloat[Dynamic locking\label{fig:dynlock}]{\includegraphics[width=.5\linewidth]{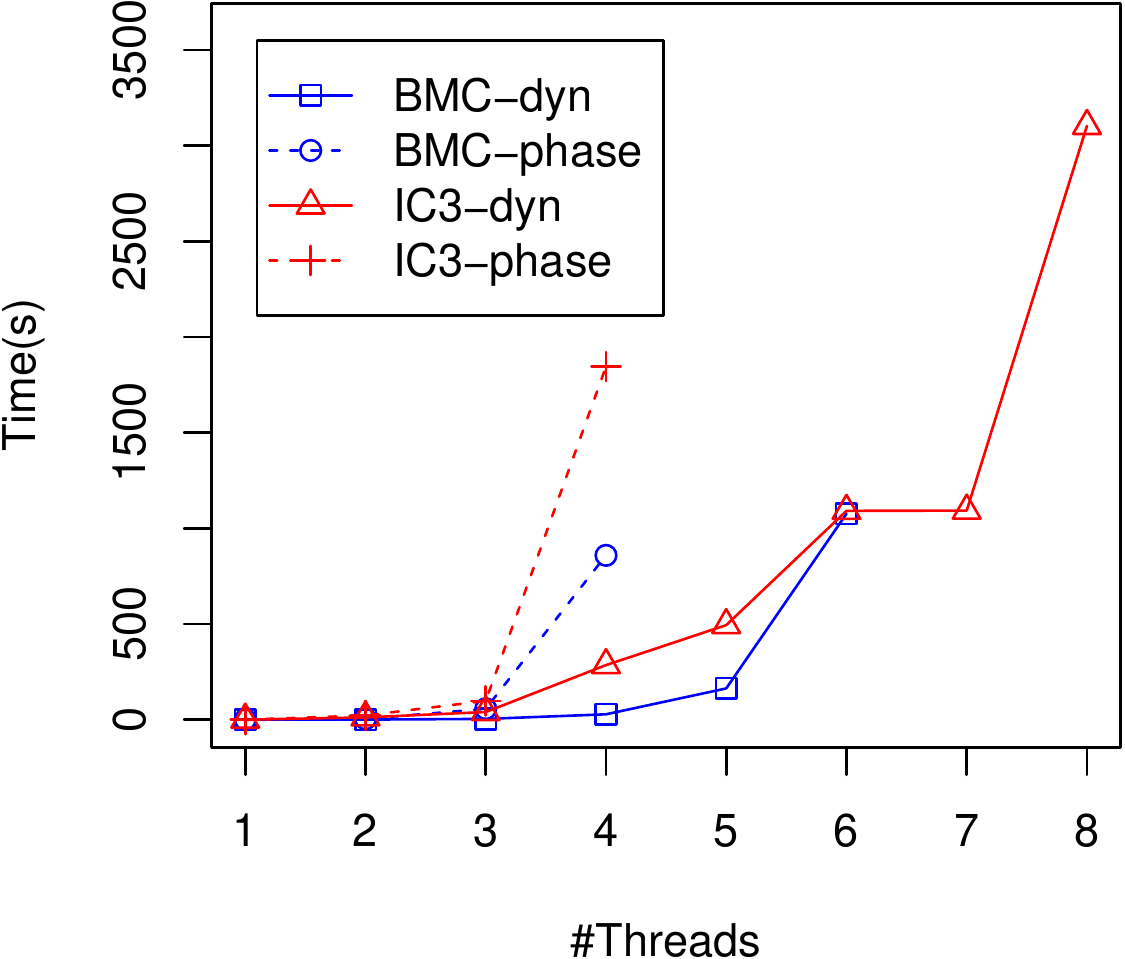}}
  \caption{Hash table and dynamic locking benchmark results}
\end{figure}

\paragraph{Dynamic locking}
To study the effect of lock pointer analysis, we use a benchmark in
which multiple threads use a common lock to access two global
variables. The single lock these threads use is randomly picked
from a set of locks at the beginning of the program.
Each thread writes the same value four times
to both global variables. Because all threads use the same lock, after
all threads terminate, the value of both global variables must be the
same, which is the verification condition.
We use our static pointer heuristic to determine that all global variable
accesses are protected by the same lock, potentially allowing
the critical section to become a single transaction.
\autoref{fig:dynlock}
shows that the dynamic reduction indeed kicks in and benefits both IC3 and BMC.
\paragraph{Load balancing}
We use the load-balancing example (\autoref{ex:pointers}). 
It relies on the static pointer heuristic.
We verified that the computed sum of the counters is indeed the expected result.
Our experiment revealed that dynamic reductions reduce the 
runtime from 15 minutes to 97 seconds for two threads already.


\paragraph{SVComp}
We also ran our IC3 implementation on the \texttt{pthread-ext} and
\texttt{pthread-atomic} categories of the \emph{software verification
  competition} (SVComp) benchmarks~\cite{svcomp16,svcomp}. In instances
with an unbounded number of threads, we introduced a limit of three
threads.
To check the effect of different reduction-strategies on the verification time, we tested the following reductions:
\begin{description}
\item[\emph{dyn}:] Dynamic with all heuristics from \autoref{tab:heuristics}.
\item[\emph{phase}:] Dynamic phases only (equal to \cite{qadeer-transactions}).
\item[\emph{static}:] Static (as in \autoref{sec:prelim}).
\item[\emph{nored}:] No reduction, all interleavings considered.
\end{description}

\autoref{fig:svcomp} shows that static Lipton reduction yields an
average six-fold decrease in runtime when compared to no
reduction. Enabling the various dynamic improvements (\textit{dyn}, \textit{phase}) does not show
improvement over the static case (\textit{static}), since most of the benchmarks are either too small or do
not have opportunities for reductions, but also not
much overhead (up to $7\%$).  Comparing the \textit{nored} case with the other cases shows the
benefit of removing intermediate states. 

\begin{figure}[t]
  \centering\vspace{-0em}
  \includegraphics[width=.75\linewidth]{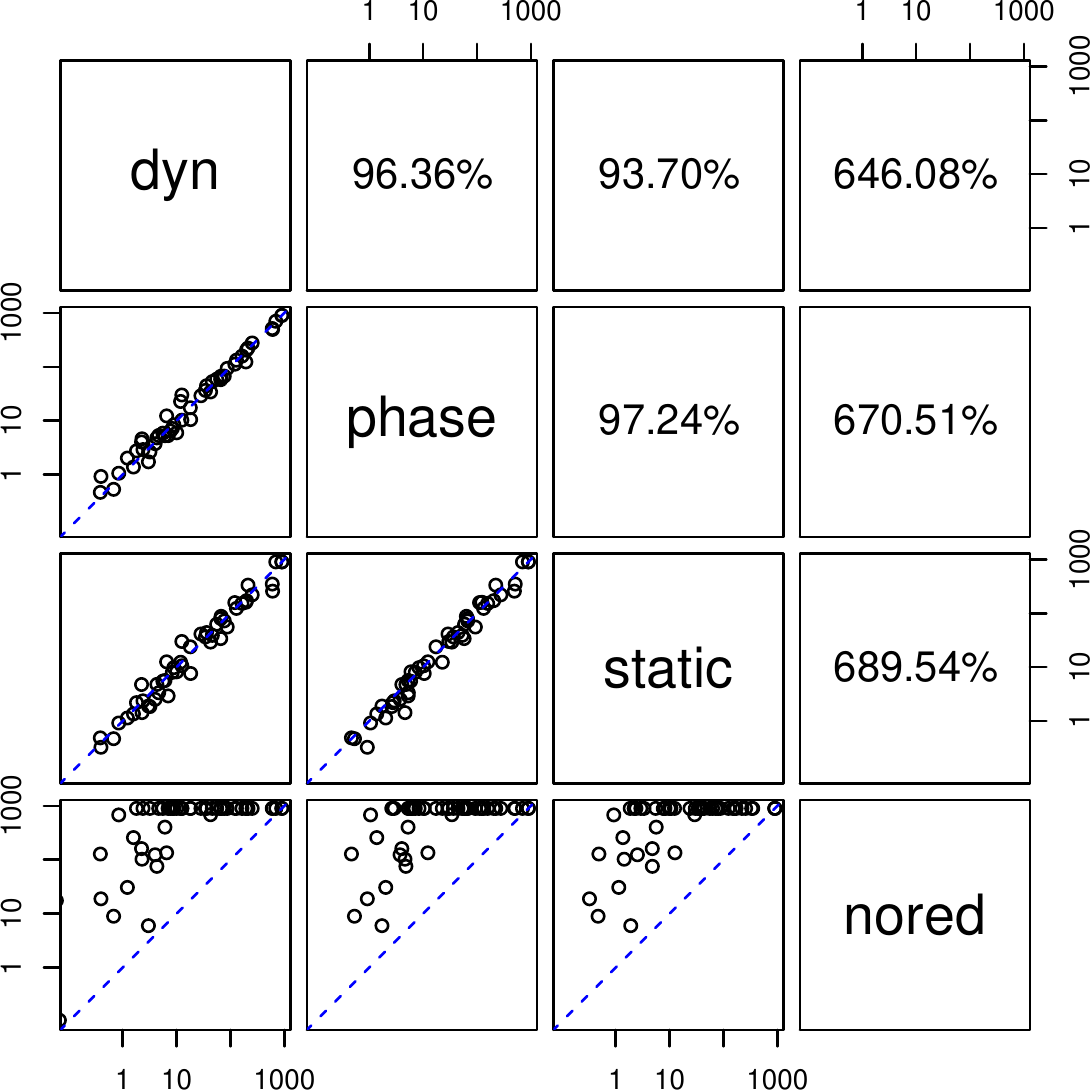}
  \caption{Scatterplots comparing runtimes for all combinations of
  		reduction variants on SVComp benchmarks.
   The upper half shows relative accumulated runtimes for these combinations.}
  \label{fig:svcomp}
\end{figure}



\section{Related Work}
\label{sec:related}

Lipton's reduction was refined multiple times
\cite{lamport-lipton,gribomont,Doeppner:1977:PPC:512950.512965,lamport-tla,Stoller2003}.
It has recently been applied in the context
of compositional verification \cite{PopeeaRW14}.
Qadeer and Flanagan proposed reductions with dynamic
phases~\cite{qadeer-transactions} using phase variables to identify internal and external states
and also provided a dynamic solution for determining locked regions. Their approach, 
however, does not solve the examples featured in the current paper
and also relies on a specialized deductions incompatible with IC3.
Moreover, in \cite{qadeer-transactions}, the phases of
target locations of non-deterministic conditions are required to
agree. This restriction is not present in our encoding.


Grumberg et al.~\cite{Grumberg:2005:PUM:1040305.1040316} present
underapproximation-widening, which iteratively refines an under-approximated
encoding of the system. In their implementation, interleavings are constrained
to achieve the under-approximation. Because refinement is done based on verification
proofs, irrelevant interleavings will never be considered.
The technique currently only supports BMC and the implementation is not available,
so we did not compare against it.

Kahlon et al.~\cite{Kahlon2006} extend the dynamic solution of \cite{qadeer-transactions},
by supporting a strictly more general set of lock patterns.
They incorporate the transactions into the stubborn set POR method~\cite{valmari-91}
and encode these in the transition relation in similar fashion as in~\cite{alur1997partial}.
Unlike ours, their technique does not support other constructs than locks.

While in fact it is sufficient for \autoref{item2:postterm} of
\autoref{def:pas2} to pinpoint a single state in each bottom
SCC of the CFG, we use feedback sets because
the encoding in \autoref{sec:encoding} also requires them.
Moreover, we take a syntactical definition for ease of explanation.
Semantic heuristics for better feedback sets can be found in 
\cite{kurshan1998static} and can easily be supported via state predicates.
(Further optimizations are possible~\cite{kurshan1998static}.)
Obtaining the smallest (vertex) LFS is an \NP-complete problem
well known in graph theory~\cite{bondy}.
As CFGs are simple graphs, estimations via basic DFS suffice.
(In POR, similar approaches are used for
\concept{the ignoring problem}~\cite{Valmari1991,twophase}.)

Elmas et al.~\cite{ElmasQT09} propose
dynamic reductions for type systems, where the 
invariant is used to weaken the mover definition. The
over-approximations performed in IC3, however,
decrease the effectiveness of such approaches.

In POR, similar techniques have been employed in~\cite{dwyer2004exploiting} and
the earlier-mentioned \concept{necessary enabling sets} of \cite{godefroid,parle89}.
Completely dynamic approaches exist~\cite{flanagan2005dynamic},
but symbolic versions remain highly static~\cite{alur1997partial}.
Notable exceptions are peephole and monotonic POR by
Wang et al.~\cite{peephole,kahlon-wang-gupta}. 
Like sleep sets~\cite{godefroid}, however, these only reduce
the number of transitions---not states---which is crucial in e.g. IC3 to cut
counterexamples to induction~\cite{ic3}.
Cartesian POR~\cite{cartesian} is a dynamic form of
Lipton reduction for explicit-state model checking.


\section{Conclusions}
\label{sec:conclusion}

Our work provides a novel dynamic reduction for symbolic software model checking.
To accomplish this, we presented a reduction theorem generalized with
bisimulation, 
facilitating various dynamic instrumentations as our heuristics show. 
We demonstrated its effectiveness with an encoding used by the BMC and
IC3 algorithms in VVT.

\bibliographystyle{plain}
\bibliography{lit}

\begin{thebibliography}{10}

\bibitem{alur1997partial}
R.~Alur, R.~K. Brayton, T.~A. Henzinger, S.~Qadeer, and S.~K. Rajamani.
\newblock Partial-order reduction in symbolic state space exploration.
\newblock In {\em CAV}, volume 1254 of {\em LNCS}, pages 340--351. Springer,
  1997.

\bibitem{svcomp}
Dirk Beyer.
\newblock {The Software Verification Competition}.
\newblock \url{http://sv-comp.sosy-lab.org/2016/}.

\bibitem{Beyer2016}
Dirk Beyer.
\newblock Reliable and reproducible competition results with {BenchExec} and
  witnesses (report on {SV-COMP} 2016).
\newblock In {\em TACAS}, volume 9636 of {\em LNCS}, pages 887--904. Springer,
  2016.

\bibitem{svcomp16}
Dirk Beyer.
\newblock Reliable and reproducible competition results with benchexec and
  witnesses report on sv-comp 2016.
\newblock In {\em TACAS}, volume 9636 of {\em LNCS}, pages 887--904. Springer,
  2016.

\bibitem{BeyerCGKS09}
Dirk Beyer, Alessandro Cimatti, Alberto Griggio, M.~Erkan Keremoglu, and
  Roberto Sebastiani.
\newblock Software model checking via large-block encoding.
\newblock In {\em FMCAD}, pages 25--32. IEEE, 2009.

\bibitem{beyerCAV07}
Dirk Beyer, Thomas~A. Henzinger, and Gr{\'e}gory Th{\'e}oduloz.
\newblock Configurable software verification.
\newblock In {\em CAV}, volume 4590 of {\em LNCS}, pages 504--518. Springer,
  2007.

\bibitem{CTIGAR}
Johannes Birgmeier, Aaron Bradley, and Georg Weissenbacher.
\newblock Counterexample to induction-guided abstraction-refinement ({CTIGAR}).
\newblock In {\em CAV}, volume 8559 of {\em LNCS}, pages 829--846. Springer,
  2014.

\bibitem{bondy}
John~Adrian Bondy and Uppaluri Siva~Ramachandra Murty.
\newblock {\em Graph theory with applications}, volume 290.
\newblock Macmillan London, 1976.

\bibitem{ic3}
Aaron~R. Bradley.
\newblock {SAT}-based model checking without unrolling.
\newblock In {\em VMCAI}, volume 6538 of {\em LNCS}, pages 70--87. Springer,
  2011.

\bibitem{CimattiGMT14}
Alessandro Cimatti, Alberto Griggio, Sergio Mover, and Stefano Tonetta.
\newblock {IC3} modulo theories via implicit predicate abstraction.
\newblock In {\em TACAS}, volume 8413 of {\em LNCS}, pages 46--61. Springer,
  2014.

\bibitem{lamport-tla}
Ernie Cohen and Leslie Lamport.
\newblock Reduction in {TLA}.
\newblock In {\em CONCUR}, volume 1466 of {\em LNCS}, pages 317--331. Springer,
  1998.

\bibitem{dimitrov2014commutativity}
Dimitar {Dimitrov et al.}
\newblock Commutativity race detection.
\newblock In {\em ACM SIGPLAN Notices}, volume 49 (6), pages 305--315. ACM,
  2014.

\bibitem{Doeppner:1977:PPC:512950.512965}
Thomas~W. Doeppner, Jr.
\newblock Parallel program correctness through refinement.
\newblock In {\em POPL}, pages 155--169. ACM, 1977.

\bibitem{dwyer2004exploiting}
Matthew~B. Dwyer, John Hatcliff, Robby, and Venkatesh~Prasad Ranganath.
\newblock Exploiting object escape and locking information in partial-order
  reductions for concurrent object-oriented programs.
\newblock {\em FMSD}, 25(2-3):199--240, 2004.

\bibitem{ElmasQT09}
Tayfun Elmas, Shaz Qadeer, and Serdar Tasiran.
\newblock A calculus of atomic actions.
\newblock In {\em POPL}, pages 2--15. ACM, 2009.

\bibitem{flanagan2005dynamic}
Cormac Flanagan and Patrice Godefroid.
\newblock Dynamic partial-order reduction for model checking software.
\newblock In {\em POPL}, volume 40 (1), pages 110--121. ACM, 2005.

\bibitem{qadeer-transactions}
Cormac Flanagan and Shaz Qadeer.
\newblock Transactions for software model checking.
\newblock {\em ENTCS}, 89(3):518 -- 539, 2003.
\newblock Software Model Checking.

\bibitem{qadeer-atomicity}
Cormac Flanagan and Shaz Qadeer.
\newblock A type and effect system for atomicity.
\newblock In {\em PLDI}, pages 338--349. ACM, 2003.

\bibitem{godefroid}
Patrice Godefroid, editor.
\newblock {\em Partial-Order Methods for the Verification of Concurrent
  Systems}, volume 1032 of {\em LNCS}.
\newblock Springer, 1996.

\bibitem{gribomont}
Pascal Gribomont.
\newblock Atomicity refinement and trace reduction theorems.
\newblock In {\em CAV}, volume 1102 of {\em LNCS}, pages 311--322. Springer,
  1996.

\bibitem{Grumberg:2005:PUM:1040305.1040316}
Orna Grumberg, Flavio Lerda, Ofer Strichman, and Michael Theobald.
\newblock Proof-guided underapproximation-widening for multi-process systems.
\newblock In {\em POPL}, pages 122--131. ACM, 2005.

\bibitem{cartesian}
Guy Gueta, Cormac Flanagan, Eran Yahav, and Mooly Sagiv.
\newblock Cartesian partial-order reduction.
\newblock In {\em SPIN}, volume 4595 of {\em LNCS}, pages 95--112. Springer,
  2007.

\bibitem{vvtwebsite}
Henning G\"unther.
\newblock {The Vienna Verification Tool website}.
\newblock \url{http://vvt.forsyte.at/}.
\newblock Last accessed: 2016-11-21.

\bibitem{vvt}
Henning G{\"u}nther, Alfons Laarman, and Georg Weissenbacher.
\newblock {Vienna Verification Tool}: {IC3} for parallel software.
\newblock In {\em TACAS}, volume 9636 of {\em LNCS}, pages 954--957. Springer,
  2016.

\bibitem{GuntherW14}
Henning G{\"{u}}nther and Georg Weissenbacher.
\newblock Incremental bounded software model checking.
\newblock In {\em SPIN}, pages 40--47. ACM, 2014.

\bibitem{hrmgs02}
Thomas~A. Henzinger, Ranjit Jhala, Rupak Majumdar, and Gr\'{e}goire Sutre.
\newblock Lazy abstraction.
\newblock In {\em POPL}, pages 58--70. ACM, 2002.

\bibitem{Kahlon2006}
Vineet Kahlon, Aarti Gupta, and Nishant Sinha.
\newblock Symbolic model checking of concurrent programs using partial orders
  and on-the-fly transactions.
\newblock In {\em CAV}, volume 4144 of {\em LNCS}, pages 286--299. Springer,
  2006.

\bibitem{kahlon-wang-gupta}
Vineet Kahlon, Chao Wang, and Aarti Gupta.
\newblock Monotonic partial order reduction: An optimal symbolic partial order
  reduction technique.
\newblock In {\em CAV}, volume 5643 of {\em LNCS}, pages 398--413. Springer,
  2009.

\bibitem{kurshan1998static}
Robert Kurshan, Vladimir Levin, Marius Minea, Doron Peled, and H{\"u}sn{\"u}
  Yenig{\"u}n.
\newblock Static partial order reduction.
\newblock In {\em TACAS}, volume 1384 of {\em LNCS}, pages 345--357. Springer,
  1998.

\bibitem{boosting}
A.W. {Laarman}, J.C. van~de Pol, and M.~{Weber}.
\newblock {Boosting Multi-Core Reachability Performance with Shared Hash
  Tables}.
\newblock In {\em FMCAD}, pages 247--255. IEEE-CS, 2010.

\bibitem{lamport-lipton}
Leslie Lamport and Fred~B. Schneider.
\newblock Pretending atomicity.
\newblock Technical report, Cornell University, 1989.

\bibitem{lipton}
Richard~J. Lipton.
\newblock Reduction: A method of proving properties of parallel programs.
\newblock {\em Comm. of the ACM}, 18(12):717--721, 1975.

\bibitem{impact}
Kenneth~L. McMillan.
\newblock Lazy abstraction with interpolants.
\newblock In {\em CAV}, volume 4144 of {\em LNCS}, pages 123--136. Springer,
  2006.

\bibitem{Milner}
Robin Milner.
\newblock {\em Communication and Concurrency}.
\newblock Prentice Hall, 1989.

\bibitem{twophase}
Ratan Nalumasu and Ganesh Gopalakrishnan.
\newblock {An Efficient Partial Order Reduction Algorithm with an Alternative
  Proviso Implementation}.
\newblock {\em FMSD}, 20(3):231--247, 2002.

\bibitem{papadimitriou}
Christos Papadimitriou.
\newblock {\em The theory of database concurrency control}.
\newblock Principles of computer science series. Computer Science Pr., 1986.

\bibitem{Park}
David Park.
\newblock Concurrency and automata on infinite sequences.
\newblock In {\em Theoretical Computer Science}, volume 104 of {\em LNCS},
  pages 167--183. Springer, 1981.

\bibitem{peled-93}
Doron Peled.
\newblock {All from One, One for All: on Model Checking Using Representatives}.
\newblock In {\em CAV}, volume 697 of {\em LNCS}, pages 409--423. Springer,
  1993.

\bibitem{PopeeaRW14}
Corneliu Popeea, Andrey Rybalchenko, and Andreas Wilhelm.
\newblock Reduction for compositional verification of multi-threaded programs.
\newblock In {\em FMCAD}, pages 187--194. IEEE, 2014.

\bibitem{Stoller2003}
Scott~D. Stoller and Ernie Cohen.
\newblock Optimistic synchronization-based state-space reduction.
\newblock In {\em TACAS}, volume 2619 of {\em LNCS}, pages 489--504. Springer,
  2003.

\bibitem{parle89}
A.~Valmari.
\newblock {Eliminating Redundant Interleavings During Concurrent Program
  Verification}.
\newblock In {\em PARLE}, volume 366 of {\em LNCS}, pages 89--103. Springer,
  1989.

\bibitem{valmari-91}
A.~Valmari.
\newblock {Stubborn Sets for Reduced State Space Generation}.
\newblock In {\em ICATPN/APN'89}, volume 483 of {\em LNCS}, pages 491--515.
  Springer, 1991.

\bibitem{Valmari1991}
Antti Valmari.
\newblock Stubborn sets for reduced state space generation.
\newblock In {\em Advances in Petri Nets}, volume 483 of {\em LNCS}, pages
  491--515. Springer, 1990.

\bibitem{peephole}
Chao Wang, Zijiang Yang, Vineet Kahlon, and Aarti Gupta.
\newblock Peephole partial order reduction.
\newblock In {\em TACAS}, volume 4963 of {\em LNCS}, pages 382--396. Springer,
  2008.

\end{thebibliography}

\clearpage
 
\appendix

\tikzset{
   brace/.style={
     decoration={brace, mirror},
     decorate
   }
}
\tikzset{
   bracea/.style={
     decoration={brace},
     decorate
   }
}

\section{Dynamic Reduction}
\label{app:proofs}

The current appendix provides proofs for \autoref{sec:dynamic}.
Some definitions presented here are more general in that they support both
dynamic right and left movers, whereas \autoref{sec:dynamic} only considers
dynamic both movers for ease of explanation.
The instrumentation is consequently also adapted to support dynamic left and right movers.
Moreover, \autoref{def:pas2} is decomposed into
\autoref{def:pas}, \autoref{def:tas} and \autoref{def:errors} in order to
introduce these concepts in a top-down, incremental fashion. Several
lemmas are proved along the way, to finally conclude
soundness and completeness of the axiomatized reduction in
\autoref{th:soundandcomplete}.

\autoref{lem:instrument} goes on to show that our
instrumentation preserves errors. And, \autoref{lem:instrument2} shows that it
fulfills the reduction axioms. So the instrumentation can be used as a valid basis
for obtaining reductions.

The inspiration for our reduction theorem comes from
\cite{qadeer-transactions}, which in turn is based on a string of earlier
works~(see \autoref{sec:related}).
Its generalization with bisimulation is
necessary to accommodate the dynamic behavior of movers, which causes
entire atomic sections to ``switch phase'' as explained in \autoref{sec:reduction}.

\begin{definition}[phase-annotated transition system]\label{def:pas}~\\
A parallel phase-annotated transition system $\PAS_\TS$
is a transition system $\TS=\tuple{S, \tatrans}$ with a
(parallel) transition relation $\tatrans = \bigcup_i \tatrans_i$ whose states are partitioned in three sets (of phases), for each thread $i$, and for all threads $i$ there exists a thread bisimulation $\cong_i$. Additionally, we require the following properties
(for all $i$ and all $j \neq i$):
\begin{enumerate}
\item\label{item:part}
	$S= \RR_i \uplus \LL_i \uplus \NN_i$
	\hfill
	(the 3-partition)

\item\label{item:disjoint}
	$\tatrans_i\cap\tatrans_j=\emptyset$
	\hfill
	($i$'s and $j$'s transitions are disjoint)

\item\label{item:invar}
	$\tatrans_i\,\subseteq \mathcal{\LL}_j^2\cup\mathcal{\RR}_j^2\cup\mathcal{\NN}_j^2$
	\hfill
	($i$ preserves $j$'s phase)

\item\label{item:bisim-threads}
	$\cong_i\,\subseteq \mathcal{\LL}_j^2\cup\mathcal{\RR}_j^2\cup\mathcal{\NN}_j^2$
	\hfill
	($\cong_i$ preserves phase-equality for $j$)
	
\end{enumerate}
\end{definition}

Note that all transitions in a parallel phase-annotated transition system are distinguishable, i.e., can be assigned to unique threads performing them (due to the parallel transition relation property and~\autoref{item:disjoint}).
We will apply this feature silently and assign threads to steps in
paths whenever needed.

We next define equivalence relations $\cong_X$ for $X \subseteq\threads$ and let  $\cong\defn \cong_\threads$. We put $$\cong_X = \left(\bigcup_{i\in X} \cong_i \right)^*.$$ Hence, $\cong_X$ is the equivalence closure of the union of all $\cong_i$. Note that (1) $\cong_{\{i\}} = \cong_i$, and (2) $\cong_X \subseteq \mathcal{\LL}_j^2\cup\mathcal{\RR}_j^2\cup\mathcal{\NN}_j^2$ for all $j \notin X$. The following properties are immediate.

\begin{corollary}\label{cor:bisim-upto1}
The relation $\cong_X$ is a thread bisimulation. As a consequence, it is also a standard bisimulation of $(S, \rightarrow)$. 
\end{corollary}

\begin{corollary}\label{cor:bisim-threads}
For any path starting in $\sigma$, if $\sigma'$ is such that $\sigma'\cong_X \sigma$ for
$X\subseteq \threads$, then 
there is a (bisimilar) path from $\sigma'$ where the same
threads transit in the same order.
\end{corollary}

\begin{corollary}\label{cor:bisim-upto}
For a path $\sigma_1\tatrans \sigma_2\tatrans\ldots$, let the phase trace for $i$
be a sequence $X_1,X_2,\ldots$ with $X_x = \RR \Leftrightarrow \sigma_x\in \RR_i$,
etc. Bisimilar states $\sigma\cong_{X} \sigma'$ have equal
phase traces for all $i$ such that $i\notin X$.
\end{corollary}

Let $\NN_T \defn \bigcap_{i\in T} \NN_i$ and $\NN = \NN_\threads$.

\begin{definition}[parallel \emph{transaction} system]
\label{def:tas}
A parallel transaction system $\TAS_\TS$ is a phase-annotated
transition system $\PAS_\TS\defn(S, \tatrans)$ s.t.
for all $i$ and all $j\neq i$:
%
%
\begin{enumerate}
\item \label{item:post}
	$\mathcal{\LL}_i\lrestr \tatrans_i\rrestr \mathcal{\RR}_i=\emptyset$
	\hfill
	((locally) post does not reach pre)

\item \label{item:rmover}
	$(\tatrans_i \rrestr \RR_{i}) \rcomm_j \tatrans_j$
	 \hfill
	($i$ to pre right commutes up to $j$ with~$j$)

\item \label{item:lmover}
	$(\LL_i \lrestr \tatrans_i) \lcomm_{\set{i,j}} \tatrans_j$
	 \hfill
	($i$ from post left~commutes up to $\{i,j\}$ with~$j$)

\item \label{item:postterm}
	$\forall \sigma \in \LL_i \colon \exists \sigma'\in\NN_i \colon
	\sigma \tatrans_i^+ \sigma'$
	 \hfill
	\hspace{-1em}(post phases terminate)
\end{enumerate}

\noindent The  transaction relation is $\trtrans = \bigcup_i \trtrans_i$ with

\qquad$\trtrans_i \defn \NN_{j\neq i} \lrestr \tatrans_i$
\qquad ($i$ only transits when all $j$ are in external).\\

\end{definition}

\autoref{th:dynred} shows that for all paths in a transaction system there also
exists a bisimilar \emph{transaction path} ($\trtrans$).
This enables reduction by removing  interleavings.
But first, we prove several intermediate lemmas, before we can introduce
the error states used in the theorem to connect divergent paths between
the two transition relations.

The above 
definition differs on some key points from that of Flanagan and Qadeer:
Most importantly, it lifts the commutativity condition from the phase change.
This means that while transitions should still commute, their phase change
might not, crucially enabling the dynamic instrumentation.
To preserve correctness while allowing phase changes when moving transitions,
the system is constructed in such a 
way that these phase-unequal states are still bisimilar
as the above corollaries demonstrate.

\subsection{Incomplete Transaction Paths}
To reason over (transaction) paths where some threads are still in the middle of their 
transaction (open), we introduce the notion of
\concept{open transaction sequences} (ots).
An ots consists of two concatenated subpaths:
The first (called cts) has only transactions in committed form
(ending in post-commit or external phase),
while the second (called uts) only contains uncommitted transitions
(ending in the pre-commit phase).
For example, the phase trace for 3 different threads of the same path could 
look like:
\begin{verbatim}
1: NNNRRLLLNRRLLLLLLNNNNNNNNNNnRRRRRR
2: NNNNNNNNNNNNNNNNNNRRRRRLLLLnNNNNRR
3: NLLLLLLLLLLLLLLLLLLLLLLLLLLLLLLLLL
\end{verbatim}
Here the small letters indicate phases of states that are part of
both the cts (where it is the last state of the cts) and the uts
(the first state).
We start with a definition of the uts, as it is easier to understand.
We also prove two lemmas that are needed for correctness of the reduction
in \autoref{s:theorem}.

\begin{definition}[Uncommitted transaction sequence
    (uts)]\label{def:uts}~\\
  A {\em uts} under a set of threads $T$ is a sequence
  of transactions
  \[ q = q_1
  \tatrans_{\psi(1)}^+q_2\tatrans_{\psi(2)}^+\ldots\tatrans_{\psi(l)}^+q_{l+1} \]
  (where $l\ge 0$ and $\psi$ is a mapping from indices to threads)
  such that for all $1\leq i\leq l$, we have
  \[ q_i=q_{i,1}\tatrans_{\psi(i)}\ldots\tatrans_{\psi(i)}q_{i,x_i}=q_{i+1} \]
  with
  \begin{enumerate}
  \item $\psi(i)\in T$ and $q_1\in \NN_T$,
  \item $q_{i,1}\in\NN_{\psi(i)}$, and
  \item $q_{i,2},\ldots,q_{i,x_i}\in\RR_{\psi(i)}$.
  \end{enumerate}
\end{definition}

\begin{corollary}\label{cor:uts-threads}
  For every uts as defined in \autoref{def:uts} it holds that
\[\forall i,j\in\{1,\dots,l\}\colon i\neq j \implies \psi(i)\neq\psi(j)\,.\]
\end{corollary}

\begin{corollary}\label{cor:split-uts}
  Every uts $q_1
  \tatrans_{\psi(1)}^+q_2\tatrans_{\psi(2)}^+\ldots\tatrans_{\psi(l)}^+q_{l+1}$
  under $T$ 
  can be split at an arbitrary index
  $1\le k \le l+1$ into two uts's:\\
$q_1 \tatrans_{\psi(1)}^+\ldots \tatrans_{\psi(k-1)}^+q_{k}$\,,
	\hfill a uts under $T$, and\\
$q_k \tatrans_{\psi(k)}^+\ldots\tatrans_{\psi(l)}^+q_{l+1}$\,,
	\hfill a uts under $\set{i\in T \mid q_k \in \NN_i}$.
\end{corollary}

\begin{lemma}
  \label{lemma:right-move-uts}
  Given a uts
   $q_1\tatrans_{\psi(1)}^+q_2\tatrans_{\psi(2)}^+\ldots\tatrans_{\psi(l)}^+q_{l+1}$ under $T_u\defn \mathit{range}(\psi)$ and a transition $q_{l+1}\tatrans_i q$ with $i\not\in T_u$, then there exists a transition $q_1\tatrans_i q'$ and a uts $q' = q_1'\tatrans_{\psi(1)}^+q_2'\allowbreak\tatrans_{\psi(2)}^+\ldots\tatrans_{\psi(l)}^+q_{l+1}'$ under $T_u$ such that $q_{l+1}'\cong_{\set i} q$. 
Illustratively:
\begin{center}
\begin{tikzpicture}
   \tikzstyle{e}=[minimum width=1cm]
   \tikzstyle{every node}=[font=\small, node distance=.5cm]

  \node (s1) {$q_1$};
    
  \node (s2) [gray,node distance=1.2cm,below of=s1] {$q_1'$};
  \node (m1) [gray,right of=s2, xshift=.9cm] {$q_2'$};
  \node (m2) [gray,right of=m1, xshift=.9cm] {$q_l'$};
  \node (s3) [gray,right of=m2, xshift=.9cm] {$q_{l+1}'$};
  \path (s1) -- node[gray,midway,sloped]{$\tatrans_i$} (s2);
  \path (s2) -- node[gray,pos=.45]{$\tatrans_{\psi(1)}^+$} (m1);
  \path (m1) -- node(d1)[gray,pos=.45]{$\ldots$} (m2.west);
  \path (m2) -- node[gray,pos=.45]{$\tatrans_{\psi(l)}^+$} (s3);

  \node (n1) [right of=s1, xshift=.9cm] {$q_2$};
  \node (n2) [right of=n1, xshift=.9cm] {$q_l$};
  \node (s4) [xshift=.9cm,right of=n2] {$q_{l+1}$};
  \path (s1) -- node[midway,sloped]{$\tatrans_{\psi(1)}^+$} (n1);
  \path (n1) -- node(d2) [midway,sloped]{$\ldots$} (n2);
  \path (n2) -- node[midway,sloped]{$\tatrans_{\psi(l)}^+$} (s4);

  \path (d1.south) -- node[midway,sloped]{$\Leftarrow$} (d2.north);

  \node (s3p) [node distance=.8cm,below right of=s4] {$q$};
  \path (s4.south) -- node[midway,sloped]{$\tatrans_i$} (s3p.north);

  \path (s3p) -- node[gray,sloped,pos=.48]{$\cong_{\set i}$} (s3);

\draw [bracea] (s1.north) -- node [above, pos=0.5] {uts} (s4.north);
\draw [brace,gray] (s2.south) -- node [below, pos=0.5,gray] {uts} (s3.south);
\end{tikzpicture}
\end{center}
\end{lemma}

\begin{proof}
	The proof proceeds by nested induction over the transaction blocks in
	the uts and the steps inside each block.

  The outer induction is performed over the length $l$ (i.e., the number of
  transactions) of the uts.
  \begin{itemize}
  \item Base case ($l=0$):
    Choose $q_1'=q_{l+1}'=q'$, which is trivially a uts of length 0.
    Furthermore, $q_{l+1}'\cong q$.
  \item Induction hypothesis (IH):
    Suppose that the lemma holds for the sequence
    \[
    q_1\tatrans_{\psi(1)}^+q_2\tatrans_{\psi(2)}^+\ldots\tatrans_{\psi(l-1)}^+q_{l}\tatrans_i q
    \]
    (or $q_1(\tatrans^+)^{l-1} q_{l} \tatrans_i q$, for short).
    
  \item Inductive case:
    Using the induction hypothesis for the uts
    $q_1(\tatrans^+)^{l-1} q_{l}$ of length $l-1$,
    we show that the lemma must hold equally for a sequence
    \[
    \underbrace{q_1(\tatrans^+)^{l-1} q_{l} \tatrans_{\psi(l)}^+ q_{l+1}}_{\text{uts}}
    \tatrans_i q'
    \]
    (where also $\psi(l) \neq i$).

    Let $j\defn \psi(l)$ be the last thread executing before $i$.
    We thus have (gray parts are the proof obligation):\\
\begin{tikzpicture}
   \tikzstyle{e}=[minimum width=1cm]
   \tikzstyle{every node}=[font=\small, node distance=.75cm]

  \node (s1) {$q_1$};

  \node (s2) [node distance=1.2cm,below of=s1] {$q_1'$};
  \node (m1) [right of=s2, xshift=.9cm] {$q_2'$};
  \node (m2) [right of=m1, xshift=.9cm] {$q_{l-1}'$};
  \node (m3) [right of=m2, xshift=.9cm] {$q_l'$};
  \node (s3) [gray,right of=m3, xshift=.9cm] {$q_{l+1}'$};
  \path (s1) -- node[midway,sloped]{$\tatrans_i$} (s2);
  \path (s2) -- node[pos=.45]{$\tatrans_{\psi(1)}^+$} (m1);
  \path (m1) -- node(d1)[pos=.45]{$\ldots$} (m2);
  \path (m2) -- node[pos=.45]{$\tatrans_{\psi(l-1)}^+$} (m3);
  \path (m3) -- node[gray,pos=.45]{$\tatrans_j^+$} (s3);

  \node (n1) [right of=s1, xshift=.9cm] {$q_2$};
  \node (n2) [right of=n1, xshift=.9cm] {$q_{l-1}$};
  \node (n3) [right of=n2, xshift=.9cm] {$q_l$};
  \node (s4) [xshift=.9cm,right of=n3] {$q_{l+1}$};
  \path (s1) -- node[midway,sloped]{$\tatrans_{\psi(1)}^+$} (n1);
  \path (n1) -- node(d2) [midway,sloped]{$\ldots$} (n2);
  \path (n2) -- node[midway,sloped]{$\tatrans_{\psi(l-1)}^+$} (n3);
  \path (n3) -- node[midway,sloped]{$\tatrans_j^+$} (s4);

  \path (d1) -- node[midway,sloped]{$\Leftarrow$} (d2);

  \node (n3p) [node distance=.8cm,xshift=.5cm,below right of=n3] {$q$};
  \path (n3) -- node[midway,sloped]{$\tatrans_i$} (n3p);
  \path (m3) -- node[sloped,pos=.48]{$\cong_{\set i}$} (n3p);

  \node (s3p) [gray,node distance=.8cm,below right of=s4,xshift=.5cm] {$q'$};

  \path (s3p) -- node[gray,sloped,pos=.48]{$\cong_{\set i}$} (s3);
  \path (s4) -- node[gray,midway,sloped]{$\tatrans_i$} (s3p);

\draw [bracea] (s1.north) -- node [above, pos=0.5] {uts} (s4.north);
\draw [brace,gray] (s2.south) -- node [below, pos=0.5,gray] {uts} (s3.south);

\end{tikzpicture}

The nested induction is over the steps in the transaction $\tatrans_j^+$.
By induction over the length $m$ of the path 
$q_{l} = q_{l,1}\tatrans_{j} \ldots \tatrans_{j} q_{l,m} = q_{l+1}$
with $q_{l,2}, \ldots, q_{l,m} \in \RR_{j}$, we show that there is 
a path
$q_{l,1} \tatrans_i q_{l,1}'\tatrans_{j} \ldots \tatrans_{j} q_{l,m}'$
with $q_{l,2}', \ldots, q_{l,m}' \in \RR_{j}$ and $q_{l,m}'\cong_{i} q'$, if there is a transition
$q_{l,m} \tatrans_i q'$:
	
\begin{center}
  \begin{tikzpicture}
    \tikzstyle{e}=[minimum width=1cm]
    \tikzstyle{every node}=[font=\small, node distance=.5cm]
    
    \node (s1) {$q_{l,1}$};
    
    \node (s2) [gray,node distance=1.2cm,below of=s1] {$q_{l,1}'$};
    \node (m1) [gray,right of=s2, xshift=.9cm] {$q_{l,2}'$};
    \node (m2) [gray,right of=m1, xshift=.9cm] {$q_{l,m-1}'$};
    \node (s3) [gray,right of=m2, xshift=.9cm] {$q_{l,m}'$};
    \path (s1) -- node[gray,midway,sloped]{$\tatrans_i$} (s2);
    \path (s2) -- node[gray,pos=.45]{$\tatrans_{j}$} (m1);
    \path (m1) -- node(d1)[gray,pos=.45]{$\ldots$} (m2);
    \path (m2) -- node[gray,pos=.45]{$\tatrans_{j}$} (s3);
    
    \node (n1) [right of=s1, xshift=.9cm] {$q_{l,2}$};
    \node (n2) [right of=n1, xshift=.9cm] {$q_{l,m-1}$};
    \node (s4) [xshift=.9cm,right of=n2] {$q_{l,m}$};
    \path (s1) -- node[midway,sloped]{$\tatrans_{j}$} (n1);
    \path (n1) -- node(d2) [midway,sloped]{$\ldots$} (n2);
    \path (n2) -- node[midway,sloped]{$\tatrans_{j}$} (s4);
    
    \path (d1) -- node[midway,sloped]{$\Leftarrow$} (d2);
    
    \node (s3p) [node distance=.8cm,xshift=.5cm,below right of=s4] {$q'$};
    \path (s4) -- node[midway,sloped]{$\tatrans_i$} (s3p);
    
    \path (s3p) -- node[gray,sloped,pos=.48]{$\cong_{\set i}$} (s3);
    \draw [bracea] (n1.north) -- node [above, pos=0.5] {$\in\RR_j$} (s4.north);
    \draw [brace,gray] (m1.south) -- node [below, pos=0.5,gray] {$\in\RR_j$} (s3.south);
  \end{tikzpicture}
\end{center}

\begin{itemize}
  \item Base case ($m=2$ and $q_{l,1}\tatrans_{j}q_{l,m}\tatrans_{i}q'$):
    By \autoref{item:rmover} of the definition of transaction systems
    (\autoref{def:tas}), we obtain
    $q_{l,1}\tatrans_{i} q_{l,1}' \tatrans_{j} q_{l,2}'$
    with $q_{l,2}' \cong_{\set i} q'$ and $q_{l,2}'\in \RR_j$.
    
  \item Induction hypothesis (IH2):
    Assume that the hypothesis holds for $(m - 1) > 0$.

  \item Induction step:
	We show that our claim holds for a path up to $m$ ending with:
	$q_{l,m-1}\tatrans_{j} q_{l,m} \tatrans_{i} q'$ with $q_{l,m}\in\RR_j$.
	By \autoref{item:rmover} of \autoref{def:tas}, we obtain
	$q_{l,m-1}\tatrans_{i} q_{l,m-1}' \tatrans_{j} q_{l,m}'$
	with $q_{l,m}' \cong_{\set i} q'$ and $q_{l,m}'\in \RR_j$.
	We then apply IH2 on the path of length $m-1$ to obtain a path
	$q_{l,1} \tatrans_i q_{l,1}'\tatrans_{j} \ldots \tatrans_{j} q_{l,m-2}'\tatrans_{j} q_{l,m-1}''$
	with $q_{l,2}', \ldots, q_{l,m-1}' \in \RR_{j}$ and  $q_{l,m-1}'' \cong_{\set i} q_{l,m-1}'$. This implies
	$q_{l,m-1}'' \tatrans_j q_{l,m}''$ with $q_{l,m}'\cong_{\set i} q_{l,m}''$
        (since $\cong_{\set i}$ is a bisimulation relation).
        Since $\cong_{\set i}$ preserves the phase of other
        threads, it follows that 
        $q_{l,m}''\in R_j$ from $q_{l,m}'\in R_j$.

\begin{center}
\begin{tikzpicture}
   \tikzstyle{e}=[minimum width=1cm]
   \tikzstyle{every node}=[font=\small, node distance=.5cm]

  \node (s1) {};
  \node (n1) [right of=s1, xshift=.9cm] {$q_{l,m-2}$};
  \node (n2) [right of=n1, xshift=.9cm] {$q_{l,m-1}$};
  \node (s4) [xshift=.9cm,right of=n2] {$q_{l,m}$};
  \path (s1.east) -- node(d1)[midway,sloped]{\vphantom{R}$\ldots$} (n1.west);
  \path (n1.east) -- node[midway,sloped]{$\tatrans_{j}$} (n2.west);
  \path (n2.east) -- node[midway,sloped]{$\tatrans_{j}$} (s4.west);
    
  \node (s2) [gray,node distance=1.2cm,below of=s1] {};
  \node (m1) [gray,right of=s2, xshift=.9cm] {$q_{l,m-2}'$};
  \node (m2) [gray,right of=m1, xshift=.9cm] {$q_{l,m-1}'$};
  \node (m2p) [gray,node distance=1.6cm,below right of=m1] {$q_{l,m-1}''$};
  \node (m1p) [gray,node distance=1.6cm,right of=m2p] {$q_{l,m}''$};
  \node (s3) [gray,right of=m2, xshift=.9cm] {$q_{l,m}'$};
  \path (s2.east) -- node [gray,pos=.45]{$\ldots$} (m1.west);
  \path (m1) -- node[gray,midway,sloped]{$\tatrans_{j}$} (m2p);
  \path (m2p.east) -- node[gray,midway]{$\tatrans_{j}$} (m1p.west);
  \path (m2.east) -- node[gray,pos=.45]{$\tatrans_{j}$} (s3.west);

  \path (n1) -- node[gray,midway,sloped]{$\tatrans_i$} (m1);
  \path (n2) -- node[gray,midway,sloped]{$\tatrans_i$} (m2);

  \node (s3p) [node distance=1.2cm,below right of=s4] {$q$};
  \path (s4) -- node[midway,sloped]{$\tatrans_i$} (s3p);

  \path (m2p) -- node[gray,sloped,pos=.48]{$\cong_{\set i}$} (m2);
  \path (m1p) -- node[gray,sloped,pos=.48]{$\cong_{\set i}$} (s3);
  \path (s3p) -- node[gray,sloped,pos=.48]{$\cong_{\set i}$} (s3);
\draw [bracea] (d1.north) -- node [above, pos=0.5] {$\in\RR_j$} (s4.north);
\node [right of=s3,gray, xshift=.3cm,yshift=-.1cm] {$\in\RR_j$};
\node [below of=m1,gray,xshift=-.4cm] {$\in\RR_j$};
\node [left of=m2p,gray, yshift=-.1cm,xshift=-.5cm] {$\RR_j\ni$};
\node [right of=m1p,gray, xshift=.3cm,yshift=.1cm] {$\stackrel{}{\in}\RR_j$};
\end{tikzpicture}
\end{center}
	The resulting path satisfies the required property.
	Note that $q_{l,1}'\tatrans_{j} \ldots \tatrans_{j} q_{l,m}'$ is a uts.

\end{itemize}
	We thus obtain a transition sequence:
    \[ \underbrace{q_1\tatrans_{\psi(1)}^+\ldots\tatrans_{\psi(l-1)}^+q_l \tatrans_i q_l' }_{A}{\tatrans_{\psi(l)^+}q_{l+1}'}
    \]
    with $\psi(1),\ldots,\psi(l-1)\neq i$.
    Applying the IH to $A$ yields a path:
    \[ q_1\tatrans_i
    	\underbrace{q_1'\tatrans_{\psi(1)}^+q_2'\tatrans_{\psi(2)}^+\ldots q_{l-1}'\tatrans_{\psi(l-1)}^+q_l''}_{B}
    \]
    Since, $q_l' \cong_{\set i} q_l''$, we also have
    $q_l''\tatrans_{j}q_{l+1}''$, by
    \autoref{cor:bisim-threads} and $i\neq j$.
    This is a uts under $\set j$.
    Since $B$ remains a uts under $T_u$,
    the path  $q_1'(\tatrans^+)^nq_{l+1}'$
     satisfies the needed property (it is a uts under $T_u\uplus \set j$).
     
\begin{tikzpicture}
   \tikzstyle{e}=[minimum width=1cm]
   \tikzstyle{every node}=[font=\small, node distance=.8cm]

  \node (s1) {$q_1$};
    
  \node (s2) [node distance=1.4cm,below of=s1] {$q_1'$};
  \node (m1) [right of=s2, xshift=.9cm] {$q_2'$};
  \node (m2) [right of=m1, xshift=.9cm] {$q_{l-1}'$};
  \node (m3) [right of=m2, xshift=.9cm] {$q_l''$};
  \node (s3) [gray,right of=m3, xshift=.9cm] {$q_{l+1}''$};
  \path (s1) -- node[midway,sloped]{$\tatrans_i$} (s2);
  \path (s2) -- node[pos=.45]{$\tatrans_{\psi(1)}^+$} (m1);
  \path (m1) -- node(d1)[pos=.45]{$\ldots$} (m2);
  \path (m2) -- node[pos=.45]{$\tatrans_{\psi(l-1)}^+$} (m3);
  \path (m3) -- node(e1)[gray,pos=.45]{$\tatrans_{j}^+$} (s3);

  \node (n1) [right of=s1, xshift=.9cm] {$q_2$};
  \node (n2) [right of=n1, xshift=.9cm] {$q_{l-1}$};
  \node (n3) [right of=n2, xshift=.9cm] {$q_l$};
  \node (s4) [xshift=.9cm,right of=n3] {$q_{l+1}$};
  \path (s1) -- node[midway,sloped]{$\tatrans_{\psi(1)}^+$} (n1);
  \path (n1) -- node(d2) [midway,sloped]{$\ldots$} (n2);
  \path (n2) -- node[midway]{$\tatrans_{\psi(l-1)}^+$} (n3);
  \path (n3) -- node[midway,sloped]{$\tatrans_{j}^+$} (s4);

  \path (d1) -- node[midway,sloped]{$\Leftarrow$} (d2);

  \node (s3p) [node distance=.8cm,xshift=.3cm,below right of=s4] {$q_{l+1}'$};
  \path (s4) -- node[midway,sloped]{$\tatrans_i$} (s3p);
  
  \node (s3pp) [node distance=.8cm,xshift=.3cm,below right of=n3] {$q_l'$};
  \path (n3) -- node[midway,sloped]{$\tatrans_i$} (s3pp);
  \path (m3) -- node[sloped,pos=.48]{$\cong_{\set i}$} (s3pp);

  \path (s3p) -- node(e2)[midway]{$\tatrans_{j}^+$} (s3pp);

  \path (s3p) -- node[gray,sloped,pos=.48]{$\cong_{\set i}$} (s3);

  \path (e1) -- node[gray,sloped,pos=.48]{$\Leftarrow$} (e2);

\draw [bracea] (s1.north) -- node [above, pos=0.5] {uts} (n3.north);
\draw [brace,gray] (s2.south) -- node [below, pos=0.5,gray] {uts} (s3.south);

\end{tikzpicture}

  \end{itemize}
\end{proof}

\begin{definition}[Committed transaction sequence (cts)]\label{def:cts}~\\
  A cts under $T\subseteq\threads$ is a sequence of transactions with $k\ge 0$:
  \[ p = p_1 \tatrans_{\varphi(1)}^+p_2\tatrans_{\varphi(2)}^+\ldots\tatrans_{\varphi(k)}^+p_{k+1} \]
  such that for all $1\leq i\leq k$, we have $p_i=p_{i,1}\tatrans_{\varphi(i)}\ldots\tatrans_{\varphi(i)}p_{i,x_i}=p_{i+1}$ with
  \begin{enumerate}
  \item $\varphi(i)\in T$ and $p_1\in N_T$,
  \item $p_{i,1}\in\NN_{\varphi(i)}$,
  \item $p_{i,2},\ldots,p_{i,x_i-1}\in\RR_{\varphi(i)}\cup\LL_{\varphi(i)}$, and
  \item $p_{i,x_i}\in\LL_{\varphi(i)}\cup\NN_{\varphi(i)}$.
  \end{enumerate}
\end{definition}


\begin{corollary}\label{cor:split-cts}
An analogous ``splitting'' property to the one from \autoref{cor:split-uts} also holds for cts's.
\end{corollary}

\begin{definition}[Open transaction sequence (ots)]~\\
If $A$ is a cts under $T_A$ and $B$ is a uts under $T_B$ such that $A$ ends in the same state in which $B$ starts, then the sequence of transactions $A\circ B$ obtained by appending $B$ to $A$ is an ots under $T_A\cup T_B$.
\end{definition}

We sometimes write $AB$ for $A \circ B$. We refer to the cts part of $AB$ by using states $p_x, p_{x+1}$ (transactions start/end states), $p_{x,y}, p_{x,y+1}$ (internal transitions)
and threads $\varphi(x)$ for a transaction index $x$ and an internal transition index $y$ (as in \autoref{def:cts}). 
Similarly, we refer to the uts part by using states $q_x, q_{x+1}$ (transactions start/end states), $q_{x,y}, q_{x,y+1}$ (internal transitions)
and threads $\psi(x)$ for a transaction index $x$ and an internal transition index $y$ (as in \autoref{def:uts}).

\begin{corollary}\label{cor:suffix}
For an ots containing a sub-trace $\ldots r \tatrans_{i} \ldots$ with
$r\notin \NN_j$ for $j\neq i$, the suffix from $r$ does not contain
$j$ transitions.
\end{corollary}

\begin{corollary}\label{cor:split-ots}
An analogous ``splitting'' property to the one from \autoref{cor:split-uts} also holds for ots's.
\end{corollary}

\begin{corollary}
  \label{cor:always_ots}
  Every path $p\tatrans_i^*q$ with $p \in \NN_i$
  in a transaction system is an ots under~$\{i\}$.
\end{corollary}


\begin{lemma}[From transaction system path to ots]\label{lem:ots}~\\
Let $\TAS$ be a transaction system and
$\tatrans_T \defn \bigcup_{i\in T}\tatrans_i$.

Suppose that $ p\tatrans^*_T q$ is a \TAS path with $p\in \NN_{T}$,
then there exists an ots $p\tatrans^* q'$ under $T$ s.t. $q'\cong q$.
\end{lemma}

\begin{proof}
We prove the hypothesis by induction on the length of the path
$p\tatrans^* q$. The induction hypothesis (IH) is:

For every path $p\tatrans_T^* q$ with $p\in\NN_T$ there is a path
\[
  p=\underbrace{p_1\tatrans_{\varphi(1)}^+p_2\ldots p_{k+1}}_A=\underbrace{q_1\tatrans_{\psi(1)}^+q_2\ldots q_l\tatrans_{\psi(l)}^+q_{l+1}}_B
\]
such that $q_{l+1}\cong q$, $A$ is a cts under $T_C\defn\mathit{range}(\varphi)$, and $B$ a uts under $T_U\defn\mathit{range}(\psi)$ (so that $AB$ is an ots under $T_C\cup T_U \subseteq T$).

For the base case, $k=l=0$, take $q'\defn p = q$ and the conclusion holds
trivially.
Let $p\tatrans^*q \tatrans_i r$ be the path extension with $i$,
so the IH holds for $p\tatrans^*q$.
We show that there is an
$r' \cong r$ 
and a path $p\tatrans^+ r'$ that is an ots (under $T_C\cup T_U\cup\{i\}$).

We do case analysis over the $i$-phase of $q$:

\begin{description}
\item[$q\in\NN_i$:]
  It follows that $i\not\in T_U$, because otherwise we would have $q\in\RR_i$
  (by an inductive argument using the definition of the uts and \autoref{item:invar} of \autoref{def:pas}, also needed in the proof of \autoref{cor:suffix}).
  Using \autoref{lemma:right-move-uts}, one can move the transition to the end of the cts phase:
  \[ \underbrace{p_1\tatrans_{\varphi(1)}^+\ldots\tatrans_{\varphi(k)}p_{k+1}}_{\mathrm{cts}}\tatrans_i\underbrace{q_1'\tatrans_{\psi(1)}^+\ldots\tatrans_{\psi(l)}^+q_{l+1}'}_{\mathrm{uts}} \]
  such that $q_{l+1}'\cong_{\set i} q_{l+1}$.
  Depending on the phase of transition $i$'s target state $q_1'$, the transition can either be part of the cts ($q_1'\in\NN_i$ or $q_1'\in\LL_i$) or of the uts ($q_1'\in\RR_i$). We pick $r'\defn q'_{l+1}$. By \autoref{lemma:right-move-uts}, we have $r'\cong_{\set i} r$.
  The new path has the same transitions as the original path plus $i$
  (if $i$ was not in the path of the induction hypothesis).
  Hence, the property is satisfied.
\item[$q\in\LL_i$:]
  As in the previous case, we have $i\not\in T_U$.
  There is at least one $x$ such that  $\varphi(x)=i$, 
  otherwise we would have $q\in\NN_i$ (again applying the same inductive
  reasoning over the cts prefix path).
  Let $x$ be largest index for which $\varphi(x)=i$,  so that
  $\tatrans_{\varphi(x)}$ is the rightmost occurrence of $i$ in the ots.
  Therefore, the path has the form:
  \begin{align*}
    p_1&\tatrans_{\varphi(1)}^+\ldots p_x\tatrans_{\varphi(x)}^+p_{x+1}\ldots
    \tatrans_{\varphi(k)}^+p_{k+1}=q_1\\
    &\tatrans_{\psi(1)}^+
      \ldots \tatrans_{\psi(l)}^+q_{l+1}=q\tatrans_i r.
  \end{align*}
  Using \autoref{item:lmover} from \autoref{def:tas}, one can construct a path
  \[ p_1\tatrans_{\varphi(1)}^+\ldots\tatrans_{\varphi(x-1)}^+p_x\tatrans_{\varphi(x)}^+p_{x+1}\tatrans_i 
   \underbrace{p_{x+1}'\tatrans^* q'}_A \]
  with $q'\cong r$.
  By \autoref{cor:split-cts}, the prefix path including $p_{x+1}$ is a cts, so
  the prefix path including $p_{x+1}'$ is also a cts.
  Moreover, since by \autoref{cor:split-ots} 
  the original suffix path $p_{x+1}\tatrans^* q$ is an ots
  under $T'$ for some $T'\subseteq T \setminus \set i$,
  the threads that transit in $A$
  are still those in $T'$ by virtue of \autoref{item:lmover}
  and \autoref{cor:bisim-threads}.
  By contraposition of \autoref{cor:suffix},
  we also have $p_{x+1}\in \NN_{T'}$ and
  therefore also  $p_{x+1}'\in \NN_{T'}$
  (by \autoref{item:invar} of \autoref{def:pas} and $i\notin T'$).
  
  Since $1\le x \le k$, the suffix path $A$ is shorter than the original path.
  Hence one can apply the (strong) induction hypothesis to $A$ with $T'$
  to obtain an ots under  $T'$ starting in $p'_{x+1}\in \NN_{T'}$.
  This yields an ots under $T$:
  \[\overbrace{
  	p_1\tatrans_{\varphi(1)}^+\ldots\tatrans_{\varphi(x-1)}^+p_t
	\tatrans_{\varphi(t)}^+p_{x+1}\tatrans_i 
  	\underbrace{p_{x+1}'\tatrans^*_{T'} q''}_{\text{ots under}~T'}
  }^{\text{ots under}~T} \]
  with $q'' \cong q'$.
  Let $r'\defn q''$. By transitivity of $\cong$, we have $r' \cong r$.
  Thus the property is satisfied.
\item[$q\in\RR_i$:]
  There is at least one $x$ such that  $\psi(x)=i$, 
  otherwise we would have $q\in\NN_i\cup \LL_i$
  (again applying the same inductive
  reasoning over the ots).
  Let $x$ be the largest index for which $\psi(x)=i$,
  so that $\tatrans_{\psi(x)}$ is the rightmost occurrence of $i$ in the ots.
   
  One can then split the path into a cts $A$ under $T$
  (\autoref{cor:split-uts}),
  	a uts $B$ under $T_B\subseteq T$ and a uts $C$ under
	$T_C = T \setminus T_B$ by \autoref{cor:uts-threads}:

  \begin{align*}
    \underbrace{p_1\tatrans_{\varphi(1)}^+\ldots\tatrans_{\varphi(k)}p_{k+1}}_{A}\\
    =\underbrace{q_1\tatrans_{\psi(1)}^+\ldots q_x\tatrans_{\psi(x)}^+q}_{B}
    \underbrace{_{x+1}\ldots\tatrans_{\psi(l)}^+ q_{l+1}}_{C}
    \tatrans_i r
  \end{align*}
  From the above, we have  $i\in T_B$, but $i\notin T_C$.
  Hence, one can apply \autoref{lemma:right-move-uts} to $C$ and $i$
  to obtain a uts $C'$ under $T_C$ (ignore $D$ for now):
\begin{align*}
  \underbrace{p_1\tatrans_{\varphi(1)}^+\ldots\tatrans_{\varphi(k)}p_{k+1}}_{A}=
  \underbrace{\overbrace{q_1\tatrans_{\psi(1)}^+\ldots q_x}^{D}
  \tatrans_{\psi(x)}^+q_{x+1}}_{B}\\
  \tatrans_i 
  \underbrace{q_{x+1}'\ldots\tatrans_{\psi(l+1)}^+ q_{l+1}'}_{C'}=r'\cong_i r
\end{align*}

  If $q_{x+1}'\in\RR_i$, we are done as this path is an ots (under $T\ni i$).
  Otherwise,
  by repeatedly applying \autoref{lemma:right-move-uts} to
  the each of the $i$-transitions $q_x \tatrans_i^+ q'_{x+1}$ and $D$
  (a path without $i$),
  we right-commute all $i$-transitions to the end of the cts:
\[
	\overbrace{\underbrace{p_1\tatrans_{\varphi(1)}^+\ldots\tatrans_{\varphi(k)}p_{k+1}
		  }_{A} \tatrans_{i}^+ p_{k+1}'}^{
		  \text{cts under~}T\cup \set i}
		  =
  \underbrace{q_1'\tatrans_{\psi(1)}^+\ldots q_{x}'
  		}_{B'}
\]
with $q_{x}'=q_{x+1}''\cong_{\set i} q_{x+1}'$.
By \autoref{cor:always_ots}, we obtain a new cts $B'$ under
$T\cup \set i$.
Moreover, by bisimilarity up to $i$ (\autoref{cor:bisim-upto}), we
also have a path:
\[
  \underbrace{q_{x+1}''\ldots\tatrans_{\psi(l+1)}^+ q_{l+1}''}_{C''}=r'', \quad\text{~with~} r'' \cong_i r'
\]
This path is still a uts under $T_C$, because $i \notin T_C$.
Therefore, the resulting path 
$p_1\tatrans_T^+ r''$
 is an ots under $T$ with $r'' \cong_i r$
by transitivity of $\cong_i$.
This path satisfies the property.
\end{description}
Since $q\in\RR_i\uplus \LL_i\uplus\NN_i$, our case analysis and proof is 
complete.
\end{proof}

\subsection{Dynamic Reduction Theorem}\label{s:theorem}

We aim at preserving reachability of error states under the reduction.
We require the following of error states, without loss of generality
(e.g. error states can be implemented using sinks or a monitor process).

\begin{definition}\label{def:errors}
Let $\WW \defn \bigcup_t \WW_i$ be error states in a transaction system,
then for all $j\neq i$ and all $k$:

\begin{enumerate}
\item  $\WW_i\subseteq \NN_i$
	\hfill
	(errors are external states)
\item\label{item:errors}
	$\WW_i \lrestr \tatrans_i \rrestr \overline{\WW_i}=\emptyset$
	\hfill
	(local transitions preserve errors)
\item\label{item:errors2}
	$\tatrans_i \, \subseteq \WW_j^2 \cup \overline{\WW_j}^2$
	\hfill
	(transitions preserves remote (non)errors)
\item\label{item:errors3}
	$\cong_i\, \subseteq \WW_k^2 \cup \overline{\WW_k}^2$
	\hfill
	(bisimulations preserve (non)errors)
\end{enumerate}
\end{definition}

\begin{corollary}
If $q\in \WW_i$ for some $i$ in \autoref{lem:ots}, then
also $q'\in\WW_i$, owing to
\autoref{def:errors} \autoref{item:errors3}.
\end{corollary}

\begin{theorem}[reduction of interleavings]\label{th:dynred}~\\
Let $\TAS=\tuple{S, \to}$ be a parallel transaction system as in \autoref{def:tas}
with errors as in \autoref{def:errors} and
$\LL \defn \bigcup_t \LL_i$.

\noindent
Suppose that $p\tatrans^{*} q$, $p\in \NN$ and $q\in \WW$, then
there is
$p \trtrans^{*} q'$ s.t. $q\cong q'\in \WW$.
\end{theorem}
\begin{proof}
By \autoref{lem:ots},
the hypothesis implies that there is an ots under \threads:\\
$p = p_1 (\tatrans^{+}_{T_c})^* p_{k+1} = q_1 (\tatrans^{+}_{T_u})^* q_{l+1}$ 
with $q_{l+1} \cong q$ and $T_c \cup T_u = \threads$.

First, we show that there also is an ots ending in
$q_{l+1}$ with $q_{l+1}\notin \LL$.
We do this by completing paths that are stuck in the post phase (\LL),
starting with the thread that gets stuck the first (left-most in the path).
Let $x$ be the lowest index, such that there exists a subpath:
\[
p_{x} \tatrans_{\varphi(x) }^+ p_{x+1}
\tatrans_{\varphi(x+1)}^+ p_{x+2}
\text{~~with~~}  p_{x+1} \in \LL_{\varphi(x)}
\]
Let $i \defn \varphi(x)$.
By \autoref{item:invar} of \autoref{def:pas}, this subpath must exists when $q_{l+1} \in \LL_i$.
Hence, by construction, we have $p_x\in \NN$
(trivially from the cts definition, we have $p_x\in \NN_i$ and for
$j\neq i$ anything but $p_x \in \NN_j$ would contradict the assumption of
having chosen the lowest $x$).
Also, by \autoref{cor:suffix}, the suffix from $p_{x+1}$ does not contain
$i$ transitions.

Now we extend the path with $i$: $\LL_i \ni q_{l+1} \tatrans_i q''$.
This is possible by \autoref{item:postterm} of \autoref{def:tas}.
Using again strong induction with the same IH as in the proof of
\autoref{lem:ots}, we may move the new $i$ transition
`in place' as a first step towards constructing a new ots.
Thus using \autoref{item:lmover} from \autoref{def:tas} and
  \autoref{cor:bisim-threads}, one can construct a path:
  \[ p_1\tatrans_{\varphi(1)}^+\ldots\tatrans_{\varphi(x-1)}^+p_x\tatrans_{\varphi(x)}^+p_{x+1}\tatrans_i 
   \underbrace{p_{x+1}'\tatrans^* q'''}_A\cong q'' \]
  By \autoref{cor:split-cts}, the prefix path including $p_{x+1}$ is a cts, so
  the prefix path including $p_{x+1}'$ is also a cts
  (by \autoref{item:post} of \autoref{def:tas}).
  Moreover, since by \autoref{cor:split-ots} 
  the original suffix path $p_{x+1}\tatrans^* q_{l+1}$ is an ots
  under $T'= \threads \setminus \set i$.
  The threads that transit in $A$
  are still contained in $T'$ by virtue of \autoref{item:lmover}
  and \autoref{cor:bisim-threads}.
  Since $p_x \in \NN$,
  we also have $p_{x+1},p_{x+1}'\in \NN_{T'}$
  (by \autoref{item:invar} of \autoref{def:pas} and $i\notin T'$).
  
  Since $1\le x \le k$, the suffix path $A$ is shorter than the original path.
  Hence one can apply the (strong) induction hypothesis to $A$ with $T'$
  to obtain an ots under  $T'$ starting in $p'_{x+1}\in \NN_{T'}$.
  This yields a new ots under $\threads$:
  \[\overbrace{
  	p_1\tatrans_{\varphi(1)}^+\ldots\tatrans_{\varphi(x-1)}^+p_t
	\tatrans_{\varphi(t)}^+p_{x+1}\tatrans_i 
  	\underbrace{p_{x+1}'\tatrans^*_{T'} q''''}_{\text{ots under}~T'}
  }^{\text{ots under}~\threads} \]
  with $q'''' \cong q''' \cong q'' \cong q_{l+1} \cong q$.
  Because the bisimilarity preserves error states,
  we have $q''''\in \WW$.

  We must further have $q_{x+1}'\in \LL_i$ or $q_{x+1}'\in \NN_i$
  (note that the phase may have changed compared to $q''$ by the moving 
  operation).
  In the former case, we repeat this process of extending the path with $i$ 
  until  $q_{x+1}'\in \NN_i$
  (since the system is finite this happens eventually).
  In the latter case, we are done for this $i$, because $q''''\in \NN_i$.
  It follows that for all other $j\notin T'$, we also have $q''''\in \NN_j$.
  If $q''''\in \overline{\LL}$, we are done showing that there is also
  an ots path leading to an error without ending in a post phase.
  Otherwise, we pick a new left-most $x'$ and a new $i'$ to repeat the process.
  Because, the number of threads is finite, $i'\neq i$ and
  $x' > x$ this process eventually terminates.

Let
$p_1 (\tatrans^{+}_{T_c})^* p_{k+1} = q_1 (\tatrans^{+}_{T_u})^* q_{l+1}
\in \overline{\LL}$ 
be the ots after the above extension.
Since, we reduced the case to an ots where
$q_{l+1} \notin \LL$, it remains to show 
that there exists a transaction path  $\trtrans$ to some $q'\in \WW$.
We know that $q\in\WW_i\subseteq \NN_i$ for 
some thread $i$. Therefore, we have $i\notin T_u$ and $q_1 \in \WW_i
		\subseteq \NN_i$
by \autoref{item:invar} of \autoref{def:pas}.
Also, because $q_{l+1}\in \overline{\LL}$, we have $q_1 \in \NN$.
Let $q'\defn q_1 = p_{k+1}$.
It is easy to show that the this cts is also a
$\bigcup_i \trtrans_i$ path.
Hence, the conclusion is satisfied.
\end{proof}

\noindent
The following theorem shows that internal states may just as well be skipped.
This theorem is found in \autoref{sec:dynamic}~as~\autoref{th:blockred}.

\begin{theorem}[Atomic block reduction]
\label{th:blockred2}~\\
The block-reduced transition relation is defined~as:
\[
\brtrans_i \defn \NN_i\lrestr 
	(\trtrans_i\rrestr \overline{\NN_i})^* \trtrans_i \rrestr \NN_i
	\text{~~and~~}
	\brtrans \defn \bigcup_i \brtrans_i
\]
\noindent
There is path $p\tatrans^{*} q$, $p\in \NN$ and $q\in \WW$, then
there is
$p \brtrans^{*} q'$ s.t. $q\cong q'\in \WW$.
\end{theorem}
\begin{proof}
By \autoref{th:dynred}, the premise gives us that there is a transaction path
$p \trtrans^{*} q''$ s.t. $q\cong q''\in \WW$.
We induce backwards over that transaction path to show that it is a
block transition path $\brtrans$.
Let $x$ be an index s.t. the path contains a transition
$p_{x}\trtrans_i p_{x+1}$ with $p_{x}\notin \WW$ and $p_{x+1}\in \WW$.
We have $p_{x+1}\in \WW_i\subseteq \NN_i$, or else a contradiction with
\autoref{item:invar} (invariance).
By definition of $\tatrans$, we also have $p_{x+1}\in \NN$ and
$p_x \in \NN_{\threads\setminus\set i}$ by invariance.
If $p_x = p$, we are done since with $p\in \NN$ this is a block transition
(\brtrans).
Else either $p_x\in \NN_i$, and again we have a block transition or not.
In the first case, we repeat the process until eventually we hit
$p_{z} = p$ with $z < x$.
In the letter case, there exists a $y< x$ s.t.
$p_y \in \NN$  and all intermediate states in 
$\NN_{\threads\setminus\set i}$ by invariance. Clearly, this is also a
block transition step and one can repeat the process until we find some
$p_{z} = p$ with $z<y$.
Taking $q'\defn p_{x+1}$ therefore satisfies the induction hypothesis.
\end{proof}

\begin{theorem}[\autoref{th:blockred} in the main part]
\label{th:soundandcomplete}
Dynamic transaction reduction is sound and complete for the reachability of error states.
\end{theorem}
\begin{proof}
Completeness (no error states are missed) follows from
\autoref{th:dynred}(/\autoref{th:blockred2}). 
Soundness (no errors are introduced) follows immediately from the
fact that the reduced transition relation of \autoref{def:tas} is
a subset of the original transition relation $\tatrans$
(and the reduced relations of \autoref{th:blockred2} 
are subsets of the closure of the original transition relation: $\tatrans^+$).
\end{proof}

\newpage
\subsection{Instrumentation}
\label{a:instrument}

First, we define sufficient criteria for dynamic moving conditions
as a simple interface to easily design other heuristics. Then we give a
formal definition of the heuristics discussed in \autoref{s:dyn-movers}
and show that they satisfy the critaria.

\begin{definition}[Dynamic moving conditions]\label{def:dynamic}~\\
A state predicate (a subset of states)
$c^L_\alpha$ is a dynamic \emph{left}-moving
condition for a CFG edge $\tuple{l,\alpha,l'}\in\delta_i$,
if for all $j\neq i$, $\beta\in\Delta_j$:
$(c^L_\alpha \lrestr \tr\alpha{i}) \lcomm (c^L_\alpha \lrestr \tr\beta{j})$ 
and \tr{\beta}{j} preserves $c_\alpha=\true$, i.e.
$c^L_\alpha\lrestr \tr{\beta}{j} \rrestr \overline{c^L_\alpha}= \emptyset$.

A state predicate $c^R_\alpha$ is a dynamic \emph{right}-moving
condition for a CFG edge $\tuple{l,\alpha,l'}\in\delta_i$,
if for all $j\neq i$, $\beta\in\Delta_j$:
$(c^R_\alpha \lrestr \tr{\alpha}{i}) \rcomm  \tr\beta{j}$.
\end{definition}

\defmath\conflict{\mathsf{Conflict}}

To formulate the heuristics, we use some basic static analysis on the CFG
similar to e.g.~\cite[Sec. 3.4]{godefroid}.
The relation $\conflict$ on actions (relations on data) relates dependent
actions, i.e., $\conflict(\alpha,\beta)$ holds if $\alpha$ accesses
(reads or writes) variables that are written by $\beta$, or vice versa.
The requirement on the implementation of \conflict is such that:
\begin{align}
\neg\conflict(\alpha,\beta) \Longrightarrow
 \tr{\alpha}{}\bowtie \tr{\beta}{}
\label{eq:conflict}
\end{align}
We also write $\conflict(\alpha) = \set{\beta\mid \conflict(\alpha,\beta) }$.

\begin{table*}[ph!]
\caption{Some heuristics to establish dynamic commutativity conditions for a
CFG edge $(l_a,\alpha,l_b)\in \delta_i$ of thread $i$.
As the reachability heuristic is always applicable,
it can be considered used when the restrictions of the
other conditions do not hold on the CFGs.
\label{tab:heuristics}}
\label{t:dyn}
\begin{tabular}{p{2.5cm}|p{9.35cm}}
\toprule
Name
	& Dynamic condition $c_\alpha$ and  
	 static analysis on $\alpha$ and $\bigcup_{j\neq i} G_j$
\\
							
\midrule

\begin{minipage}{2.5cm}
	Reachability\\
	(for \autoref{fig:lazy})
\end{minipage}
	&
\begin{minipage}{9.35cm}
	The condition
	$c^{\mathit{reach}}_\alpha \defn \bigwedge_{j\neq i} \bigwedge_{l\in L(j)} (\pc_j \neq l)$
	guarantees that remote threads $j$ are not in certain locations $l\in L(j)$.
	The locations $l\in V_j$ considered are all locations that either:\\
	(1) have outgoing edges $\beta$ conflicting with $\alpha$, or\\
	(2) can reach another location $l_\beta\in V_j$ through $G_j$ where (1) holds.\\
	Therefore, $L(j)$ is defined as follows:\\
	$L(j) \defn \set{ l\in V_j \mid \exists
			 (l_\beta,\beta,l_\beta') \in \delta_j \colon \conflict(\alpha,\beta)
			 \land (l,l_\beta)\in \delta_j^* } $

	This heuristic can be applied for any action $\alpha\in\Delta_i$.
	When there are no conflicts, one obtains $L(j) = \emptyset$ and
	$c^{\mathit{reach}}_\alpha = \true$, yielding a static mover through our instrumentation
	(see \autoref{tab:instrument} and explanations).
\end{minipage}
\\

\hline

\begin{minipage}{2.5cm}
	Static pointer\\ dereference\\
	(for \autoref{ex:pointers})
\end{minipage}
	&
\begin{minipage}{9.35cm}
	The condition 	$c^{\mathit{deref}}_\alpha$ is only applicable when
	$\alpha$ is a pointer dereference of some pointer \ccode p, written here as $\alpha = \ccode{*p}$
	(the action might also be a modifying dereference such as $\alpha = \ccode{*p++}$ from
	\autoref{ex:pointers}).
	
	The condition further requires that all potentially conflicting actions $\beta$
	from other threads $j \neq i$ are also pointer dereferences of some pointer \ccode{p'}.
	It guarantees that the pointers:
		\begin{itemize}[nosep]
		\item[(1)] have a different value $\ccode{p} \neq \ccode{p'}$, and
		\item[(2)] the value of \ccode{p'} is not modified by some thread $k\neq i$ in the future.
		\end{itemize}
	The condition is defined as:\\
	$c^{\mathit{deref}}_\alpha \defn
		\bigwedge_{j\neq i} \bigwedge_{\tuple{\beta,\ccode{p'}}\in C_j(\alpha)}
		(\ccode{p} \neq \ccode{p'}) \land \bigwedge_{(k,l)\in F_i(\ccode{p'})} \pc_k \neq l$.\\
	Here, $C_j(\alpha)$ is the set of all conflicting actions that
	are pointer dereferences with their corresponding pointers:\\
	$C_j(\alpha) \defn \set{\tuple{\beta,\ccode{p'}} \mid 
			\beta \in\conflict(\alpha)\cap\Delta_j \land \beta=\ccode{*p'} }$.	\\
	And the set $F_i(\ccode{p'})$ holds pairs of thread ids $k$ and locations
	$l\in V_k$, such that $l$ can reach an action modifying the value of \ccode{p'}:\\
	$F_i(\ccode{p'}) \defn \{ (k,l) \mid k\neq i \land
		\exists (l_\gamma,\gamma,l'_\gamma)\in\delta_k,  (d,d')\in\gamma \colon
			d[\ccode{p'}]\neq d'[\ccode{p'}]$\\
	\vphantom{;}~~~~~~~~~~~~~~~~~~~~~~~$\land\, (l,l_\gamma)\in\delta_k^* \}$
	
	Note that the condition is restricted to actions $\alpha$ that are
	pointer dereferences and that all its conflicting actions $\beta$ must be pointer dereferences
	as well. It is easy to lose the second conjunct of this requirement, by
	carefully incorporating $c_\beta^{\mathit{reach}}$ in the condition.
	However, we did not do so to keep the presentation simple.
	Our implementation also only supports the restrictive, simpler version of
	this condition, making it less often applicable in real-world programs.
	For our examples, however, the simple version sufficed.
\end{minipage}
\\
\hline

\begin{minipage}{2.5cm}
	Monotonic atomic\\
	(for \autoref{ex:ht})
\end{minipage}
	&
\begin{minipage}{9.35cm}
	The condition requires that $\alpha$ and all its conflicting operations $\beta$
	are CAS checking for the same expected value $c$, i.e.: 
	$\alpha,\beta = \ccode{cas(p, c, x)}$ for some \ccode{p} and \ccode x.
	It is defined as 
	$c^{\mathit{atomic}}_\alpha \defn  (\ccode{*p} \neq \ccode{c})$ and
	guarantees that the CAS operation won't write to the location that \ccode p references.

\end{minipage}
\\

\bottomrule
\end{tabular}
\end{table*}

The conflict relation/function might over-estimate the set of conflicts
(see the implication in \autoref{eq:conflict}) as 
static analysis can be imprecise.
It should be noted that static analysis runs with a tight constraint on
computational resources. Typically, it is ran once over the syntactical
program structure to derive all the aliasing
constraints, and it completes in polynomial time---often quadratic or cubic---in
the input size, much unlike the expensive model checking procedure.
The heuristics we provide deal with the consequental imprecision of
static analysis, by deferring various judgements to the
model checker.



To reason over the control flow of the threads, we define location reachability.
Abusing notation slightly,
let $\delta_i^*$ be the transitive and reflexive closure of the CFG location relation 
$\set{(l,l') \mid \exists \alpha \colon (l,\alpha,l')\in \delta_i }$.

\autoref{tab:heuristics} provides formalized versions of the heuristics
presented in \autoref{s:dyn-movers}.
A static constraint on the CFG defines where each heuristic is applicable as
the table caption describes. Note that the table considers edges of the CFG of a thread~$i$.
For e.g. the dereference of a pointer \ccode p, we could write $d[d[\ccode p]]$,
where $d$ is a data assignment.
For simplicity, we assume that $\ccode p \neq \ccode{p'}$ ensures that
dereferencing these pointers does not conflict, i.e., references may not
partly overlap.
However, for clarity, we write 
\texttt c-style expressions, such as~\ccode{*p}.

\newpage
\begin{lemma}[\autoref{lem:dyncond} in the main part]~\\
The conditions in \autoref{t:dyn} are dynamic both-movers (\autoref{def:dynamicboth}).
\end{lemma}

\begin{proof}[Proof sketch]
We show that all dynamic both-mover constraints from
\autoref{def:dynamicboth} hold for all three heuristic
dynamic moving conditions $c_\alpha$ in the table, i.e.,
for all $i$, $j\neq i$ and all $\alpha\in\Delta_i$, $\beta\in\Delta_j$:
\begin{itemize}
\item (1) the predicate $c_\alpha$ should not be disabled by transition \tr{\alpha}i,
\item (2) the predicate $c_\alpha$ should not be disabled by transition \tr{\beta}j, and
\item (3) the actions $c_\alpha\lrestr \tr{\alpha}i$ and $c_\alpha\lrestr \tr\beta{j}$
should (both) commute.
\end{itemize}

\noindent
In the following cases,
let $(l_\alpha, \alpha, l'_\alpha)\in\delta_i$ and  $(l_\beta, \beta, l'_\beta)\in\delta_j$.

\setlist[description]{font=\normalfont\itshape}
\begin{description}
\item[Reachability]:
\begin{itemize}
\item[(1)] As $c_\alpha := c_\alpha^{\mathit{reach}}$ only refers to control location of $j$, $\tr{\alpha}i$ cannot disable it.

\item[(2)]
Assume that we have
$c_\alpha\lrestr \tr{\beta}j \rrestr \overline{c_\alpha}\neq \emptyset$.
There are states $(\pc,d) \in c_\alpha$ and $(\pc',d') \notin c_\alpha$
such that  $(\pc,d)\tr{\beta}j (\pc',d')$.
Clearly, the conjuncts of $c_\alpha$ contain $\pc_j \neq l_\beta'$
(the first one to be invalidated to make $(\pc',d') \notin c_\alpha$).
However, by the additionally required upwards closure of $\delta_j^*$ enforces that
$\pc_j \neq l_\beta$ is also part of $c_\alpha$,
contradicting that $(\pc,d) \in c_\alpha$.

\item[(3)]

Assume that 
the condition includes the start location of $\beta$:
$c_\alpha = \dots \land \pc_j \neq l_\beta \land \dots$.
Therefore, $c_\alpha \lrestr \tr\beta{j} = \emptyset$. This makes commutativity trivially hold:
$(c_\alpha \lrestr \tr\beta{j})(c_\alpha \lrestr \tr\alpha{i}) =
(c_\alpha \lrestr \tr{\alpha}i)(c_\alpha \lrestr \tr{\beta}j) = \emptyset$.

Now take the complementary assumption. We have $\beta\notin \conflict(\alpha)$
and therefore commutativity.

\end{itemize}

\item[Static pointer dereferences]:
\begin{itemize}
\item[(1)] The condition $c_\alpha := c^{\mathit{deref}}_\alpha$ only checks the program counter
of other threads and compares pointer values.
As $\tr{\alpha}i$ only dereference one of those pointers,
the transition cannot disable the condition on that account.
Moreover, since the location checks only involve remote threads,
the pc update cannot disable $c_\alpha$ either.

\item[(2)]
This follows from a similar argument as for \textit{Reachability} Case (2),
as also this heuristic includes the upwards closure of edges
changing the pointer value that is dereferenced.
Moreover, since under this closure the pointer values do not
change value, their dereferences can also not start conflicting.

\item[(3)]
This follows from a similar argument as for \textit{Reachability} Case (3).

\end{itemize}
%

\item[Monotonic atomic]:
\begin{itemize}
\item[(1)] As $c_\alpha := c_\alpha^{\mathit{atomic}}$ ensures that the expected value check of the
compare and swap operation fails, no write will occur and $c_\alpha$
remains enabled after $\alpha$.

\item[(2)]
For all states $\sigma\in c_\alpha$, it holds that there is no $\beta$-transition
to a state $\sigma'\notin c_\alpha$, as the only conflicting operations with $\alpha$
are compare and swap operations that check for the same constant value.
Since the expected value of those CAS operations does not agree with 
the pointer dereference, these operations cannot conflict with it.

\item[(3)]
Follows from a similar argument as in Case 2.
\end{itemize}

%
\end{description}

\end{proof}

\begin{corollary}~\\
The conditions in \autoref{t:dyn} are dynamic left and right movers (\autoref{def:dynamic}).
\end{corollary}
\begin{proof}
\autoref{def:dynamicboth} is stronger than the conjunction of
dynamic left and right movers from \autoref{def:dynamic} (taking $c_\alpha^L=c_\alpha^R = c_\alpha$). 
\end{proof}

\begin{table}[t]
\caption{Instrumentation with dynamic right/left movers}
\label{tab:instrument}
\begin{tabular}{l|p{4.3cm}|l}
&$G_i\defn (V_i,\delta_i) $&
$V_i',\delta_i'$ in $G_i'$ (pictured)
\\\hline

R0&
$\forall 
l_a\in V_i\colon$
&
\raisebox{-.5\height}{
\begin{tikzpicture}\centering
  \tikzstyle{e}=[->]
  \tikzstyle{every node}=[font=\small, node distance=1.5cm]

  \node (n) at (-1,0) {$l_a^{N}$};
  \node (r) at ( 0,0) {$l_a^{R}$};
  \node (r) at ( 1,0) {$l_a^{L}$};
  \node (r) at ( 2,0) {$l_a^{\Rn}$};
  \node (r) at ( 3,0) {$l_a^{\Ln}$};
\end{tikzpicture}
}
\\\hline

R1&
$\forall \tuple{l_a, \alpha, l_b} \in \delta_i
\colon$ &  
\raisebox{-.5\height}{
\begin{tikzpicture}\centering
  \tikzstyle{e}=[->]
  \tikzstyle{every node}=[font=\small, node distance=1.5cm]

  \node (x) at (-1,0) [anchor=east] {\hphantom{$l_a^N$}};
  \node (n) at (1.2,0)  {$l_a^N$};
  \node (rN) at (3, .5) {$l_b^R$};
  \node (lN) at (3,-.5) {$l_b^L$};
  \path (n) edge[e,bend left =15] node[above,sloped] {$c^R_\alpha \lrestr \alpha $} (rN);
  \path (n) edge[e,bend right=15] node[below,sloped] {$\lnot c^R_\alpha \lrestr \alpha $} (lN);
\end{tikzpicture}
}\\\hline

R2&
$\forall \tuple{l_a, \alpha, l_b} \in \delta_i
\colon$ &  
\raisebox{-.5\height}{
\begin{tikzpicture}\centering
  \tikzstyle{e}=[->]
  \tikzstyle{every node}=[font=\small, node distance=1.5cm]

  \node (x) at (-1,0) [anchor=east] {\hphantom{$l_a^N$}};
  \node (rr) at (1.2, 0) {$l_a^{\Rn}$};
  \node (rN) at (3, .5) {$l_b^R$};
  \node (lN) at (3,-.5) {$l_b^L$};
  \path (rr) edge[e,bend left =15] node[above,sloped]
  			{$c^R_\alpha \lrestr \alpha$} (rN);
  \path (rr) edge[e,bend right=15] node[below,sloped]
  			{$\lnot c^R_\alpha \lrestr \alpha$} (lN);
\end{tikzpicture}
}\\\hline

R3&
$\forall 
l_a\in V_i\colon$
&
\raisebox{-.5\height}{
\begin{tikzpicture}\centering
  \tikzstyle{e}=[->]
  \tikzstyle{every node}=[font=\small, node distance=1.5cm]

  \node (r) at (-1,0) [anchor=east] {$l_a^R$};
  \node (rp) at (1.2, 0) {$l_a^{\Rn}$};
  \path (r) edge[e] node[above,sloped] {\true} (rp);
\end{tikzpicture}
}
\\\hline

R4 \& R5
&
$\forall 
l_a\in V_i\setminus\LFS_i\colon$
&
\raisebox{-.5\height}{
\begin{tikzpicture}\centering
  \tikzstyle{e}=[->]
  \tikzstyle{every node}=[font=\small, node distance=1.5cm]

  \node (l) at (-1,0) [anchor=east] {$l_a^L$};
  \node (ll) at (1.2, .5) {$l_a^{\Ln}$};
  \node (lN) at (1.2,-.5) {$l_a^{N}$};
  \path (l) edge[e,bend left =15] node[above,sloped]
  			{$c^L_{\alpha_1} \land ..\land c^L_{\alpha_n}$} (ll);
  \path (l) edge[e,bend right=15] node[below,sloped]
  			{$\lnot (c^L_{\alpha_1} \land ..\land c^L_{\alpha_n}) $} (lN);   
  \node (lL1) at (3, .85) {$l_{b1}^L$};
  \node (lLd) at (3, .5) {$\dots$};
  \node (lLn) at (3, .15) {$l_{bn}^L$};
  \node (XX) at (3,-.75) {\parbox{2.4cm}{with~$\tuple{l_a, \alpha_1, l_{b1}},..$\\
  									$,\tuple{l_a, \alpha_n, l_{bn}}\in \delta_i$}};
  \path (ll) edge[e] node[above,sloped] {$\alpha_1$} (lL1);
  \path (ll) edge[e,dotted] node[above] {} (lLd);
  \path (ll) edge[e] node[below,sloped] {$\alpha_n$} (lLn);
\end{tikzpicture}
}
\\\hline

%

R6&
$\forall 
l_a\in \LFS_i\colon$
&
\raisebox{-.5\height}{
\begin{tikzpicture}\centering
  \tikzstyle{e}=[->]
  \tikzstyle{every node}=[font=\small, node distance=1.5cm]

  \node (l) at (-1,0) [anchor=east] {$l_a^L$};
  \node (n) at (1.2, 0) {$l_a^{N}$};
  \path (l) edge[e] node[above,sloped] {\true} (n);
\end{tikzpicture}
}
\\\hline
\end{tabular}
\end{table}

(Note that  the proof of \autoref{th:instrument} uses \autoref{def:dynamic}, but
also holds with the stronger \autoref{def:dynamicboth}
---a hint is provided in the figure in the proof of \autoref{th:instrument}
case \autoref{def:tas}.2.R1.
This is because in addition to a stronger commutativity condition,
which also requires that $(c_\alpha \lrestr \tr\alpha{i}) \rcomm (c_\alpha \lrestr \tr\beta{j})$,
it also requires that $c_\alpha \lrestr \tr{\alpha}i \rrestr c_\alpha$.)

Let $G_i = (V_i, \delta_i)$ be the CFG for thread $i$.
We transform this into an instrumented CFG
$G'_i\defn (V_i', \delta'_i)$
by copying all locations $l_a\in V_i$ to pre-commit, post-commit, and external locations: $l^R_a,l^L_a,l^N_a$, as well as introducing intermediate states (see~\autoref{sec:instrument}). The edges in the instrumentation $G'_i$ as given in~\autoref{tab:instrument}. Note that all discussion on the instrumentation from the main paper holds here as well, and additionally we distinguish here between dynamic left/right/both-moving conditions. They are defined in the same way as dynamic both-moving conditions~\autoref{def:dynamicboth} by replacing the commuting condition to left/right-commuting for left/right-moving conditions.


\begin{theorem}[\autoref{lem:instrument2} in the main part]\label{th:instrument}
Let \tatrans be the transition relation of an instrumented system.\footnote{Within this theorem only, to simplify writing we use $\tatrans, \delta, V_i$ instead of $\tatrans', \delta_i', V_i'$.} Let:
\begin{align*}
\WW_i
	&\defn\bigr\{ (\pc,d) \mid
				\pc_i\in \{l_{sink}^N,l_{sink}^R,l_{sink}^L \} \bigr\}
	&\text{(Error)} \\
\RR_i
	&\defn \bigr\{ \tuple{\pc, d} \mid \pc_i \in
					\{l^R,l^{\Rn}\} \bigr\} \setminus \WW_i
	&\text{(Pre-commit)} \\
\LL_i
	&\defn \bigr\{	
		\tuple{\pc, d} \mid \pc_i \in \{l^{L}, l^{\Ln} \}\bigr\} \setminus \WW_i
	&\text{(Post-commit)} \\
\EX_i
	&\defn \bigr\{ \tuple{\pc, d}  \mid \pc_i \in
  				\{l^{N}\}\bigr\}  \setminus \WW_i
	&\text{(Ext./non-error)} \\
\NN_i
	&\defn \EX_i \uplus \WW_i
	&\text{(External)}
\end{align*}

\noindent
We define an equivalence relation $\cong_i$ over the locations
\begin{align*}
  \cong_i\defn&\set{ (l^X,l^Y)\mid l\in V_i\land X,Y\in\set{L,R}}\cup\\
  &\bigl\{(l^{X},l^{Y})\mid \exists l\in V_i\colon X,Y\in\set{\Rn,\Ln,N}\bigr\}
\end{align*}
and lift it to semantic states: 
  \begin{align*}
    \cong_i \defn \bigl\{ ((\pc,d),(\pc',d'))&\mid
    	 d=d'
    \land \pc_i\cong_i\pc'_i\\
    & \land\forall j\neq i\colon \pc_j=\pc_j' \bigr\}.
  \end{align*}
We show that \autoref{def:pas}, \autoref{def:tas} and \autoref{def:errors} hold.
\end{theorem}
\begin{proof}

We show that all five conditions of \autoref{def:pas} are satisfied:\\
  \begin{enumerate}
  \item By definition, $\RR_i$, $\LL_i$ and $\NN_i$ are pair-wise disjoint.
  \item By definition, all $V_i$ are pairwise disjoint.
  \item By definition of $\tatrans_i$, the phases of $j\neq i$ remain the same.
  \item By definition of $\cong_i$, equivalent states must have the same phase for all $j\neq i$.
  \item Let $(\pc^a,d^a)\cong_i (\bar\pc^a,\bar d^a)$ and $(\pc^a,d^a)\tatrans (\pc^b,d^b)$.
  We show that there is some $ (\bar\pc^b,\bar d^b)$ such that
  $(\bar\pc^a,\bar d^a) \tatrans (\bar\pc^b,\bar d^b)$ and 
  $(\pc^b, d^b)\cong_i (\bar\pc^b,\bar d^b)$.
  
  We have some $j$ s.t. $(\pc^a,d^a)\tatrans_j (\pc^b,d^b)$.
  If $j\neq i$ then $\pc^a_i = \pc^b_i$ and the hypothesis holds trivially from the 
  definitions (because $\pc^a_j = \bar\pc^a_j$ and
  	$\pc^a_i = \pc^b_i$ together with $\bar\pc^a_i = \bar\pc^b_i$).
  
  Otherwise, $\pc^a_i \cong_i \bar\pc^a_i$ and $\bar d_a = d_a$.
  We inspect all possible CFG edges in the instrumentation.
  \begin{itemize}
  \item[R1:] We have $\pc^a_i = l^N$ for some $l\in V_i$ then we have either
  	$\pc^b = \pc^a[i := l^R]$ or $\pc^b = \pc^a[i := l^L]$ where
	$l^R$ and $l^L$ are bisimilar.
	Moreover, we have $d^b = \alpha(d^a)$ where $\alpha(d^a)$  is the result of 
	applying action $\alpha$ on $d^a$.
\begin{center}
\begin{tikzpicture}
   \tikzstyle{e}=[minimum width=1cm]
   \tikzstyle{ll}=[font=\scriptsize, node distance=2cm]
   \tikzstyle{every node}=[font=\small, node distance=.8cm]

  \node (s1) {$(\pc^a,d^a)$};
  \node (s2) [right of=s1, xshift=.9cm] {$(\pc^b,d^b)$};
  \path (s1) -- node[pos=.45]{$\tatrans_j$} (s2);
  \node (s4) [node distance=1.2cm,below of=s1] {$(\bar \pc^a, d^a)$};
  \path (s1) -- node[midway,sloped]{$\cong_i$} (s4);

  \node (s3) [gray,node distance=1.2cm,below of=s2] {$(\bar \pc^b,\bar d^b)$};
  \path (s4) -- node[gray,pos=.48]{$\tatrans_j$} (s3);
  \path (s2) -- node[midway,gray,sloped]{$\cong_i$} (s3);

  \node (S1) [ll,left of=s1] 		{$\pc^a_i = l_a^N$};
  \node (S2) [ll,right of=s2] 		{$\pc^b_i \in \{l_a^R,l_a^L\}$};
  \node (S4) [ll,left of=s4]	 	{$\bar\pc^a_i\in\{l_a^N,l_a^{\Rn},l_a^{\Ln}\}$};
  \node (S3) [ll,right of=s3,gray]	{$\bar\pc^b_i = ?, \bar d^b = ?$};

  \path (s1) -- node[sloped,pos=.48]{$\implies$} (s3);
\end{tikzpicture}
\end{center}

  	For all values of $\bar\pc^a_i$, we show that there exists a
	$(\bar\pc^b,\bar d^b)$ satisfying the hypothesis.
    \begin{itemize}
    \item[R1:] $l_a^{N}$: Trivially, as $(\pc^a,d^a) = (\bar \pc^a,\bar d^a)$.
    \item[R2:] $l_a^{\Rn}$: As before, all target states are bisimilar.
    		Now we have either 
  		$\bar\pc^b_i = l_b^R$ or $\bar\pc^b_i = l_b^L$. Nonetheless, it holds that
		$(\bar\pc^b,\bar d^b) \cong_i(\pc^b, d^b)$.
    \item[R4,5:] $l_a^{\Ln}$: Since, for all $(l_a, \beta,l_b) \in \delta_i$ there is a
    				$(l_a^{\Ln}, \beta,l_b^{L}) \in \delta_i'$, there is also
				$(l_a^{\Ln}, \alpha, l_b^{L}) \in \delta_i'$.
				After this edge, we have
				$\bar d^b = \alpha(\bar d^a)= \alpha(d^a) = d^b$, and also
				$(\bar\pc^b,\bar d^b) \cong_i(\pc^b, d^b)$.
    \end{itemize}
	In all cases, it follows that $(\pc^b,d^b)\cong_i (\bar\pc^b,\bar d^b)$.

  \item[R2:] We have $\pc^a_i$ is $l^{\Rn}$, and a similar argument can be used as for R1.
  \item[R3:] We have $\pc^a_i$ is $l^{R}$, $\pc^b_i$ is $l^{R'}$ and $d^a = d^b$.
  			 Bisimilarity $\pc^a \cong_i \bar\pc^a$ gives us three cases for $\pc^a_i$:
			 \begin{itemize}
		     \item[R3:] $\pc^a_i = \bar\pc^a_i$ and $\bar d^b = d^b$ follows trivially.
		     \item[R4,5:] $\pc^a_i = l^L$, because either edge of R4,5 is taken and
		     				neither performs an action,
		     			we get $\bar d^b = d^b$ and $\bar\pc^b_i=\pc^b[i := l^{L'}]$ or
						$\bar\pc^b_i=\pc^b[i := l^N]$.
		     \item[R5:] $\pc^a_i = l^L$, because both edges (R3 and R5)
		     			perform no action and are always enabled,
		     			we get $\bar d^b = d^b$ and $\bar\pc^b_i=\pc^b[i := l^N]$.
			 \end{itemize}					     
			In all cases, it follows that $(\pc^b,d^b)\cong_i (\bar\pc^b,\bar d^b)$.
    
  \item[R4,5:] If $\pc^a_i$ is $l^{L}$ with $l\notin\LFS_i$, then we have:
\begin{center}
\begin{tikzpicture}
   \tikzstyle{e}=[minimum width=1cm]
   \tikzstyle{ll}=[font=\scriptsize, node distance=2cm]
   \tikzstyle{every node}=[font=\small, node distance=.8cm]

  \node (s1) {$(\pc^a,d^a)$};
  \node (s2) [right of=s1, xshift=.9cm] {$(\pc^b,d^b)$};
  \path (s1) -- node[pos=.45]{$\tatrans_j$} (s2);
  \node (s4) [node distance=1.2cm,below of=s1] {$(\bar \pc^a, d^a)$};
  \path (s1) -- node[midway,sloped]{$\cong_i$} (s4);

  \node (s3) [gray,node distance=1.2cm,below of=s2] {$(\bar \pc^b,\bar d^b)$};
  \path (s4) -- node[gray,pos=.48]{$\tatrans_j$} (s3);
  \path (s2) -- node[midway,gray,sloped]{$\cong_i$} (s3);

  \node (S1) [ll,left of=s1] 		{$\pc^a_i = l_a^L$};
  \node (S2) [ll,right of=s2] 		{$\pc^b_i \in \{ l_a^{\Ln},l_a^N \}$};
  \node (S4) [ll,left of=s4]	 	{$\bar\pc^a_i \in \{l_a^{R},l_a^{L}\}$};
  \node (S3) [ll,right of=s3,gray]	{$\bar\pc^b_i = ?, \bar d^b = ?$};

  \path (s1) -- node[sloped,pos=.48]{$\implies$} (s3);
\end{tikzpicture}
\end{center}
	Since $\pc^b_i \in \{ l_a^{\Ln},l_a^N \}$,
	it must be the case that $\bar \pc^b_i \in \{ l^N_b, l_b^{\Ln}, l_b^{\Rn} \}$
	according to \autoref{tab:instrument}.
	By definition, we have $l_a^{\Ln} \cong_i l_a^N \cong_i l_b^{\Rn}$.
	None of the actions from above edges modifies any data.
	Hence, we also have $d^a = d^b = \bar d^a = \bar d^b$.
	Thus, we may conclude that $(\bar\pc^b,\bar d^b) \cong_i(\pc^b, d^b)$.
  \item[R6:] If $\pc^a_i$ is $l^{L}$ and $l\in\LFS_i$, we may apply a similar argument as for R3.
  \end{itemize}
  \end{enumerate}

  We show that all four conditions of \autoref{def:tas} are satisfied using the
  bisimulations established above for all $i$ and $j\neq i$:
  \begin{enumerate}
  \item Since the instrumentation doesn't contain any edge from some
  $l^L,l^{\Ln}\in V_i'$
  	to some $l^R, l^{\Rn} \in V_i'$,
	this condition is fulfilled by the definition of $\tatrans_i$.
  \item We look at all CFG edges ending in locations that constitute
  states $\tuple{\pc,d}\in\RR_i$
  and see if the edges are indeed right-movers up to $\cong_j$.
  Note that we reason separately over transitions and their actions, since
  the actions might commute, but the transitions may not perfectly commute
  if they end up taking a different branch of the instrumentation
  (e.g. see R1, R2, R4 and R5 in \autoref{tab:instrument}).
  Within this case and the next case only,
  we use $\tatrans', \delta_i', V_i', G_i'$ again to refer to the
  instrumented CFG to distinguish it from its original $\tatrans, \delta, V_i, G_i$.
  
    \begin{itemize}
    \item[R1] We have $\pc_i=l_a^{R}$:
      Only $G_i'$-edges $\tr{c_\alpha^R \lrestr \alpha}{}_i'$ reach the $l_a^R$ location
	(with $\tr{c^R_\alpha \lrestr \alpha}{}_i' = c^R_\alpha \lrestr \tr{ \alpha}{i}_i'$).
		Let $\sigma_1 \tr{c^R_\alpha \lrestr \alpha}{}_i \sigma_2$ with
		$\exists d \colon \sigma_2 = (\pc,d)$.
	    Also let $\sigma_2 \tr{\beta}{}_j' \sigma_3$.
      Indeed, only a path $\sigma_1 \tr{c^R_\alpha \lrestr \alpha}{}_i \sigma_2
      						\tr{\beta}{}_j' \sigma_3$
      fulfills the premise of the right-moving condition (recall \autoref{eq:commute}).
      Now we show that its conclusion is satisfied up to $\cong_i$, by returning
      our attention to the properties that hold in the original CFG $G_i$.

	   By \autoref{def:dynamicboth}, we have that $\tr{\alpha}i$ preserves $c^R_\alpha$ 
     with $c_\alpha = c_\alpha^R$.
     Therefore, we also obtain $\sigma_2\in c_\alpha^R$ and the valid path refinement
      $\sigma_1 \tr{c_\alpha^R \lrestr \alpha}{}_i' \,\sigma_2
          \tr{c_\alpha^R \lrestr \beta}{}_j' \,\sigma_3$ in the instrumented system
          (occurrences of control locations in $c_\alpha^R$ need to be updated to
          refer to all of their counterparts in the instrumented system).
       By \autoref{def:dynamicboth}, we therefore also have
       $ \tr{c_\alpha \lrestr\alpha}{i}\rcomm \tr\beta{j}$,

      It remains to show that the side effects of the pc update in the instrumented system
		preserve commutativity up to~$\cong_j$.
       \autoref{app:dynmovers} shows that the other thread $j$ might lose its
      phase in the moving process, because it might be in a location
      $l^L$, $l^N$ or $l^{R'}$ and R1, R2, R3 or R4 of \autoref{tab:instrument}
      could be forced in a different branch. The location of $i$ however remains
      unaffected.
	
	Therefore, we have $\tr{ c_\alpha \lrestr \alpha}{}_i' \rcomm_j \tr{\beta}{}_j'$:

      \begin{tikzpicture}
   \tikzstyle{e}=[minimum width=1cm]
   \tikzstyle{ll}=[font=\scriptsize, node distance=1cm]
   \tikzstyle{every node}=[font=\small, node distance=1.7cm]

  \node (s1) {$\sigma_1$};
  \node (s2) [below of=s1] {$\sigma_2$};
  \path (s1) -- node[sloped,pos=.45]{$\tr{c^R_\alpha \lrestr \alpha}{}_i'$} (s2);
  \node (s3) [right of=s2] {$\sigma_3$};
  \path (s2) -- node[midway,sloped]{$\tr{\beta}{}_j'$}  (s3);

  \node (s2p) [gray,right of=s1] {$\sigma_2'$};
  \path (s1) -- node[midway,gray,sloped]{$\tr{\beta}{}_j'$}  (s2p);
  \node (s3p) [gray,below right of=s2p] {$\sigma_3'$};
  \path (s2p) -- node (XX)[midway,gray,sloped]{$\tr{c^R_\alpha \lrestr \alpha}{}_i'$} (s3p);

  \node (S1) [ll,left of=s1] 		{$c^R_\alpha \ni$};
  \node (S2) [ll,gray,left of=s2] 		{$c^R_\alpha \ni$};
  \node (S4) [ll,below of=s2p,gray]	{$$};
  \node (S3) [ll,right of=s3p,gray]	{$$};

  \path (s2) -- node[sloped,pos=.48,allow upside down]{$\implies$} (s2p);
  \path (s3) -- node[sloped,pos=.48,allow upside down]{$\cong_j$} (s3p);
  
	\end{tikzpicture}
     
    \item[R2] We again have $\pc_i=l_a^{R}$ and can apply a similar argument as for R1.
    \item[R3] We have $\pc_i=l_a^{\Rn}$:
    		The action associated with the edge performs no state modification and is always 
    			enabled. It is therefore a both-mover.
			
			Because the edge is an internal edge added by our instrumentation,
		    i.e. from original location $l_a$ (copied to $l_a^R$) to
      		$l_a$ (copied to $l_a^{\Rn}$), the conditional movers of
      		other threads cannot detect it and the corresponding transition
			is a (perfect) both-mover, i.e., up to $\cong_{\set{}}$.
    \end{itemize}
As \autoref{def:dynamicboth} is stronger than both the left and the right mover case in
\autoref{def:dynamic}, this proof also holds for the latter.

  \item Similarily, we look at all the edges from locations constituting 
  	states $\tuple{\pc,d}\in\LL_i$ and show that they are left-mover up to $\cong_{\set{i,j}}$.
	Note again that we reason separately over transitions and there actions.
	  Within this case and the previous case only,
  we use $\tatrans', \delta_i', V_i', G_i'$ again to refer to the
  instrumented CFG to distinguish it from its original $\tatrans, \delta, V_i, G_i$.

    \begin{itemize}
        \item[R4] We have $\pc_i=l_a^{L}$:
        There are only two edges leaving this $l_a^{L}$ location ($l_a\notin\LFS_i$).
      Because it is always possible to take one of the edges to either
      $l_a^N$ or $l_a^{\Ln}$ and both target locations are bisimilar, and, since the 
      edges do not change the state, nor can together be disabled,
      the edges left move up to $\cong_i$.
      
      Because the edge is an internal edge added by our instrumentation,
      i.e. from original location $l_a$ (copied to $l_a^L$) to
      $l_a$ (copied to $l_a^{L'}$ or $l_a^N$), the conditional movers of
      other threads cannot detect it and the corresponding transition
      is a both-mover up to $\cong_i$.

    \item[R5] We have $\pc_i=l_a^{\Ln}$:
		For each $G_i'$-edge $\tr{c_\alpha^L \lrestr \alpha}{}_i'$ leaving this location,
		we show that it is a left mover up to $\cong_{\set{i,j}}$.
		Let $\sigma_2 \tr{c_\alpha^L \lrestr \alpha}{}_i' \sigma_3$ with $\exists d \colon \sigma_2 = (\pc,d)$.
	    Also let $\sigma_1 \tr{\beta}{}_j' \sigma_2$.
      	Indeed, only a path $\sigma_1 \tr{\beta}{}_j' \sigma_2 \tr{c_\alpha^L \lrestr \alpha}{}_i' \sigma_3$
      	fulfills the premise of the left-moving condition
      	(recall the dual of \autoref{eq:commute}).
	    Now we show that its conclusion is satisfied up to $\cong_{\set{i,j}}$.

		Because $\pc_i=l_a^{\Ln}$, we also have $\pc_i' = l_a^{\Ln}$ with
		$\exists d' \colon \sigma_1 = (\pc',d')$.
		Therefore, there must be a path $\sigma\tr{c^L_{\alpha_1} \land ..\land c^L_{\alpha_n}}{}_i'
		\sigma'(\to')^*\sigma_1$, i.e., including the action of the positive branch of R4
		whose guard is a conjunction of dynamic left-moving conditions
		including the dynamic left mover for~$\alpha$.
		
	    By \autoref{def:dynamicboth}, we obtain $\sigma_1\in c_\alpha^L$ with $c_\alpha^L=c_\alpha$,
	    as the dynamic left-moving condition may not be disabled by remote threads $j\neq i$.
	    Therefore, we can safely refine $\tr{\beta}{}_j'$ to  $c_\alpha^L \lrestr\tr{ \beta}{}_j'$.
	    Also by \autoref{def:dynamicboth}, we therefore also have that
	    $\tr{c_\alpha^L \lrestr\alpha}{i}\lcomm \tr{\beta}{j}$.
	    
	    It remains to show how the pc change in the instrumented system preserves
	    the commutativity up to $\cong_{\set{i,j}}$.
	    Again refer to \autoref{app:dynmovers}
	    for an overview of all possibilities. However,
	    because only $i$ and $j$ are involved in the move, only those two threads could
	    change phase.
		Therefore, it is easy to see that
			    $\tr{c_\alpha^L \lrestr \alpha}{i}\lcomm_{\set{i,j}}\tr{\beta}{j}$
			    is met.

            \begin{tikzpicture}
   \tikzstyle{e}=[minimum width=1cm]
   \tikzstyle{ll}=[font=\scriptsize, node distance=1cm]
   \tikzstyle{every node}=[font=\small, node distance=1.7cm]

  \node (s1) {$\sigma_1$};
  \node (s2) [right of=s1] {$\sigma_2$};
  \path (s1) -- node[pos=.45]{$\tr{\beta}{}_j'$} (s2);
  \node (s3) [below right of=s2] {$\sigma_3$};
  \path (s2) -- node[midway,sloped]{$\tr{c^L_\alpha \lrestr \alpha}{}_i'$}  (s3);

  \node (s2p) [gray,below of=s1] {$\sigma_2'$};
  \path (s1) -- node[midway,gray,sloped]{$\tr{c^L_\alpha \lrestr \alpha}{}_i'$}  (s2p);
  \node (s3p) [gray,right of=s2p] {$\sigma_3'$};
  \path (s2p) -- node[midway,gray,sloped]{$\tr{\beta}{}_j'$}  (s3p);

  \node (S1) [ll,gray,left of=s1] 		{$c^L_\alpha \ni$};

  \path (s2) -- node[sloped,pos=.48,allow upside down]{$\implies$} (s2p);
  \path (s3p) -- node[sloped,pos=.48,allow upside down]{$\cong_{\set{i,j}}$} (s3);
\end{tikzpicture}

\item[R6] By a similar argument as for case R3 (Condition 2 of \autoref{def:tas}),
			this edge is a perfect both-mover.
    
     \end{itemize}
As \autoref{def:dynamicboth} is stronger than both the left and the right mover case in
\autoref{def:dynamic}, this proof also holds for the latter.

  \item Because we require for every cycle in the $\mathrm{CFG}_i$ that there is at least one state in the location feedback set $\mathrm{LFS}_i$, we know that there is no cycle in $\mathrm{CFG}_i'$ such that every location is an $\LL_i$-state.
  \end{enumerate}

  We show that all two conditions of \autoref{def:errors} are satisfied:
  \begin{enumerate}
  \item Because $\set{l_{sink}^N,L_{sink}^L,l_{sink}^R}\subseteq\NN_i$, we have $\WW_i\subseteq\NN_i$.
  \item By assumption, there is no transition from a sink state to a non-sink state, so all states following an error state are error states as well.
  \item Same argument as before.
  \item Because sink locations are only $\cong_i$-equivalent to other sink locations, the simulation preserves error states.
  \end{enumerate}

This proves the hypothesis.
\end{proof}

\begin{lemma}[\autoref{lem:instrument} in the main part]\label{lem:instrument3}
An error state is $\to$-reachable in the original system
iff an error state is $\to'$-reachable
in the instrumented system.
\end{lemma}

\begin{proof}
Actually, a more general statement holds (in the left-to-right direction): If
$(\pc_0, d_0) \to^* (\pc_n, d_n)$ in the original system $C$, then there is an 
annotated state with the same data value
$(\bar\pc_n, d_n)$ such that  $(\bar\pc_0, d_0) \to'^{*} (\bar\pc_n,d_n)$ in the 
instrumented system $C'$ and moreover, for all $i$, 
$\bar\pc_{n,i} = \pc_{n,i}^X$ for $X \in \set{R,L, N}$. Here, $(\bar\pc_0, d_0)$ is 
the initial state of the instrumented system, i.e., $\bar\pc_{0,i} = \pc_{0,i}^N$. 
We prove this property by induction on the length of the path in the original 
system ($n$). 

For a path of length $0$, we have $(\pc_n, d_n) = (\pc_0,d_0)$ and the statement 
holds for $(\bar\pc_n, d_n) = (\bar\pc_0, d_0)$.
Assume the statement holds for a path of length $n$, and assume $(\pc_0, d_0) \to^* (\pc_{n+1}, d_{n+1})$ in the original system via a path of length $n+1$. Then there 
is a state $(\pc_n,d_n)$ such that $(\pc_0, d_0) \to^* (\pc_{n}, d_{n})$ via a path 
of length $n$ and $(\pc_{n}, d_{n}) \to (\pc_{n+1}, d_{n+1})$ in $C$. By the 
inductive hypothesis, $(\bar\pc_0, d_0) \to'^{*} (\bar\pc_n,d_n)$ in the 
instrumented system $C'$ such that, for all $i$, 
$\bar\pc_{n,i} = \pc_{n,i}^X$ for $X \in \set{R,L, N}$. Let $j$ and $\alpha$ be 
such that $(\pc_{n}, d_{n}) \to (\pc_{n+1}, d_{n+1})$ due to 
$\pc_{n,j} \stackrel{\alpha}{\to}_j \pc_{n+1,j}$ and $(d_n,d_{n+1}) \in \alpha$. 
Inspecting the instrumented edges,
we see that there is always an annotated state $(\bar\pc_{n+1}, d_{n+1})$ with
$\bar\pc_{n+1} = \bar\pc_n[j: = \pc_{n+1}^Y]$ for some $Y \in \set{R,L}$ and a path 
(of length one or two) $(\bar\pc_n,d_n) \to'^* (\bar\pc_{n+1},d_{n+1})$, which proves the 
inductive step.

Now, for the left-to-right direction of our lemma, we only need to note that if
$(\pc_n, d_n)$ is an error state in the original system, then $(\bar\pc_n, d_n)$ is 
an error state in the instrumented system. 

Similarly, for the opposite direction, we first note that if $(\bar\pc_o,d_0) 
\to'^* (\bar\pc_n,d_n)$ and $(\bar\pc_n,d_n)$ is an error state, then without loss 
of generality we can assume that for all $i$, $\bar\pc_{n,i} = l_i^X$ for
$X \in\set{R,L,N}$. The reason is that if this is not the case, i.e., if there is an 
intermediate location at some position $i$, then there is an internal edge (without 
an action) that eventually brought us in this location. The source of the 
corresponding internal transition is a state as required, and we can always extract 
a shorter path  with the required property. 

Now, we can prove by induction on $n$, the number of actions in $(\bar\pc_o,d_0) \to'^* (\bar\pc_n,d_n)$,
that there is a path $(\pc_0,d_0) \to^* (\pc_n,d_n)$ in the original system $C$ 
with $(\pc_n,d_n)$ being an error state. Actually, we prove that the statement 
holds for $\pc_n = \strip(\bar\pc_n)$ where $\strip(\bar\pc_n)_i = l_i$ iff
$\bar\pc_{n,i} = l_i^X$ for $X \in \{R,L,N\}$.

For $n=0$, the statement holds. Assume it holds for $n$ and consider
$(\bar\pc_o,d_0) \to'^* (\bar\pc_{n+1},d_{n+1})$ via $n+1$ actions, where the 
target state has no intermediate locations. Then, again without loss of generality, 
we can assume that $(\bar\pc_o,d_0) \to'^* (\bar\pc_{n},d_{n})$ via $n$ actions, 
and $(\bar\pc_{n},d_{n})$ has no intermediate locations, and 
$(\bar\pc_{n},d_{n}) \to'^* (\bar\pc_{n+1},d_{n+1})$ in at most two steps
(one optional internal/conditional step and one action). This last property is due to the 
bisimilarity in the proof of~\autoref{th:instrument}. If the internal/conditional 
step does not immediately precede the last action, then (1) if the conditional is 
enabled, then this step can move to the right, right before the last action; (2) if 
the conditional is disabled, then there an internal/conditional step for the same 
action that is enabled and leads to a bisimilar state (recall that bisimilarity 
preserves error states). By the inductive hypothesis, from
$(\bar\pc_o,d_0) \to'^* (\bar\pc_{n},d_{n})$ via $n$ actions, we get a path
$(\pc_0, d_0) \to^* (\pc_n,d_n)$ performing $n$ actions, and from 
$(\bar\pc_{n},d_{n}) \to'^* (\bar\pc_{n+1},d_{n+1})$ in (at most) two steps, and the shape of 
the instrumentation, we get a step $(\pc_n,d_n) \to (\pc_{n+1},d_{n+1})$ with
$\pc_{n+1} = \strip(\bar\pc_{n+1})$.
\end{proof}

\subsection{With Instrumentation, Left/Right Moving Affects Phases}
\label{app:dynmovers}

The following composition of left/right moving with in presence of
dynamic movers is exhaustive and refers to the rows in \autoref{tab:instrument}.

\noindent
Left mover $i$, when $j$ becomes dynamically right / left moving by R1,2 and R4,5:
\begin{equation*}
\begin{aligned}
\text{

\begin{tikzpicture}\centering

   \tikzstyle{e}=[minimum width=1cm]
   \tikzstyle{every node}=[font=\small, node distance=.5cm]

  \node (s1) {$\sigma_1$};
  \node (s2) [right of=s1, xshift=.5cm] {$\sigma_2$};
  \node (s3) [node distance=1.2cm,below of=s2,xshift=.5cm] {$\sigma_3$};
  \path (s2.south) -- node[midway,sloped,allow upside down]{$\tatrans_i$} (s3.north);
  \path (s1.west) -- node[pos=.45]{$\tatrans_j$} (s2.east);

  \node (s4) [node distance=1.2cm,gray,below of=s1,xshift=-.5cm] {$\sigma_4$};
  \node (s3p) [gray,node distance=1.2cm,below of=s2,xshift=-.5cm] {$\sigma_3'$};
  \path (s1.south) -- node[sloped,allow upside down,gray,midway]{$\tatrans_i$}
  (s4.north);
  \path (s4.west) -- node[gray,pos=.48]{$\tatrans_j$} (s3p.east);

  \node (S1) [left of=s1, xshift=-.3cm,e] {$\LL_i, {\nLL_j}\ni$};
  \node (S4) [left of=s4, gray,xshift=-.07cm,e] {$\RR_j\ni$};

  \node (S2) [right of=s2, xshift=.3cm,e] {$\in \LL_i, \LL_j$};
  \node (S3) [right of=s3, xshift=.07cm,e] {$\in \LL_j$};

  \node (S4) [gray,below of=s3p] {{$\in \RR_j$}};

  \path (s3p) -- node[gray,sloped,allow upside down,pos=.48]{$\cong_j$} (s3);
 
%
\end{tikzpicture}

\begin{tikzpicture}\centering

   \tikzstyle{e}=[minimum width=1cm]
   \tikzstyle{every node}=[font=\small, node distance=.5cm]

  \node (s1) {$\sigma_1$};
  \node (s2) [right of=s1, xshift=.5cm] {$\sigma_2$};
  \node (s3) [node distance=1.2cm,below of=s2,xshift=.5cm] {$\sigma_3$};
  \path (s2.south) -- node[midway,sloped,allow upside down]{$\tatrans_i$} (s3.north);
  \path (s1.west) -- node[pos=.45]{$\tatrans_j$} (s2.east);

  \node (s4) [node distance=1.2cm,gray,below of=s1,xshift=-.5cm] {$\sigma_4$};
  \path (s1.south) -- node[gray,midway,sloped,allow upside down]{$\tatrans_i$}
  (s4.north);
  \node (s3p) [gray,node distance=1.2cm,below of=s2,xshift=-.5cm] {$\sigma_3'$};
  \path (s4.west) -- node[gray,pos=.48]{$\tatrans_j$} (s3p.east);

  \node (S1) [left of=s1, xshift=-.3cm,e] {$\LL_i, \LL_j\ni$};
  \node (S4) [left of=s4, xshift=-.07cm,gray,e] {$\LL_j\ni$};

  \node (S2) [right of=s2, xshift=.3cm,e] {$\in \LL_i, \NN_j$};
  \node (S3) [right of=s3, xshift=.07cm,e] {$\in \NN_j$};
  \node (S4) [gray,below of=s3p, xshift=.07cm,e] {$\in \LL_j$};
  \path (s3p) -- node[sloped,allow upside down,pos=.48]{$\cong_j$} (s3);
\end{tikzpicture}

}
\end{aligned}
\end{equation*}

\noindent
Right mover $i$, when $j$ becomes dynamically non-right / non-left moving by R1,2 and R4,5:
\begin{equation*}
\begin{aligned}
\text{


\begin{tikzpicture}\centering

   \tikzstyle{e}=[minimum width=1cm]
   \tikzstyle{every node}=[font=\small, node distance=.5cm]

  \node (s1) {$\sigma_1$};
  \node (s2) [node distance=1.2cm,below of=s1] {$\sigma_2$};
  \node (s3) [right of=s2, xshift=.5cm] {$\sigma_3$};
  \path (s2.east) -- node[pos=.45]{$\tatrans_j$} (s3.west);
  \path (s1.south) -- node[midway,sloped]{$\tatrans_i$} (s2.north);

  \node (s4) [xshift=.9cm,gray,right of=s1] {$\sigma_4$};
  \node (s3p) [gray,node distance=1.2cm,below right of=s4,xshift=-.5cm] {$\sigma_3'$};
  \path (s1.east) -- node[gray,midway,sloped]{$\tatrans_j$}
  (s4.west);
  \path (s4.south) -- node[gray,midway,sloped]{$\tatrans_i$} (s3p);

  \node (S1) [left of=s1, xshift=-.07cm,e] {$\nLL_j\ni$};
  \node (S2) [below of=s2, xshift=.4cm,e] {$\in\RR_i, \RR_j$};

  \node (S4) [right of=s4, gray,xshift=.07cm,e] {$\in\LL_j$};
  \node (S3) [below right of=s3, xshift=.4cm,e] {$\in \RR_i,\RR_j$};

  \path (s3) -- node[gray,sloped,pos=.48]{$\cong_j$} (s3p);

 \node (S3) [gray,right of=s3p, xshift=.07cm,e] {$\in \LL_j$};  
%
%
%
%
\end{tikzpicture}

\begin{tikzpicture}\centering

   \tikzstyle{e}=[minimum width=1cm]
   \tikzstyle{every node}=[font=\small, node distance=.5cm]

  \node (s1) {$\sigma_1$};
  \node (s2) [node distance=1.2cm,below of=s1] {$\sigma_2$};
  \node (s3) [right of=s2, xshift=.5cm] {$\sigma_3$};
  \path (s2.east) -- node[pos=.45]{$\tatrans_j$} (s3.west);
  \path (s1.south) -- node[midway,sloped]{$\tatrans_i$} (s2.north);

  \node (s4) [xshift=.9cm,gray,right of=s1] {$\sigma_4$};
  \path (s1.east) -- node[gray,midway,sloped]{$\tatrans_j$}
  (s4.west);

  \node (s3p) [node distance=1.2cm,gray,below right of=s4,xshift=-.5cm] {$\sigma_3'$};
  \path (s4.south) -- node[gray,midway,sloped]{$\tatrans_i$} (s3p.north);

  \node (S1) [left of=s1, xshift=-.07cm,e] {$\LL_j\ni$};
  \node (S2) [left of=s2, xshift=-.3cm,e] {$\RR_i, \LL_j \ni$};

  \node (S4) [right of=s4, gray,xshift=.07cm,e] {$\in\NN_j$};
  \node (S3) [below of=s3, xshift=.3cm,e] {$\in \LL_j$};
  \node (S5) [gray,right of=s3p, xshift=.3cm,e] {$\in \RR_i,\NN_j$};
  \path (s3p) -- node[sloped,pos=.48]{$\cong_j$} (s3);
\end{tikzpicture}

}
\end{aligned}
\end{equation*}

\noindent
Right mover $i$, when $i$ remains dynamic right mover by R1,R2 and a left mover $i$
when $i$ itself loses dynamic left-movability (for the next transition) by R4,5:
\begin{equation*}
\begin{aligned}
\text{

\begin{tikzpicture}\centering

   \tikzstyle{e}=[minimum width=1cm]
   \tikzstyle{every node}=[font=\small, node distance=.5cm]

  \node (s1) {$\sigma_1$};
  \node (s2) [node distance=1.2cm,below of=s1] {$\sigma_2$};
  \node (s3) [right of=s2, xshift=.9cm] {$\sigma_3$};
  \path (s2.east) -- node[pos=.45]{$\tatrans_j$} (s3.west);
  \path (s1.south) -- node[midway,sloped]{$\tatrans_i$} (s2.north);

  \node (s4) [xshift=.9cm,gray,right of=s1] {$\sigma_4$};
  \node (s3p) [gray,node distance=1.2cm,below right of=s4] {$\sigma_3'$};
  \path (s1.east) -- node[gray,midway,sloped]{$\tatrans_j$}
  (s4.west);
  \path (s4.south) -- node[gray,midway,sloped]{$\tatrans_i$} (s3p);

  \node (S1) [left of=s1, xshift=-.07cm,e] {};
  \node (S2) [left of=s2, xshift=-.07cm,e] {$\RR_i \ni$};

  \node (S4) [right of=s3p, gray,xshift=.07cm,e] {$\in \RR_i$};
  \node (S3) [below left of=s3, xshift=.07cm,e] {$\in \RR_i$};

  \path (s3) -- node[gray,sloped,pos=.48]{=} (s3p);
%
%
%

\end{tikzpicture}

\begin{tikzpicture}\centering

   \tikzstyle{e}=[minimum width=1cm]
   \tikzstyle{every node}=[font=\small, node distance=.5cm]

  \node (s1) {$\sigma_1$};
  \node (s2) [right of=s1, xshift=.9cm] {$\sigma_2$};
  \node (s3) [node distance=1.2cm,below right of=s2] {$\sigma_3$};
  \path (s2.south) -- node[midway,sloped]{$\tatrans_i$} (s3.north);
  \path (s1.west) -- node[pos=.45]{$\tatrans_j$} (s2.east);

  \node (s4) [node distance=1.2cm,gray,below of=s1] {$\sigma_4$};
  \path (s1.south) -- node[gray,midway,sloped]{$\tatrans_i$}
  (s4.north);
  \node (s3p) [gray,node distance=1.2cm,below of=s2] {$\sigma_3'$};
  \path (s4.west) -- node[gray,pos=.48]{$\tatrans_j$} (s3p.east);

  \node (S1) [left of=s1, xshift=-.07cm,e] {$\LL_i \ni$};
  \node (S4) [left of=s4, xshift=-.07cm,gray,e] {$\NN_i\ni$};

  \node (S2) [right of=s2, xshift=.07cm,e] {$\in \LL_i$};
  \node (S3) [right of=s3, xshift=.07cm,e] {$\in \LL_i$};
  \node (S4) [gray,below of=s3p, xshift=.07cm,e] {$\in \NN_i$};
  \path (s3p) -- node[sloped,pos=.48]{$\cong_{i}$} (s3);

\end{tikzpicture}}
\end{aligned}
\end{equation*}

We conclude that the instrumentations allows other threads $j$ to lose their phase
(end up in a different location) when a transition of $i$ is moved to the right (over $j$)
in a global trace.
On the other hand, when $i$ moves a transition to the left over $j$, both $i$ and $j$ could
change their phase. This explains why \autoref{def:tas} defines the commutativity the way it does,
to wit:
$(\tatrans_i \rrestr \RR_{i}) \rcomm_j \tatrans_j$ and
$(\LL_i \lrestr \tatrans_i) \lcomm_{\set{i,j}} \tatrans_j$.
Note that we have equality for all threads $k\neq i,j$ (not participating in the move)
from the definition of $\cong_{i,j}$.


\end{document}